\documentclass[a4paper, 11pt, abstracton, titlepage, oneside]{scrartcl}
\makeatletter

\usepackage[T1]{fontenc}
\usepackage[utf8]{inputenc}
\usepackage[onehalfspacing]{setspace}
\usepackage[headsepline]{scrlayer-scrpage}
\clearscrheadfoot % reset standard format
\usepackage{layout}
\usepackage[english]{babel}
\usepackage[left=2.54cm,right=2.54cm,top=2.54cm,bottom=2.54cm]{geometry}

\usepackage{algorithm}
\usepackage[noend]{algpseudocode}
\usepackage{adjustbox}
\usepackage{eurosym}
\usepackage{lscape}
\usepackage{nicefrac}
\usepackage{array}
\usepackage{multicol}
\usepackage{paralist}
\usepackage{authblk}
\usepackage{graphicx}
\usepackage{booktabs}
\usepackage{placeins}
\usepackage{amsmath}
\usepackage{mathtools}
\usepackage{bm}
\usepackage{amssymb}
\usepackage{color}
\usepackage{changepage} 
\usepackage{supertabular}
\usepackage[acronym]{glossaries}
\usepackage{pgfplots}
\usepackage[normal]{threeparttable}
\usepackage{ulem}
\usepackage{natbib}
\usepackage[titletoc,title]{appendix}
\usepackage[hyphens]{url} % keep it here even if not alphabetically due to hyperref clashes!
\usepackage{ellipsis}
\usepackage{breqn} 
\let\counterwithin\relax

\usepackage{chngcntr}  
\usepackage{enumerate}
\usepackage{theorem}
\usepackage{etoolbox}
\usepackage{multirow}
\usepackage{colortbl}
\usepackage{subcaption}
\newcommand\fnote[1]{\captionsetup{font=small}\caption*{#1}}
\usepackage{tikz}
\usetikzlibrary{arrows.meta}
\usepackage{relsize}
\usepackage{stackengine}
\usepackage{wrapfig}
\usepackage{lscape}
\usepackage{epstopdf}
\usepackage{rotating}
\patchcmd{\endabstract}{\null}{}{}{} % replaces \null with third argument (empty)
\usepackage{babel}
\addto{\captionsenglish}{}

\usepgfplotslibrary{groupplots}
\usepgfplotslibrary{dateplot}
\usepackage[utf8]{inputenc}
\ifdefined\DeclareUnicodeCharacter
\DeclareUnicodeCharacter{2212}{-}
\fi

\usepackage{array}
\newcolumntype{x}[1]{>{\centering\arraybackslash\hspace{0pt}}p{#1}}

\usepackage{amsthm}
\newtheorem{prop}{Proposition}

\newtheorem{res}{Result}
\newtheorem{observation}{Observation} %[section]

\let\vec\mathbf

\ifdefined\mathleft
\else
\newcommand{\mathleft}{}
\fi	
\ifdefined\mathright
\else
\newcommand{\mathright}{}
\fi
\ifdefined\mathcenter
\else
\newcommand{\mathcenter}{}
\fi

\DeclareMathOperator*{\argmin}{arg\,min}

\makeatletter
\providecommand\phantomcaption{\caption@refstepcounter\@captype}
\makeatother

\newglossary{infoslist}{1i}{1o}{Information variants}
\makeglossaries
\newacronym{cs}{CS}{Charging Station}
\newacronym{ev}{EV}{electric vehicle}
\newacronym{sdp}{SDP}{stochastic dynamic programming}
\newacronym{mdp}{MDP}{Markov decision process}
\newacronym{pne}{PNE}{pure Nash equilibrium}
\newacronym{fcsag}{FCSA}{fleet charging station allocation}
\newacronym{wlo}{WLO}{white label operator}
\newacronym{md}{MD}{mechanism design}
\newacronym{ai}{AI}{artificial intelligence}

\newglossaryentry{self}{name={\textit{P-SELF}},
	description={Low-density scenario},
	type=infoslist}
\newglossaryentry{greed}{name={\textit{D-SELF}},
	description={Low-density scenario},
	type=infoslist}
\newglossaryentry{vcg}{name={\textit{VCG}},
	description={Low-density scenario},
	type=infoslist}

\newglossaryentry{all}{name={\textit{equal}},
	description={Low-density scenario},
	type=infoslist}
\newglossaryentry{big}{name={\textit{big}},
	description={Low-density scenario},
	type=infoslist}
\newglossaryentry{small}{name={\textit{small}},
	description={Low-density scenario},
	type=infoslist}

\definecolor{lightblue}{rgb}{0.60784,0.76078,0.90196}
\definecolor{darkblue}{rgb}{0.26667,0.44706,0.76863}
\definecolor{lightgreen}{rgb}{0.66275,0.81569,0.55686}
\definecolor{darkgreen}{rgb}{0.43922,0.67843,0.27843}
\definecolor{orange}{rgb}{0.92941,0.49020,0.19216}
\definecolor{yellow}{rgb}{1.00000,0.75294,0.00000}
\definecolor{grey}{rgb}{0.64706,0.64706,0.64706}
\definecolor{purple}{rgb}{0.51373,0.23529,0.04706}

\DeclareOldFontCommand{\rm}{\normalfont\rmfamily}{\mathrm}
\DeclareOldFontCommand{\sf}{\normalfont\sffamily}{\mathsf}
\DeclareOldFontCommand{\tt}{\normalfont\ttfamily}{\mathtt}
\DeclareOldFontCommand{\bf}{\normalfont\bfseries}{\mathbf}
\DeclareOldFontCommand{\it}{\normalfont\itshape}{\mathit}
\DeclareOldFontCommand{\sl}{\normalfont\slshape}{\@nomath\sl}
\DeclareOldFontCommand{\sc}{\normalfont\scshape}{\@nomath\sc}
\makeatother
\counterwithin*{equation}{section}

\bibpunct[, ]{(}{)}{,}{a}{}{,}%
\newcolumntype{L}[1]{>{\raggedright\let\newline\\\arraybackslash\hspace{0pt}}p{#1}}
\newcolumntype{C}[1]{>{\centering\let\newline\\\arraybackslash\hspace{0pt}}p{#1}}
\newcolumntype{R}[1]{>{\raggedleft\let\newline\\\arraybackslash\hspace{0pt}}p{#1}}
\addtokomafont{captionlabel}{\footnotesize\bfseries} % style for caption labels
\addtokomafont{caption}{\footnotesize\bfseries} % and the caption

\AtBeginDocument{%
	\abovedisplayskip=2.5pt plus 0pt minus 0pt
	\abovedisplayshortskip=2.5pt plus 0pt minus 0pt
	\belowdisplayskip=2.5pt plus 0pt minus 0pt
	\belowdisplayshortskip=2.5pt plus 0pt minus 0pt
}

\newsavebox{\abstractbox}
\renewenvironment{abstract}
{\begin{lrbox}{0}\begin{minipage}{\textwidth}
			\begin{center}\normalfont\sectfont\abstractname\end{center}\quotation}
		{\endquotation\end{minipage}\end{lrbox}%
	\global\setbox\abstractbox=\box0 }

\makeatletter
\expandafter\patchcmd\csname\string\maketitle\endcsname
{\vskip\z@\@plus3fill}
{\vskip\z@\@plus2fill\box\abstractbox\vskip\z@\@plus1fill}
{}{}
\makeatother

\makeatletter
\def\BState{\State\hskip-\ALG@thistlm}
\makeatother

\counterwithin*{equation}{section}

\newenvironment{myfont}{\fontfamily{ccr}\selectfont}{\par}
\DeclareTextFontCommand{\textmyfont}{\myfont}

\begin{document}
% first put your acronyms and abbrevs.

\definecolor{lightblue}{rgb}{0.60784,0.76078,0.90196}
\definecolor{darkblue}{rgb}{0.26667,0.44706,0.76863}
\definecolor{lightgreen}{rgb}{0.66275,0.81569,0.55686}
\definecolor{darkgreen}{rgb}{0.43922,0.67843,0.27843}
\definecolor{orange}{rgb}{0.92941,0.49020,0.19216}
\definecolor{yellow}{rgb}{1.00000,0.75294,0.00000}
\definecolor{grey}{rgb}{0.64706,0.64706,0.64706}
\definecolor{purple}{rgb}{0.51373,0.23529,0.04706}

\newcommand{\orange}[1]{\textcolor{orange}{#1}}

\newacronym{cs}{CS}{charging station}
\newacronym{dp}{DP}{dynamic programming}
\newacronym{abk:eu}{EU}{European Union}
\newacronym{abk:ecv}{ECV}{electric commercial vehicle}
\newacronym{ev}{EV}{electric vehicle}
\newacronym{abk:ghg}{GHG}{greenhouse gas}
\newacronym{abk:icev}{ICEV}{internal combustion engine vehicles}
\newacronym{mdp}{MDP}{Markov decision process}
\newacronym{abk:rcspp}{RCSPP}{resource constraint shortest path problem}
\newacronym{abk:ref}{REF}{resource extension function}
\newacronym{sdp}{SDP}{stochastic dynamic programming}
\newacronym{Soc}{SoC}{state of charge}
\newacronym{scps}{SCPS}{stochastic charge pole search}

% then the title
%----------------------------------------------------------------------------------
% titlepage template
% author: Maximilian Schiffer
% version: 1.0 / 11.05.2017
%----------------------------------------------------------------------------------

% place the title here
%\title{\large On the Benefits of Information-Sharing for Charging Station Search in Stochastic Environments:\\ a Multi-Agent Approach}
\title{Coordinating charging request allocation between self-interested navigation service platforms}

%\author{}
% and the authors
\author[1,2]{\normalsize Marianne Guillet}
\author[1,3]{\normalsize Maximilian Schiffer}
\affil{\small 
	TUM School of Management, Technical University of Munich, 80333 Munich, Germany
	
	\textsuperscript{2} TomTom Location Technology Germany GmbH, 12435 Berlin, Germany
	
%	\textsuperscript{3} Chair of Operations Management, RWTH Aachen University, 52072 Aachen, Germany
	
	\textsuperscript{3} Munich Data Science Institute, Technical University of Munich, 80333 Munich, Germany,
					
	\scriptsize
	marianne.guillet@tomtom.com,
	schiffer@tum.de}

% if you like - a date
\date{}

% in case you have a headline - otherwise outcomment
\lehead{\pagemark}
%\rehead{\normalfont\scriptsize\textbf{Schiffer et. al.:} \textit{Perspectives for electric commercial vehicles}}
%\lohead{\normalfont\scriptsize\textbf{Schiffer et. al.:} \textit{Perspectives for electric commercial vehicles}}
\rohead{\pagemark}

% finally your abstract
\begin{abstract}
	Current electric vehicle market trends indicate an increasing adoption rate across several countries. To meet the expected growing charging demand, it is necessary to scale up the current charging infrastructure and to mitigate current reliability deficiencies, e.g., due to broken connectors or misreported charging station availability status. However, even within a properly dimensioned charging infrastructure, a risk for local bottlenecks remains if several drivers cannot coordinate their charging station visit decisions. 
	%head towards the same station. 
	Here, navigation service platforms can optimally balance charging demand over available stations to reduce possible station visit conflicts and increase user satisfaction. While such fleet-optimized charging station visit recommendations may alleviate local bottlenecks, they can also harm the system if self-interested navigation service platforms seek to maximize their own customers' satisfaction. 
	To study these dynamics, we model fleet-optimized charging station allocation as a resource allocation game in which navigation platforms constitute players and assign potentially free charging stations to drivers. We show that no pure Nash equilibrium guarantee exists for this game, which motivates us to study VCG mechanisms both in offline and online settings, to coordinate players' strategies toward a better social outcome. 
	%
	%We further account for weighted players in both off- and online settings.
	%We extend the mechanisms in offline and online settings to account for weighted players and further analyze optimal players' weights in offline cases. 
	%
	Extensive numerical studies for the city of Berlin show that when coordinating players through VCG mechanisms, the social cost decreases on average by 42\% in the online setting and by 52\% in the offline setting. 
\begin{singlespace}
{\small\noindent\\
\smallskip}
\end{singlespace}
\end{abstract}

% don't forget to make the tile
\maketitle	

% and your chapters
%\title{Competing charging stations allocation with deterministic availability}
%\maketitle

\section{Introduction}

% Intro on evs that take off, need to collectively organize a better charging processs at public infrastructurewhile preserving truthfulness and further a

%overall adoption seems to have taken off, but is heterogeneous both country wise and at regions scale.

Battery-powered vehicles support the shift towards more sustainable mobility systems, especially if coupled with the overall increased use of public transportation systems.
%
%While the ev adoption 
After a slow start for EV adoption, the private electric vehicles market recently showed an encouraging growth, with global adoption doubling in 2020 and 2021, and car manufacturers planning to launch more than 100 new electric vehicle models by 2024.
However, studies reveal strong adoption heterogeneity, at national and regional scales. In 2020, Norway had a (BH)EV market share of more than 70\% \citep{cleantechnica2021}, whereas Romania's EV market share was less than 10\% \citep{AVER2021}. In the latter case, the adoption is thus far hindered by poor or missing charging infrastructure.
%Deciding whether to increase 
In general, charging-related anxieties caused by limited infrastructure availability dissuade conventional vehicle drivers from switching to electric vehicles.
%\\
Drivers may experience so-called \textit{range anxiety}, i.e., the fear of running out of battery, or \textit{charge anxiety}, i.e., the fear of not finding any available and non-broken public charging station.
% by triggering the so-called range and charge anxieties (see X).
While building new charging stations may alleviate these anxieties, operating over-dimensioned charging infrastructure in low charging demand areas can be highly cost-inefficient \citep{NelderRogers2019}.
%highly cost-inefficient.
%
Overcoming this chicken-and-egg dilemma, i.e., deciding whether to pursue EV adoption or charging infrastructure first, can be mitigated with appropriate infrastructure planning and operation \cite[cf.][]{Enlit2022}.
%https://www.enlit.world/low-carbon-transportation/electric-vehicles/overcoming-the-e-mobility-chicken-or-egg-dilemma/

From a short-term perspective, drivers' anxieties related to charging in an undersized infrastructure can be addressed via advanced operational planning, e.g., a navigation system that reliably guides EV drivers until they find an available station, anticipating possibly blocked, unreachable, or out-of-order stations.  
%Both practical and theoretical work address these issues at the operational level, 
%
To foster EV adoption from a long-term perspective, regulatory measures, e.g., a ban on conventional vehicles in city centers or financial incentivization are necessary, but must be complemented by planning solutions.
Such solutions include developing the charging infrastructure with new equipment, consolidating existing urban infrastructure (e.g., street lamp charging), or increasing the reliability of real-time availability information through detection sensors or cameras.

The expected growing charging demand can nevertheless create new bottlenecks, i.e., charging station congestion, if demand is misaligned with charging infrastructure capacities.
Even with an increased charging infrastructure coverage, several EV drivers may receive identical charging station recommendations from navigation services platforms, which creates station utilization conflicts due to long recharging times.
Here, fleet-optimized recommendations -- enabled by requests centralization -- can better distribute the charging demand over the available charging stations and increase the overall users' satisfaction \citep[cf.][]{GuilletSchiffer2022}.
%From a navigation service platform's perspective, providing fleet-optimized recommendations by centralizing user charging requests may better balance the overall charging demand over available stations, and accordingly satisfy more EV drivers.

Such a fleet-optimization may work very well in presence of a single platform but finds its limits with the presence of several self-interested navigation platforms, each seeking to maximize their own drivers' utilities. In this case, the overall outcome of the drivers' assignment to stations may be far from optimal due to conflicting charging station allocations in between platforms.
%
%If several self-interested platforms seek to maximize their own users' utilities, 
%the outcome of the overall drivers \& charging stations assignment  may become socially worst, in case of charging station allocations conflict.
%multiple-fleet optimized assignments become worst than a single system optimized assignment.
%
Accordingly, solely improving the infrastructure with no further charging demand coordination at the system level may worsen the EV charging experience by decreasing the likelihood of finding an available station within a minimum amount of time.
If a system composed of multiple navigation services platforms cannot reach any good equilibrium, a centralized non-profit-driven authority, e.g., a municipality, may aim to improve the system towards a socially efficient outcome, using financial incentives or penalties.

To understand the dynamics of such a setting, we introduce the \gls{fcsag} game in which several self-interested navigation platforms aim to individually optimize the assignment of available stations to their EV drivers, i.e., their charging demand.
%, and are aware of the competitor's charging demand.
%
In this context, the scope of this paper is twofold. First, we use a game-theoretical framework to analyze the dynamics and bottlenecks that may arise in this setting due to non-cooperative driver assignments. 	
We show that no guarantee for a \gls{pne} exists, which motivates our second scope: deriving mechanisms to stabilize and coordinate the charging station allocation, studying the impact of aligning platforms' interests from a platform but also from a driver perspective.
%then and ii) analyze the system benefits of aligning platforms' interests through market-based techniques.
%study the impact of possibly conflicting charging station assignment, realized by self-interested navigation service platforms, 
%
%When centralizing requests, such platforms may satisfy more users by providing fleet-optimized stations recommendations, to better balance their existing charging demand over all available stations, and avoid possible utilization conflicts. 

%%
%There exists however more than one navigation service provider that aims at providing fleet-optimized recommendations to its guided vehicles. Whereas this might be efficient at fleet-level, this might be inefficient at system-level, i.e., when considering all drivers getting recommendations from all platforms. 

	\subsection{Related Work}
	
%	Charging station can be modeled as indivisible resources 
%	 that must be allocated to navigation service platforms, that then assign these stations to their drivers.
	Our work relates to indivisible resource allocation problems for \glspl{ev} and other related applications, e.g., parking utilization, which we review in the following. We first discuss resource allocation games, particularly with a focus on equilibrium analyses, before we review works that apply mechanism design to resource allocation problems.

%	Modeling charging stations as indivisible resources, we review indivisible resource allocation problems in game-theoretic settings, focused on EV charging or related topics (e.g., Parking utilization, etc.). We first discuss relevant games and equilibrium analysis, before reviewing settings that utilize mechanism design. 
	
	Congestion games, initially introduced in \cite{Rosenthal1973}, model non-cooperative atomic resource allocation between homogeneous agents. Here, a player's strategy is defined as a subset of the available resources, whose utilization cost depends on the total number of players selecting the resource. 
	This class of games naturally applies to load-balancing \citep[see, e.g., ][]{EvenDarKesselmanEtAl2003,GoemansMirrokniEtAl2006}, network design \citep[see, e.g., ][]{AnshelevichDasguptaEtAl2008}, or internet routing \citep[see, e.g., ][]{SecciRougierEtAl2011}.
%	This class of games naturally applies to routing games,, in physical road network with roads as resources and players as drivers \citep[cp.][]{RoughgardenTardos2002}, or virtual networks such as Internet. % TODO REF
%%	
	Congestion games fall into the category of potential games \citep[cf.][]{MondererShapley1996}, i.e., games that admit a potential function, which describes the payoff variation between two different strategies independent of the players. Such games possess a guaranteed \gls{pne} existences.
	\cite{Milchtaich1996}'s work extends \cite{Rosenthal1973}'s work by showing the existence of a \gls{pne} in congestion games with player-specific payoff functions or weighted players, but restricted singleton strategies.
	\cite{AckermannRoeglinEtAl2009} later show that strategies that consist of the bases of a matroid defined on the resources set still guarantee the \gls{pne}'existence for games analyzed in \cite{Milchtaich1996}.
%	
% todo : Here, some recent transporotation application of congeston games

	Most EV-related games deal with charging capacity allocation in grid networks, focused on optimizing grid operators' revenues \citep[][]{TusharSaadEtAl2012}, reducing energy peak load through load and capacity-dependent energy prices \citep[][]{SheikhimBahrami2013}, or mitigating the impact of charging vehicles on transportation networks \citep[][]{SohetHayelEtAl2021}.
%	
%	Focusing on a different but related application, 
	\cite{AyalaWolfsonEtAl2011} analyze related competitive parking slot assignment problems, and show the existence of a \gls{pne} when each player, i.e., driver, must select at most one parking spot for which a successful utilization depends on her and other players' driving distances to all parking spots.
	\cite{HeYinEtAl2015} further extend the result of \cite{AyalaWolfsonEtAl2011} by analytically characterizing equilibrium parking slot assignments.
%	 and discussing optimal pricing scheme to achieve 
%
%	Ayla b and c further discuss oiptimal pricinig to 

	While the \gls{fcsag} game studied in this paper resembles a congestion game, its payoff functions are user-specific and depend on the number of other players selecting an identical resource only under some conditions: while all drivers competing for the same resource affect each other's payoff in a classical congestion game, in the \gls{fcsag} game only drivers that are closer to utilizing the resource affect the payoff function of drivers that are further away from utilizing the resource.
	Moreover, congestion games have not yet been studied in the context of charging station allocation via intermediate self-interested participants.

	To deal with inefficient equilibria, \gls{md} theory aims to design games in such a way that a socially desirable outcome is reached \citep{NisanRoughgardenEtAl2007}.
	A third party, a so-called principal, centrally allocates goods or resources to its participants, who must reveal their preferences on the resources before paying a price for the received resource.
	In this context, the well-known VCG mechanism \citep{Vickrey1961,Clarke1971,Groves1973} defines a pricing scheme that ensures revelation truthfulness by aligning participants and the overall system interests in an offline setting.
%	
%\citep[see also][]{ParkesWellman2015}
	As an offline allocation may be hard to realize in practice, \cite{Parkes2003} propose two online VCG mechanisms, namely delayed and online VCG, that utilize a \gls{mdp} to derive the minimal expected cost allocation, while ensuring that expected participant utilities align with the expected system interest. 
	They show the Bayesian-Nash truthfulness of these mechanisms and apply them successfully to WiFi pricing \citep{FriedmanParkes2003}. 
	As the pricing scheme relies on an optimal policy argument, the authors later discuss approximately efficient online \gls{md} using $\epsilon$-efficient policies. More recently, \cite{SteinOchalEtAl2020} proposed a reinforcement learning-based mechanism that guarantees strategy-proofness when the resource allocation problem is solved online. 
	Focusing on EV charging, \cite{RigasGerdingEtAl2020} apply standard VCG pricing schemes to achieve efficient charge scheduling of self-interested EV drivers by minimizing the charging demand imbalance at a station.
%	, building on \cite{RigasRamchurn2013} or related to \cite{FosterCaramanis2013}.
	%	
	\cite{GerdingRobuEtAL2011} design a mechanism to solve the online charge scheduling problem, blind to future charging demand.
%	based on a model-free online \gls{md} approach, i.e., that does not account for future charging demand.
%	In all cases, the self-interested participants constitute of homogeneous EV drivers that aim to minimize their charging costs. 
	Several works have studied mechanisms in the related context of parking slot allocation.
	\cite{XuChengEtAl2016} derive a trading mechanism to share private parking slots during office hours in big cities between self-interested owners, to remedy limited public parking spot availability.
	\cite{ZouKafleEtAl2015} extend the parking slot assignment game of \cite{AyalaWolfsonEtAl2011} by applying VCG mechanisms to a publicly-owned parking facility that aims to maximize the social welfare, both in static and dynamic settings. 
	\cite{WangGuanEtAl2014} derive an optimal allocation pricing scheme based on a Demange-Gale-Sotomayor-based mechanism, for reservable and non-reservable parking resources in a city.
	% Focusing on charging spot allocation to self-interested drivers, many work have studied how mechanism design can improve the allocationi efficiency.
	
	In summary, most EV-related settings focus on charge scheduling problems with homogeneous self-interested EV drivers, while this work focuses on navigation platforms as self-interested participants. Similarly, work on parking slot allocation problems focus on directly incentivizing drivers to truthfully report their valuations, but not on intermediaries that could balance aggregated parking demand over available slots.

	\subsection{Aims and Scope}
	
	To close the research gaps outlined above, we introduce a game-theoretical setting to study the dynamics of self-interested navigation platforms that aim to best balance their own clients' charging requests among available charging stations, focusing first on equilibrium analysis, and second on mechanism design.
	We formalize the resulting charging station allocation problem as a novel game, and apply both offline and online VCG mechanisms to ensure socially desirable outcomes in idealized but also in practical settings.
	Specifically, our contribution is three-fold.
	First, we define the \gls{fcsag} game in a perfect-information setting as a new resource allocation game, and show that this game does not possess a guaranteed \gls{pne}, but also no approximated \gls{pne} with a sufficiently small approximation factor.
	Second, we coordinate players by applying VCG pricing in both an offline and an online setting, accounting for weighted players. We extend the delayed VCG mechanism to a weighted delayed VCG mechanism, and implement a data-driven online assignment policy.
%	]problem using the scenario-optimization framework.
%	
	Finally, we conduct extensive simulation-based experiments to analyze the benefits of our coordinated allocation mechanism.
	Our results show that players' coordination via the VCG mechanism can decrease the social cost on average by 52\% in an offline setting and by 42\% in an online setting when using our data-driven assignment policy.
	Our results further show that a player with a small share of drivers has a greater interest in participating in the VCG mechanism than a player with a larger share of drivers. A player's payoff relative to its number of drivers decreases when its share of managed drivers increases.
	
%	that optimally defined players' weights may increase the interest of players in participating in the mechanism.
	
%	We first conduct a simplified analysis of the situation by formalizing the problem as an offline game where multiple navigation platforms (e.g., TomTom, Google, etc.), acting as players, need to allocate available charging stations to their EV drivers. Real-time charging station availability is deterministic and known to the platforms. We observe that there exist no guarantee of stable situations, even in this simplified setting.
%	%	
%	We then analyze how centralized assignment recommendations provided to the platforms, and then cascaded to the end users, may improve the overall system performances. We first implement the well-known VCG mechanism in an offline setting, (that incentivizes platforms to truthfully  communicate their drivers' locations, using some adapted payments). In practice this could be realized by an inter-charge operator, or a centralized non-profit authority, that may guarantee the availability of recommended stations, via stations booking.
%	%
%	We then discuss the extension of the setting to a weighted VCG mechanism, and discuss weights such that participants (i.e., platforms) have an interest in obtaining centralized assignments recommendations. Finally, we discuss an online setting for initial practical considerations, and implement the delayed VCG mechanism introduced by Parkes, that we extend with weighted players

	\subsection{Organization}
	
	The remainder of this paper is structured as follows. Section~\ref{sec:problem_setting} details our \gls{fcsag} game setting and corresponding equilibrium analyses. In Section~\ref{sec:pricing}, we introduce the VCG mechanisms, which we apply in both offline and online settings. Section~\ref{sec:experimental_setting} describes our experimental design, while we discuss our numerical results in Section~\ref{sec:results}. Section~\ref{sec:conclusion} concludes this paper and provides an outlook for future research. For the sake of conciseness, we shift all proofs to the Appendix.
%	\FloatBarrier

	\section{Problem Setting}\label{sec:problem_setting}
 	We focus on an non-cooperative offline resource allocation problem, where several navigation platforms aim to optimally assign their EV drivers to available charging stations, such that no visit conflict arises for their drivers and the total travel time for drivers to assigned stations is minimal. 
	In the following, we consider a perfect-information setting, in which each platform is aware of overall charging demand, and of accurately reported real-time charging station availabilities.
	Note that such a perfect-information setting is unrealistic in practice but allows us to analyze whether stable players' strategies exist at least in an idealized setting. 
%	We note that we later implement both an onlin
%todo : put an important sentence on why it is imporotant to study the dynamics in a perfet-information setting first
% 	
% 	We assume that the navigation platforms are aware of accurately reported real-time charging station availabilities.
%	
%	The assignment of drivers to stations is partially centralized through navigation services platforms that seek to best distribute their own charging demand over existing available charging stations.
%	
%	In the following, we consider a perfect-information setting, such that each platforms is aware of the charging requests that other platforms manage.
%	
	Each platform optimally assigns free stations to its navigated EV drivers, who accordingly drive to the stations to recharge their vehicles.
%	We assume a perfect information setting, such that platforms are aware of the overall charging demand. 
	To focus on the dynamics between the navigation platforms, we assume that electric drivers do not deviate from their assigned stations.
	In case of a conflicting assignment, i.e., if two or more platforms assign one of their drivers to the same station, the station availability is guaranteed for the driver with the earliest arrival time. The other assigned drivers fail their search, which induces a penalty cost for the respective platform.
%	, in order to represent the driver's discomfort.
	
%	Furthermore, electric car drivers are assumed to execute the assignment recommendations of their navigation platforms. 
	In the following, we formalize our problem setting as a resource allocation game, show that a \gls{pne} guarantee does not exist, and discuss the limits of refinement equilibrium concepts within this context.
	
	\subsection{Fleet Charging Station Allocation game}\label{subsec:game}

	We now define the \gls{fcsag} game, in which navigation service platforms constitute the players. A player's strategy is an assignment of charging stations to its drivers, while its payoff describes the cost of the assignment, corresponding to the sum of the time needed by each driver to travel to its assigned station. We assume that if multiple drivers are assigned to the same station, all non-closest drivers must pay an extra penalty cost that represents the discomfort for failing to find a free charging station.
	We further restrict the set of reachable stations for a driver by setting a maximal driving time, and assume a deterministic driver behavior as reasoned in the section's introduction.
%	bound the driving time from a driver's initial location to a visited station.
%	may only visit stations within a given distance from their initial location.
	
%	that an EV driver deterministically drives to the station assigned by its corresponding platform, since we constrain the charging stations assignment with realistic spatio-temporal constraints. 

%	
	Formally, we consider a set of players  $\mathcal{N}$, a set of drivers $\mathcal{D}$, and a set of available charging stations $v \in \mathcal{V}$. We denote with $\mathcal{V}' = \mathcal{V} \bigcup \{v_0\}$ the set of charging stations extended with an artificial station $v_0$, which represents a non-feasible assignment of a driver to any physical station. Such non-feasible assignments may occur when the number of drivers exceeds the number of available stations.
	 We let $t_{k,v}$ be the driving distance for  driver $k$ to charging station $v$.
%	 	 and with $\bar{S}^k$, the maximal search radius for drivers that belong to player $i\in \mathcal{N}$.
%	
	Each player must serve a charging demand, that corresponds to a set of driver locations, denoted by $\mathcal{D}_i = \{k_i \in \mathcal{D} \} \;,\forall i \in \mathcal{N}$.
%	where $k_i$ represents a driver that belongs to player $i$.
%
	A player's pure strategy $s_i$ is an allocation function of charging stations to drivers, with $s_i : \mathcal{D} \mapsto \mathcal{V}'$.
	We restrict the stations that a driver $k$ may visit to the stations reachable within $\bar{S}$ meters, i.e., $s_i(k) \in \{v : v \in \mathcal{V}', \gamma(k,v) \leq \bar{S}\}$, with $\gamma(k,v)$ being the distance from driver $k$ to station $v$.
	A player can assign a physical station at most once to its drivers, but can assign an artificial vertex $v_0$ as many times as needed, formally
	 \begin{equation}
	 \forall k,k' \in  \mathcal{D}_i, \;  (s_i(k) = s_i(k')) \land  (s_i(k) \neq v_0 ) \Rightarrow  (k = k')\; .
	 \end{equation} 
	 A player must further assign a real station to each of its drivers if possible, i.e., at most $|\mathcal{D}_i| - |\mathcal{V}|$ drivers can be assigned to the artificial station $v_0$.
	We let the strategy profile $s=(s_i)_{i\in \mathcal{N}}$ be the vector of strategies of all players, $s_{-i}$ be the vector of strategies for all players but $i$, and $S_i$ be the strategy space for player $i$.
	We denote with $u_i(s)$ the payoff for player $i$ when strategy profile $s$ is played, which corresponds to the assignment cost of $s_i$, resulting from the sum of its individual driver's payoff functions $c_i$, given the other competing players' strategies $s_{-i}$. Formally, we define $u_i(s)$ as
	\begin{equation}\label{eq:u_payoff}
	u_i(s)  = \sum_{k \in \mathcal{D}_i} c_i(k,s)\;,
	\end{equation}
	with 
	\begin{equation}\label{eq:cost_I}
	c_i(k,s) = 
	\begin{cases}
	\bar{\beta}, & \text{if } s_i(k) = v_0   \\
	t_{k,s_i(k)} , & \text{if } \forall k' \in \mathcal{D}_j \; \; \forall j \in \mathcal{N}, \;j\neq i \;  (s_j(k') = s_i(k)) \Rightarrow (t_{k',v} \geq t_{k,v})    \\
	t_{k,s_i(k)}  +\bar{\beta}& \text{otherwise } \\
	\end{cases}
	\end{equation}
	
	An individual driver $k$'s payoff $c_i(k,s)$ ensures that, in case of conflicting driver to station assignments, all but the earliest arriving drivers assigned to a station $v$ must pay a penalty cost $\bar{\beta}$.
	The penalty $\bar{\beta}$ represents the discomfort for failing the charging station search.
	Then, the goal of each player $i$ is to find a strategy $s_i$ that minimizes $u_i(s_i, s_{-i})$, given that the assignment cost depends on other players' selected strategies.

%	We can model each station $c$ to have its own driver's preferences, ranked by driving time 
%
	We consider the pure-strategy profile  $s^*$ to be a \gls{pne} when it holds that
	\begin{equation} 
		\forall i\in \mathcal{N}, s^*_{i} \in \argmin_{s_i \in S_i} u_{i} (s_i, s^*_{-i})\; ,
	\end{equation}
	i.e., no player can unilaterally decrease its payoff by changing its strategy.  
%	  
%	Both the existence and the uniqueness of a one pure-strategy equilibrium is not guaranteed.
%	\noindent 
	The sum of all players' payoffs defines the social cost. 
	
	From a system-perspective, i.e., in an idealized setting with collaborative players, the goal is to find a strategy profile that minimizes the social cost. We refer to the minimum total cost as the social optimum, defined as
%	We define the social optimum as the minimal sum over all players' payoffs as follows
	\begin{equation}\label{eq:soc_cost}
		S_g =  \min_{s \in \mathcal{S}} \sum_{i \in \mathcal{N}} \sum_{k \in \mathcal{D}_i} c_i(k,s)\;.	
	\end{equation}
	With $A_v$ being the assignment of drivers to station $v$, $s(k)$ being the station assigned to driver $k$, and strategy profile $s$, we can reformulate the social optimum as
	\begin{equation}\label{eq:soc_cost1}
	\begin{split}
		S_g&= \min_{s \in \mathcal{S}} \sum_{v \in \mathcal{V}} \sum_{k \in A_v} t_{k,s(k)} + (|A_v| -1 )\cdot \bar{\beta}\;.
	\end{split}
	\end{equation}

	\subsection{Equilibrium analysis}
	
	To analyze possible stable outcomes of the \gls{fcsag} game, we discuss the \gls{fcsag} game's equilibrium properties in the following.
	Here, we note that a Nash equilibrium corresponds to a strategy profile such that no player has an incentive to unilaterally deviate from its current strategy.
	Moreover, there always exists a Nash equilibrium for a strategy profile of mixed strategies, i.e., a profile that corresponds to a probability distribution over pure strategies.
	However, mixed strategies do not reflect the behavior of a player in practice, as a player assigns a set of physical stations to its drivers once, which makes the interpretation of mixed strategy Nash Equilibria difficult for real-world analyses.
	In contrast, a \gls{pne} allows for more interpretable outcomes but is not guaranteed to exist in the \gls{fcsag} game. 
%	
%	While there always exists a mixed-strategy Nash equilibrium, this is not a useful concept in practice to predict the outcome of our allocation game. Thus, a desirable property of a game is the existence of a \gls{pne}, as it reflects best agents' behavior in practice and allows to analyze the efficiency of a non-cooperative system, e.g., using pure prizes of Anarchy or Stability. Proposition~\ref{prop:noPNE} shows that the PNE existence in the previously defined game cannot be guaranteed.
	
	\begin{prop}\label{prop:noPNE}
		The \gls{fcsag} game does not always possess a \gls{pne}. 
	\end{prop}

	\begin{proof}
	We refer to Appendix~\ref{app:proofs} for a proof of Proposition \ref{prop:noPNE}.
	\end{proof}

%	\begin{proof}
%	We show the non-guaranteed PNE existence by finding a game instance that does not possess any. 
%%	 
%	Figure~\ref{fig:game} visualizes such game instance. The instance comprises 3 stations $a$, $b$, and $c$ for 2 players $p_1$ and $p_2$: the 1st player has 2 drivers $d_1$ and $d_2$, the 2nd player has only one driver $e_1$.
%	\begin{figure}
%		\caption{G1' instance with no PNE}
%		\label{fig:game}
%		\input{chapters/assets_tikz/G1_nonash}
%	\end{figure}
%	None of the possible game outcomes corresponds to a Nash equilibrium.
%%	
%	If each player individually optimizes the assignment of its drivers to the available station, we obtain the following profile $((b,c),(c))$, with $c$ conflicting. From there, $p_2$ should re-assign its unique driver to $b$, which is then conflicting. In this case, $p_1$ is better off reassigning $d_1$ to $c$ and $d_2$ to $a$, as the total cost (17) is lower than for the strategy $(a,c)$ (18). However, $p_1$ should now re-assign its unique driver to $c$ such that $p_1$'s best response is again its individually optimized assignment solution. The last strategy profile now equals the initial strategy profile. 
%	\end{proof}

%	\subsection{equilibrium refinement : $\rho$-approximate equilibrium}
	
	For games with non-guaranteed \gls{pne} existence, $\rho$-equilibria with pure strategies constitute an interesting alternative to interpret the outcome of a game.
%	
%	 As there is no guarantee to find a PNE, one can still wonder whether a $\rho$-equilibrium can be reached with pure strategies. 
	Specifically, a $\rho$-equilibrium corresponds to a near-stable state in which each player cannot decrease its payoff by more than the absolute factor $\rho\geq0$. In practice, such an equilibrium corresponds to each player accepting to slightly deviate from the best solution it may obtain if its deviation stabilizes the game outcome.
	Formally, this corresponds to a profile $s$, such that 
	\begin{equation}
		\forall i \in \mathcal{N}, \forall s'_i \in \mathcal{S}_i\;\;, u_i(s'_i, s_{-i}) > u_i(s) - \rho\;.
	\end{equation}
	However, we can assume that such equilibrium refinement is only of practical interest if a player's payoff marginally increases when the player deviates from its best response, i.e., if $\rho$ is sufficiently small.
%	if the acceptable deviation by a player is sufficiently small,  
	In Proposition~\ref{prop:approxNE}, we show that the existence of an equilibrium with $\rho$ being sufficiently small cannot be guaranteed, as we can construct instances in which for any set of strategy profiles at least one player can significantly decrease its payoff by deviating from its current strategy.
%	as we cannot guarantee an efficient $\rho$-\gls{pne}, i.e., with $\rho$ sufficiently small.
%	
	\begin{prop}\label{prop:approxNE}
		The existence of a $\rho$-\gls{pne} with $\rho <  \bar{\beta} - \Delta_t$, with $\Delta_t= t_{\text{max}} -  t_{\text{min} }$ being the difference between the largest and the lowest driving times between a driver and a station, cannot be guaranteed for the \gls{fcsag} game.
	\end{prop}
	
	\begin{proof}
		We refer to Appendix~\ref{app:proofs} for a proof of Proposition \ref{prop:approxNE}.
	\end{proof}
	\section{Charging station allocation mechanism}\label{sec:pricing}

In Section~\ref{sec:problem_setting}, we show that no pure \gls{pne} guarantee exists for the \gls{fcsag} game, which motivates our interest to better coordinate the different players via \gls{md} to steer the system towards more socially desirable outcomes.
To realize such a mechanism, we consider the existence of a \gls{wlo} who decides on the system-optimal allocation of stations to drivers across several players, considering station preferences reported by all players. 
%	
%In practice, the \gls{wlo} might be an inter-charge operator or a municipality that allocates with guaranteed availability, e.g., via a charging slot booking system a subset of stations to the navigation platforms, who then assign the stations to their drivers.
%
In practice, the \gls{wlo} might be an inter-charge operator or a municipality that allocates, e.g., via a charging slot booking system, a subset of stations to the navigation platforms, who then assign the stations to their drivers.	

Formally, we consider $m=|\mathcal{V}|$ charging stations  constituting $m$ indivisible resources to be concurrently allocated among $|\mathcal{N}|$ players. Each player has a cost valuation for each bundle of stations, that corresponds to the minimum cost assignment between its drivers and the bundle's stations.
Valuations are non-additive, e.g., if a bundle contains more stations than drivers, removing the non-assigned station from the bundle preserves the valuation. While the valuation is expressed in units of time in the following, it can be transposed to monetary units in practice. 
%
%In this setting, the \gls{wlo}, i.e., the principal, can compute a player's valuation for any bundle, as long as it knows the locations of the player's drivers.
%
\begin{observation}
%	A player's valuation for a bundle corresponds to the minimum cost assignment of drivers to stations. 
	As the principal already knows where stations are located, it can compute a player's cost valuation if it also knows where a player's drivers are located. 
%	Accordingly, it suffices that a player reports the locations of its drivers to the principal rather than explicitly communicating its cost valuation for every bundle. 
\end{observation}
Accordingly, it suffices that a player reports the locations of its drivers to the principal rather than explicitly communicating its cost valuation for every bundle. 
%Accordingly, it suffices that a player reports the locations of its drivers to the principal rather than explicitly communicating its cost valuation for every bundle.
%
\begin{observation}
	Reporting drivers' locations is linear in the number of drivers owned by the player. 
%	Accordingly, it requires much less information-sharing between a player and the principal, than if a player would report its cost valuation for any possible combination of stations.
%	As the reported information corresponds to the locations of a player's drivers, 
%	is linear in the number of drivers owned by the player, instead of being exponential in the number of stations, i.e., as if the player would report its cost valuation for any possible combination of stations.
\end{observation}
Accordingly, reporting drivers' locations requires much less information-sharing between a player and the principal, than if a player would report its cost valuation for any possible combination of stations.
%

%Based on this , the principal compute a global allocation of stations to drivers that minimizes the sum of all player's bundle cost valuations.

%	
The principal computes the allocation of stations to drivers that minimizes the sum of all players' bundle cost valuation and requires in return that each player pays a price for the station assignment. Each player aims to minimize its payoff, which corresponds to the sum of its cost valuation for the received station allocation and the price decided by the principal.
%, i.e., the \gls{wlo}.
%	
We assume that players may lie about their preferences if lying can decrease their payoff compared to truthfully reporting their preferences.
The goal of the principal is then to design a pricing rule, such that all players have an incentive to truthfully report their drivers' locations.
To ensure players' truth-telling behavior, we apply the well known VCG pricing scheme which ensures that it is in a player's best interest to reveal its true preference.

%There are $m=|\mathcal{V}|$ charging stations that  constitute $m$ indivisible resources to be concurrently allocated among the $n=|\mathcal{N}|$ players. Each player has a valuation for each bundle of items, that corresponds to the minimum cost assignment between its drivers and the stations' bundle. Valuations are non-additive, e.g., if a bundle contains more stations than drivers, then removing the non-assigned station from the bundle preserves the valuation. While the valuation is expressed in units of time in the following, it can in practice be transposed to monetary units. 
%%	
%The principal aims to compute a global allocation of stations to drivers that minimizes the sum of all player's bundle valuations.
%%	Item-specific valuations depend on the bundle
%%	
%%	Only the whole-bundle valuation makes sense, as the item-specific valuation depends on the bundle it belongs to by representing the driving time of the assigned driver to it. 
%
%Since the principal can compute a player's valuation, i.e., the minimum cost assignment, for each bundle based the drivers' locations belonging to the player, it is sufficient if players report the locations of their drivers rather than explicitly communicating their valuations to the principal. Accordingly, the complexity of preferences sharing becomes linear, instead of exponential in the number of items.	

In the following, we first analyze the VCG pricing scheme from an offline perspective in Section~\ref{subsec:off_demand}, to derive an upper bound on the system efficiency improvement that can be reached.
We then study how to design an online VCG mechanism for a practical implementation in Section~\ref{subsec:on_demand}.

\subsection{Offline charging demand}\label{subsec:off_demand}
	
%	In the following, we assume that all charging demand 

	We now develop an offline mechanism to ensure truthful reporting of the players' information, i.e., the locations of their drivers. % and the maximal search radius $\bar{S}$.
	Focusing on an offline setting, we assume all players' information to be simultaneously reported. The principal allocates stations based on the revealed information, and each player accordingly receives a subset of stations and the corresponding drivers assignment for each of these stations.
	In Section~\ref{subsec:offvcg_unweighted}, we discuss the mechanism for unweighted players, whereas we discuss how to account for weighted players in Section~\ref{subsec:offvcg_weighted}.
	
%	As it suffices to communicate the locations of the drivers, the goal is to ensure truthful players' revelations of the true locations of their requesting drivers as well as the maximal search radius $\bar{S}$, such that the principal allocates stations based on these revealed information. Each player accordingly receives a subset of stations from the principal and the corresponding drivers assignment for each of these stations.
%%	 A player receives a subset of stations from the principal and the corresponding drivers assignment for each of these stations.
%%		
%	Focusing on an offline setting, we assume all players' information to be simultaneously reported.
%	
	
	\subsubsection{Unweighted VCG Mechanism}\label{subsec:offvcg_unweighted}
	In the following, we adapt the VCG mechanism for payoff minimization and social cost minimization. %, i.e., a player's payoff corresponds to a cost.
% TODO : put somewhere
%	Word on coalition-free assignment
%	
	We consider a setting in which each player $i \in \mathcal{N}$ has some hidden information $\theta_i$, representing the set of its drivers' locations that it chooses to accurately reveal or not. We denote with $\theta = (\theta_i)_{i \in \mathcal{N}}$, the vector that contains all players' information. 
%	
	%	Accordingly, each player's information $\theta_i$ represents the location and the search radius of its assigned drivers, such that $\theta_i=(\theta^k_i)_{k \in \mathcal{D}_i}$, with $\theta^k_i=(l^k,\bar{S}^k)$, and $l^k$ being the location of driver $k$ and $\bar{S}^k$ its search radius.
%
	Let $a \in A$ be the set of alternatives that correspond to assignments of stations to drivers, which the principal computes and communicates to the players.
	 Here, a real station can be assigned only once to a driver and a driver may only be assigned to a station if the driving distance is less than its maximum search radius $\bar{S}$. The artificial station vertex $v_0$ can be assigned to multiple drivers and induces a penalty cost for each assigned driver.
	
	We introduce the function $f(\theta)$ that minimizes the sum of all players' valuations based on their reported information $\theta$, and refer to it as the social choice function
	\begin{equation}
	f(\theta_1, ..., \theta_n) = \argmin_{a \in A} \sum_{i} v_i(\theta_{i}, a)\;.
	\end{equation}
	The valuation  $v_i(\theta_{i}, a)$ of player $i$ for alternative $a$ corresponds to the sum of assignment costs of drivers to stations given the player's information $\theta_{i}$, formally
	\begin{equation}\label{eq:cost_valuation}
	v_i(\theta_{i}, a)  = \sum_{k \in \mathcal{D}_i} c(k,a)\;,
	\end{equation}
	with $c(k,a)$ being the cost of assigning driver $k$ to station $a(k)$
	\begin{equation}\label{eq:cost_VCG}
	c(k,a) = 
	\begin{cases}
	\bar{\beta}, & \text{if } a(k) = v_0   \\
	t_{k,a(k)} , & \text{else }\;. \\
	\end{cases}
	\end{equation}
	\noindent Within an allocation $a$, each driver $k \in \mathcal{D}$ receives an assigned station $a(k) \in \mathcal{V} \bigcup {v_0}$. Using Equation~\ref{eq:cost_valuation}, we can further express the social choice function as 
	\begin{equation}
	f(\theta_1, ..., \theta_n) = \argmin_{a \in A} \sum_{k \in \mathcal{D}} c(k,a) \;.
	\end{equation}	 
	Note that the social choice function minimizes the social cost as defined in Equation~\ref{eq:soc_cost1}, with the additional constraint that no more than one driver can be assigned to each physical station.

%	The social choice function $f$ optimizes the social welfare, i.e., the unweighted sum of all player's payoff functions, which results to the sum of all driver's stations assignment costs. 
%	Note that the social cost is currently not fair with respect to the players, as it may lead to min-cost assignment with one player having no drivers assigned at all, while another may have all its drivers assigned.
%	Here, $v_i(\theta_{i}, a)$ is as follows:
%	\[
%	v_i(\theta_{i}, a)  = \sum_{k \in \mathcal{D}_i} c(k,a)\;,
%	\]
%	with 
%	\begin{equation}\label{eq:cost_VCG}
%	c(k,a) = 
%	\begin{cases}
%	\bar{\beta}, & \text{if } a(k) = v_0   \\
%	t_{k,a(k)} , & \text{else } \\
%	\end{cases}
%	\end{equation}s
%	and $f(\theta_1, ..., \theta_n) = \argmin_{a \in A} \sum_{i} v_i(a)$.
	
	Adapting the VCG mechanisms for payoff minimization leads to non-positive VCG prices, which implies a positive transfer of money from the player to the system.
	Accordingly, the VCG pricing rule is given by
	\begin{equation}\label{eq:vcg_prices}
	p(\theta_{-i}) =  h_{-i} - \sum_{j\neq i} v_j(\theta_{j}, a) \;,
	\end{equation}
	with 
	\begin{equation}
	h_{-i} =  \min_{b \in A} \sum_{j\neq i} v_j(\theta_{j}, b) \;,
	\end{equation}
	and with $a=f(\theta)$, such that each player's payoff function results to
	\begin{equation}\label{eq:payoff}
	u_i = v_i(\theta_{i}, a) - p(\theta_{-i}) = \sum_{j \in \mathcal{N} } v_j(\theta_{j}, a) -	h_{-i}  \;.
	\end{equation}
%	.
	%	
	Using VCG prices, the principal aligns each player's interest with the overall system interest as the term $\sum_{j \neq i} v_j(\theta_{j}, a)$ in Equation~\ref{eq:payoff} aligns the player payoff with the social cost (cf. Equation~\ref{eq:soc_cost1}). Thus, a player minimizes its payoff when telling the truth, i.e., the pricing scheme guarantees incentive-compatibility \cite[cf.][]{NisanRoughgardenEtAl2007}.
%
%	A player receives a subset of stations from the principal and the corresponding drivers assignment for each of these stations.
%	
	VCG prices further ensure that a player does not have an incentive to assign different drivers to each of the stations contained in its received subset (see~Proposition~\ref{prop:no_deviation}).
	\begin{prop}\label{prop:no_deviation}
		A Player has no incentive to deviate from the drivers assignment to stations prescribed  by the allocation $a$ chosen by the principal, when telling the truth.
	\end{prop}
	\begin{proof}
		We refer to Appendix~\ref{app:proofs} for a proof of Proposition \ref{prop:no_deviation}.
	\end{proof}
	 Moreover, VCG prices ensure that truth-telling minimizes a player's payoff even when it has the possibility to deviate from the received assignment (see~Proposition~\ref{prop:no_lie}).
%		
%	VCG prices further ensure that a  player tells the truth, even when having the possibility to deviate from the received assignment. 
%	If lying, a player may decrease its payoff by reassigning drivers to its allocated station subset, which would still result in a higher payoff that when truthfully reporting its information and not deviating from the prescribed assignment.
	A player that misreports information on its drivers may receive a suboptimal stations assignment, which it may improve by modifying the assignment of its drivers to the stations contained in the received subset of stations. However, the resulting assignment will still yield higher payoff than when truthfully reporting its information and not deviating from the prescribed assignment.
	\begin{prop}\label{prop:no_lie}
		A Player has no incentive to lie on its revealed information, even when it can deviate from the realized assignment by the principal
	\end{prop}
	\begin{proof}
		We refer to Appendix~\ref{app:proofs} for a proof of Proposition \ref{prop:no_lie}.
	\end{proof}

	\subsubsection{Weighted VCG Mechanism}\label{subsec:offvcg_weighted}					
	
	The presented VCG mechanism leads to a minimal social cost but does not ensure solution fairness for a single player, e.g., it may lead to an outcome with one player having none of its drivers assigned at all, while another player may have all of its drivers assigned. 
	Accordingly, a player might be better off in some scenarios by not participating in the system-based allocation unless it gets prioritized.
%	 through a larger weight. 
	To mitigate these issues, we now study a VCG mechanism, in which players are weighted to better balance the resulting individual player's assignments.  We note that a VCG mechanism with weighted players guarantees truthfulness \citep[cf.][]{Roberts1979}.
	
%	Yet, assigning weights to players (and social choice) results in an affine minimizer social choice function $f$, which might better balance the resulting individual player's assignments.  We know that a VCG mechanism with affine minimizer social choice function remains truthful \citep[see][]{Roberts1979}.
%	
	To formalize this setting, we detail
	\begin{equation}
	f(\theta_1, ..., \theta_n) = \argmin_{a \in A}  \sum_{i} w_{i} \cdot v_i(a)\;,
	\end{equation}
	with $w_i$ being the weight associated to player $i$. 
	
	In this setting, the weighted pricing rule is
	\begin{equation}
		p^{\text{w}}(\theta_{-i}) = h_{-i} - \sum_{j\neq i} \frac{w_j}{w_i} v_j(\theta_{j}, a) \;,
	\end{equation}
	with $h_{-i}$ being a cost independent of the allocation obtained by $i$, and defined as
	\begin{equation}
		h_{-i} = \min_{b \in A} \sum_{j\neq i} \frac{w_j}{w_i}v_j(\theta_{i}, b) \;
	\end{equation}	
	to ensure positive money transfer from the players to the system.

	\paragraph{Weights optimization:}
	Assuming a weighted players VCG mechanism, the weights can be defined a-priori based on the players' characteristics, e.g., the number of controlled drivers.
	We can also optimally derive weight values such that each player benefits from participating in the centrally optimized system.
%	, for small-sized instances limited to two players.
%	
	Deciding whether to participate in the system or not relates to a game, where each player can \{opt-in, opt-out\} of the system. By opting in, the player receives the VCG allocation and pays the related price, whereas by opting out, the player selfishly assigns stations to drivers, possibly conflicting with other players' assignments.
%	
%	To analyze a computationally tractable setting, 
%	
	We assume that if one player opts out, then all players are forced to opt out and selfishly assign their charging demand, as the noise created by the opted-out participant prevents computing an optimal system assignment for the remaining participants.
%	then all players are forced to opt out and selfishly assign their charging demand. 
	Thus, a player has only an incentive to opt in if its resulting payoff with VCG is lower than its payoff without VCG.
	Accordingly, our goal is to derive weights such that all players profit from participating in the VCG mechanism.
%	, i.e., reducing weights of players that show highest benefits using VCG.

%	Focusing on the 2-players case, it formally relates to a basic game, where each player can \{opt-in, opt-out\} of the system. If both players opt-in, then they pay the resulting weighted VCG assignment cost including the related payments. If at least one player opts out, then both players are assume to optimally assign their drivers with respect to their own demand. We want to derive weights, such that both players are at least as good cooperating than while selfishly allocating stations to their drivers, i.e., their VCG-based payoff is less or equal than the payoff, when both use a selfish strategy.
%%	

	The resulting weights optimization problem can be formalized as follows: let $b=(b_i)_{i\in \mathcal{N}}$ be the selfish strategy profile for all players, and let $u_i(b)$ be the respective payoffs for each player $i$ with strategy profile $b_i$. 
	As we seek to derive the best possible weights, we adapt the pivot rule such that the minimization part of $h_{-i}$ for player $i$ is independent of the a-priori undefined weights for all players but $i$, and let 
	\begin{equation}\label{eq:weight_prices}
	h_{-i} = \min_{b \in A} \sum_{j\neq i} \frac{1}{w_i} v_j(\theta_{i}, b)\;.
	\end{equation}	
	Equation~\ref{eq:weight_prices} still guarantees the price to be non-positive such that the mechanism does not transfer a positive amount of money to its players.
%	, as for a payoff minimization game each player pays $-p$ to the system.
%	
	We denote with  $\text{OPT}_{-i}$ the social optimum realized when $i$ does not exist, as 
	\begin{equation} 
	\text{OPT}_{-i} = 	\min_{b \in A} \sum_{j\neq i} v_j(\theta_{i}, b)\;.
	\end{equation}
	We let $w=(w_i)_{i\in \mathcal{N}}$, with $1 \leq  w_i \leq  W$, and $W$ being the maximal weight value; and define the weights optimization problem as
	
%	1) compute min cost assignment ($b_1$, $b_2$) for 1,2 alone and cost $\text{OPT}_1$, $\text{OPT}_2$\\
%	2) compute utility $u_1(b_1, b_2)= A_1$,  $u_2(b_1, b_2)= A_2$\\
%	3) them solve the (non-convex?) MP, as follows,

\begin{equation}\label{eq:obj}
%\begin{split}
\min_{x \in \mathcal{A}, w} \sum_{i\in\mathcal{N}} w_i \cdot \sum_{k\in \mathcal{D}_i} \sum_{v\in\mathcal{V} \cup \mathcal{V}_0}  x_{kv} t_{kv}
\end{equation}
\begin{equation}\label{eq:assign1}
\sum_{k \in \mathcal{D}} x_{kv} \leq 1 \; \forall v \in \mathcal{V}\cup \mathcal{V}_0  % station gets at most one drievr
\end{equation}
\begin{equation}\label{eq:assign2}
\sum_{v \in \mathcal{V} \cup \mathcal{V}_0} x_{kv} = 1 \; \forall k \in \mathcal{D}  % driver gets one station exactly
\end{equation}
\begin{equation}\label{eq:better_payoff}
\sum_{j\in\mathcal{N}} w_j \cdot \sum_{k\in \mathcal{D}_j} \sum_{v\in\mathcal{V} \cup \mathcal{V}_0}   x_{kv} t_{kv}  -  \text{OPT}_{-i}\leq w_i \cdot u_i(b) \; \; \forall i \in \mathcal{N} 
%\end{split}
\end{equation}

The objective \eqref{eq:obj} minimizes the weighted social cost. 
The first two constraints (\ref{eq:assign1}\&\ref{eq:assign2}) are assignment constraints, while the last constraint \eqref{eq:better_payoff} ensures that the resulting payoff within the mechanism, i.e., $v_i(\theta_i, x) - p^{\text{w}}_i(\theta_{-i})$, is better for each player compared to the payoff when acting selfishly.

	\subsection{Online charging demand}\label{subsec:on_demand}
	
	An offline setting as described in Section~\ref{subsec:off_demand} is not applicable in practice, because it is not possible to delay the station assignment when charging requests arrive in the system. 
	Against this background, we extend the offline VCG mechanism to an online setting to allow for immediate assignment decisions by implementing a delayed VCG mechanism similar to \cite{ParkesSinghEtAl2004}.
	Section~\ref{subsec:onvcg} introduces the mechanism for both unweighted and weighted players, while Section~\ref{subsec:online_policy} describes our online allocation policies.
%	unweighted case with online allocation policies, while Section~\ref{subsec:onvcg_weighted}  extends the setting to the weighted players case.

%	it is unreasonable to delay the station assignment in order to pool requests of multiple drivers. We now extend the offline VCG to an online setting by implementing the delayed VCG mechanism introduced in \cite{ParkesSinghEtAl2004}.
%%	We implement a delayed VCG Mechanism, as introduced in \href{https://web.eecs.umich.edu/~baveja/Papers/mdponlinefinal.pdf}{Parkes' paper}. 
%%
%	Section~\ref{subsec:onvcg} describes the unweighted case with online allocation policies, while Section~\ref{subsec:onvcg_weighted}  extends the setting to the weighted players case.

	\subsubsection{Delayed VCG Mechanism}\label{subsec:onvcg}
	
	In this online setting, different platforms interact with the principal and sequentially reveal their drivers' locations to the principal during the planning horizon. We assume that requests must be served immediately.
	Then, the goal of the principal is to make decisions over time that minimize the expected total cost assignment of all drivers in the system. In this setting, each player makes a delayed payment at the end of the planning horizon, that depends on the realized assignment.

	The information vector $\theta_i$ of each player is sequentially revealed, such that $\theta_i=(\theta^k_i)_{k \in \mathcal{D}_i}$, with $k$ being the $k^{\text{th}}$ driver belonging to player $i$. We let $\theta^t:=\theta^k_i$ be the information revealed at time $t$, identifying the location of the requesting driver $k$ and the player $i$ it belongs to.
%	, which we can access with $pl(\theta^t)$. 
%	, with $\theta^t = (l^t, \bar{S}^t)$. We denote with $\Theta = \{\theta_i : i \in  \mathcal{N}\}$ the information space.
%	o which player and location $k$ the driver belongs.
	We denote with $\Theta = \{\theta_i : i \in  \mathcal{N}\}$ the set of player's information vectors and we define a state $x_t$ as the vector that describes the history of decisions and of revealed information such that $x_t = (\theta^0, ..., \theta^t, a_0, ..., a_{t-1})$. We define policy $\pi = (\pi_1, ..., \pi_t)$ as the sequence of decisions made in each epoch, with $a_t := \pi_t(x_t)$, being the station assigned to driver $\theta^t$ by the principal.
	Each assignment decision induces an immediate cost corresponding to the time required by the associated driver to reach the station, $t_{\theta^t, \pi_t(x_t)}$. At time $t$, a player's cost is the sum of all of its driver's travel times or penalties.
	A player's cost sequentially increases until the planning horizon ends. 
	
	From the principal's perspective, the sequential station assignment problem can be modeled as an \gls{mdp}, whose solution policy corresponds to the policy $\pi$.
	Accordingly, we introduce an \gls{mdp} defined by a policy $\pi$, the state space $X$, and transition functions $p(x'|x,a)$ that describe the probability that the system in state $x$ will transition to post-decision state $x'$ after having taken action $a$. As previously introduced, a state $x$ describes the current stations' assignment of existing drivers in the system and an action $a$ represents the assignment decisions realized by the principal.
	
%	 based on identical probability distribution of drivers location for all players.
%	 Specifically, we model the likelihood arrival of a driver in a certain location, similarly for all players.
%	
	We define the immediate cost $d^i_t(x_t, a_t)$ for each player $i$ as
%	($d^i_t(x_t, a_t) = 0$ or $d^i_t(x_t, a_t) = c(\theta^t, a_t)$ if the requesting driver belongs to $i$ in $x_t$).
	\begin{equation}\label{eq:cost_d}
		d^i_t(x_t, a_t) = 
		\begin{cases}
		c(\theta^t, a_t), & \text{if }  i=\hat{i}    \\
		0 , & \text{else } \;\\
		\end{cases}
	\end{equation}
	with $\hat{i}$ being the player associated to request $\theta^t$, using the station assignment cost $c(\theta^t, a_t)$ as defined in \eqref{eq:cost_VCG}. 
	Let the total immediate cost be $d_t(x_t, a_t) = \sum_{i\in\mathcal{N}}  d^i_t(x_t, a_t)$. Further, we denote by 
	\begin{equation}
		d^i_{<T}(\theta, \pi) = \sum^T_{t=0} d^i_t(x_t=(\theta^0, ..., \theta^t, a_0, ...,  a_{-1}) , \pi_t(x_t))
	\end{equation}	 
	the accumulated cost for player $i$ from $t=0$ until $T$, with reported information $\theta$.
%	A player 
%	
%	each player's expected cost and the expected system cost (the realized assignment cost)
 	We can then define the \gls{mdp} value function $V^\pi$ as the expected value of the summed costs over all decision epochs
 	\begin{equation}
 		V^\pi(x_t) = E_{\pi}[\sum^{T}_{\tau=t} d_t(x_\tau, \pi_\tau(x_\tau))]\;,
 	\end{equation}
 	which we can also recursively express as
	\begin{equation}\label{eq:bellman}
		V^\pi(x_t) = d_t(x_t, \pi_t(x_t)) + \sum_{x_{t+1}} p(x_{t+1}|x_t, \pi_t(x_t))V^\pi(x_{t+1})\;.
	\end{equation}
	The objective of the principal is then to find a policy $\pi$ that minimizes the expected cost value $V^\pi(x_0)$ over the planning horizon, with $x_0$ being the initial system state.
	The realized cost for each individual player corresponds to the summed costs over the entire planning horizon for its drivers, formally 
	\begin{equation}
		v^i(\theta^i, a_{\leq T}) = \sum_{\tau=0}^T d^i_\tau(x_\tau, a_\tau)\; .
	\end{equation}

	Similar to the offline setting in Section~\ref{subsec:off_demand}, we assume that a player $i$ may misreport her information $\theta^k_i$ and we refer to the reported information as $\hat{\theta}$, used by the planner to decide on the next action $\hat{a}_t$.
%	
%	The immediate cost induced based on the reported information depends on the action chosen $\hat{a}_t$ given the reported information and the \textit{actual} location of the driver $\theta^t$, such that $d^i_t((\theta^0, ..., \hat{\theta}^t, \pi), \hat{a}_t) =  c(\theta^t, \hat{a}_t)$ when $i$ reports information $\hat{\theta}^t$ with \textit{actual} information $\theta^t$.
%
	Then, the immediate cost induced by action $\hat{a}_t$ corresponds to the distance from the \textit{actual} driver location to the next decided station, and depends on the \textit{actual} location of the driver $\theta^t$ and the action chosen.
%	 and the \textit{actual} location of the driver $\theta^t$.
%	
	Accordingly, $d^i_t((\theta^0, ..., \hat{\theta}^t, \pi), \hat{a}_t) =  c(\theta^t, \hat{a}_t)$ when $i$ reports information $\hat{\theta}^t$ with \textit{actual} information $\theta^t$.

	Analogous to \cite{ParkesSinghEtAl2004}, we define the mechanism's prices as the difference between the sum of realized costs of all players but $i$ and the optimal realized cost with perfect information without player $i$. Each player pays the price at the end of a given planning horizon. Formally, the pricing rule for player $i$ is
	\begin{equation}\label{eq:payment}
		p^i(\theta, \pi)  = - \sum_{j\neq i} \sum^T_{t=0} d^j_t ((\theta_{<t}, a_{<t -1}),\pi_t(x_t)) + OPT_{\theta_{-i}} \;,
	\end{equation}
	with $OPT_{\theta_{-i}}$ being the optimal assignment cost that can be obtained under full information without $i$ being in the system, and $\pi$ being the optimal MDP-policy, such that
%	that is obtained once all information has been revealed, 
	\begin{equation}
	OPT_{\theta_{-i}} = \sum_{j\neq i} d^j_{<T}(\theta_{-i}, \pi)\; .
	\end{equation}
	
	Then, the mechanism $\mathcal{M}_{\text{on}} = (\Theta,p,\pi)$ defined by the players' information $\Theta$, the pricing rule $p=(p_i)_{i \in \mathcal{N}}$ and the assignment decisions policy $\pi$ constitutes the delayed VCG mechanism.
	The pricing rule defined in Equation~\ref{eq:payment} makes $\mathcal{M}_{\text{on}}$ Bayesian-Nash Incentive-compatible, i.e., 
	\begin{equation}\label{eq:BNIC_1}
	E_{\tau>t }[v^i(\theta^i, a_{\leq T})  -   p^i(\theta, \pi)] \leq  E_{\tau>t}[ v^i(\theta^i,\hat{a}_{\leq T})  -   p^i(\hat{\theta}, \pi)] \; \; \forall \hat{\theta}^{t} \; \forall t 
	\end{equation}
		with $\hat{a}_{\leq T} = (\pi_0(x_0), ..., \pi_t(\hat{x}_t),..., \pi_T(\hat{x}_T))$.  
	Here, $\hat{x}_\tau = (\hat{\theta}_{<\tau}, \hat{a}_{\leq \tau -1})$ with $\hat{\theta}_{<\tau}= (\theta^0, ..., \hat{\theta}^{t-1},$ $ ...., \hat{\theta}^{\tau}) \;\forall \tau \in [t, T]$. We assume that $\hat{\theta}^{\tau} =  \theta^{\tau}$ if player $j\neq i$ reports in epoch $\tau$, i.e., all players but $i$ truthfully report their information.
	Equation~\ref{eq:BNIC_1} ensures that the expected player's payoff cannot be better when the player reports false information $\hat{\theta}^{t}$, instead of $\theta^{t}$.

	% put liink to Parkes 'paper
	\begin{prop}\label{prop:BNIC}
		Mechanism $\mathcal{M}_{\text{on}}$ is Bayesian-Nash incentive-compatible.
	\end{prop}

	\begin{proof}
		We refer to Appendix~\ref{app:proofs} for a proof of Proposition \ref{prop:BNIC}.
	\end{proof}
	
%	
%	Note: We use a greedy policy, as it seems to provide in this very uncertain context already solutions of sufficient quality (compared to the offline solution). A data-driven policy does not outperform in practice the greedy one. However, this work can be further extended, using more advanced solutions to solve the online assignment problem by anticipating future requests and leveraging learning (RL, or deep-RL, ...)
	
%	the scenario-optimization framework, as this is otherwise highly intractable. 

	Similar to the VCG mechanism (see Section~\ref{sec:pricing}), we show that the delayed VCG mechanism $\mathcal{M}$ can be extended to account for weighted players. We can adapt the immediate cost by weighting players, and introducing the new immediate cost 
	\begin{equation}
		\tilde{d}_t(x_t, a_t)  = \sum_{i \in  \mathcal{N} }w_i \cdot d^i_t(x_t,a_t)
	\end{equation}
	with $w_i\geq 1\; \forall i \in \mathcal{N}$.
%	 positively defined  weights $w_i$ as defined in the offline weighted VCG. 
%	 
	 Note that $v^i(\theta^i, a_{\leq T})$ remains unchanged.
	We accordingly update the pricing rule as
	\begin{equation}
	\tilde{p}^i(\theta, \pi)  = - \sum_{j\neq i} \sum^T_{t=0} \frac{w_j}{w_i} d^j_t ((\theta_{<t}, a_{<t -1}),\pi_t(x_t)) + \tilde{OPT}_{\theta_{-i}}
	\end{equation}
	where $\tilde{OPT}_{\theta_{-i}} = \sum_{j\neq i} \frac{w_j}{w_i} d^j_{<T}(\theta_{-i}, \pi)$.
	
	We define $\tilde{\mathcal{M}}_{on}= (\Theta, \tilde{p}, \pi)$ as the weighted delayed VCG mechanism, and show that $\tilde{\mathcal{M}}_{on}$ is still in-expectation incentive-compatible.
%	, with Proposition~\ref{prop:BNIC2}.
	
	\begin{prop}\label{prop:BNIC2}
		Mechanism $\tilde{\mathcal{M}}_{\text{on}}$ is Bayesian-Nash incentive-compatible.
	\end{prop}

	\begin{proof}
		We refer to Appendix~\ref{app:proofs} for a proof of Proposition \ref{prop:BNIC2}.
	\end{proof}

	\subsubsection{Online allocation policy}\label{subsec:online_policy}
	
	In an online setting, each decision stage is triggered by a new driver's request, and the principal assigns a charging station to the driver reported by a player based on the online allocation policy. 
	Finally, the mechanism computes prices at the end of the planning horizon, which terminates after a fixed number of considered requests.
	In the following, we derive two online allocation policies to solve the underlying MDP of the delayed VCG mechanism $\mathcal{M}_{\text{on}}$.
	The first policy \textit{greedy} allocates the closest available station to the requesting driver. The second policy \textit{data-driven} bases on a data-driven algorithm \citep[cf.][]{GarrattiCampi2022}, and learns the policy parameterization based on historical charging requests.
%	 and to provide data-driven resource allocation.
	As both policies are suboptimal in practice, the truthfulness of the mechanism cannot be guaranteed in theory. 
%	
%	However, the player will only lie to decrease its payoff, which results in a lower expected social cost and benefits the whole system accordingly. 
%	We can further assume that participants do not have the capacity to compute such information in practice, such that we can assume participants' truth-telling behavior.
%	
	We notice that if one wants to formally ensure truthfulness, one could extend the mechanism similar to the second chance mechanism as introduced in \cite{NisanRonen2007}. In this case, players have a chance to report a different information only if it improves their utility, and accordingly the expected social cost.
%	A possible remedy to ensure truthfulness in practice is to 
	
%	

%	Their payoff is aligned on the expected system cost. Makes sense to lie to decrease the expected system allocation cost. But in this case, must compute a better MDP solution than the centralized planner, and assume that this is unrealistic.
%	Similar to VCG-based mechanisms \citep[cf.][]{NisanRonen2007}, one can however extend the mechanism such that players have a chance to report a different information that may improve their utility, and accordingly the expected social cost.
%%	
%	However, as the expected utility of a single player is aligned with the expected social cost, the incentive for a player to lie is to achieve closer-to-the optimal expected social cost than the one achieved by the principal's policy. In this case, a player's lie will benefit the whole system. Doing the reasonable assumption that a single player cannot compute a better policy than the principal, we can still argue for guaranteed truthfulness.\\
%	
%	The benefits of the mechanism mostly arise from the centralized coordination: by centrally assigning drivers to stations, the principal may anticipate possible conflicts, as it knows which available station is expected to turn occupied.\\
%	\parskip

%	\parskip
%	\noindent \textit{Greedy policy}: 
	\paragraph{Greedy policy:} The greedy policy is a deterministic policy that assigns a requesting driver to the closest available station, such that $\pi_t(x_t) = \argmin_{v\in \bar{\mathcal{V}}} t_{\theta^t,v}$, with $\theta^t$ being the location of the requesting driver in state $x_t$, and  $\bar{\mathcal{V}}$ being the set of stations that have not been already assigned to preceding drivers.
%	, deduced from $x_t$.

%	todo: we remove the referencce to the sccenario optimization framework as our formulation does not yeild the generalization guarantee, i.e., the risk minimiiztion as the problem does not cast to a scenario optimization problem, because the last constraint depends on all scenarios at once, and not on  individual scenarios.
%	\noindent \textit{Data-driven policy}: 
	\paragraph{Data-driven policy:} The data-driven policy is a probabilistic policy, i.e., a policy that chooses its action based on a probability distribution. 
	To sample actions, we use a parametric data-driven online algorithm to determine the action taken at each decision stage.
%	As probabilistic allocation policy, we use a parametric data-driven online algorithm. 
	Specifically, the algorithm parameters denote with which probability to take a specific action. We learn the parameterization of this algorithm offline, based on a large set of training input sequences, each consisting of charging requests of a fixed length.

%	We refer with parameterization to the probabilities to assign a given action in a given state. 
%	to determine this parameterization similar to the
%	 l , basedt on a parametric data-driven algorithm.
	%	
%	The algorithm parameterization is learned offline based on a large set of training input sequences of charging requests of a fixed length.
%	 scenarios, where a scenario represents a fixed-length input sequence of charging requests. % that describe the location of a driver.
%

	We derive the data-driven algorithm $A$ as follows.
	Formally, we let $\mathcal{I}$ be the set of all possible scenarios $I$ where $I = [\theta^0, ...,\theta^t, ...,  \theta^T] $, with $T+1$ being the input sequence length and $\theta^t \in \Theta$ being the location of a driver in position $t$, such that $\mathcal{D}_I$ corresponds to the set of drivers $k$ contained in $I$.
	Our goal is to find an algorithm $A$ that minimizes the competitive ratio $\alpha$ for all possible input sequences $I \in \Delta$, such that 
	\begin{equation}\label{eq:compete_ratio}
	\begin{aligned}
	A^* =&\argmin_{A} \alpha\\
	& ON(I_l, A) \leq \alpha OPT(I_l) \forall I \in \Delta, \\
	&\sum_{v \in \mathcal{V} \cup \mathcal{V}_0} x_{kv} = 1\; \forall k \in \mathcal{D}_I,\; \forall I \in \Delta, \\
	&\sum_{k \in \mathcal{D}_I} x_{kv} \leq 1\; \forall v \in \mathcal{V},\; \forall I \in \Delta\\
	&x_{kv} \in [0,1] \forall k \in \mathcal{D}_I,\;\forall v \in \mathcal{V}\;\forall I \in \Delta\; ,\\
	\end{aligned}
	\end{equation}
%	We base the solution approach on the \href{https://arxiv.org/pdf/2109.08706.pdf}{Jalota's paper}. 
	where $ON(I, A)$ corresponds to the online solution obtained with $A$ for a given scenario $l$ and  $OPT(I)$ corresponds to the offline solution obtained for $l$.
	The offline solution $OPT(I)$ corresponds to the minimum cost assignment of the drivers in $I$ to charging stations in $v \in \mathcal{V} \cup \mathcal{V}_0$, as
	\begin{equation}\label{eq:off_obj}
	OPT(I) = \min_{x_{kv}\forall k\in \mathcal{D}_I \; \forall v \mathcal{V} \cup \mathcal{V}_0 } \sum_{v\in \mathcal{V}\cup \mathcal{V}_0} \sum_{k \in \mathcal{D}_I} x_{kv} \cdot c(k,v)
	\end{equation}
	\begin{equation}\label{eq:off_1}
	\sum_{v \in \mathcal{V} \cup \mathcal{V}_0} x_{kv} = 1\; \forall k \in \mathcal{D}_I
	\end{equation}
	\begin{equation}\label{eq:off_2}
	\sum_{k \in \mathcal{D}_I} x_{kv} \leq 1\; \forall v \in \mathcal{V}
	\end{equation}
	\begin{equation}\label{eq:off_3}
	x_{kv} \in \{0,1\}\; \forall k \in \mathcal{D}_I,\;\forall v \in \mathcal{V}
	\end{equation}
	\begin{equation}\label{eq:off_4}
		c(k,v) = 
		\begin{cases}
			\bar{\beta}, & \text{if } v = v_0   \lor \gamma(k,v) \geq  \bar{S} \\
		t_{k,v} , & \text{else. }\\
		\end{cases}
	\end{equation}
%	\begin{equation}\label{eq:dd_optsol}
%	\begin{aligned}
%	OPT(I) = \min_{x_{kv}\forall k\in \mathcal{D}_I \; \forall v \mathcal{V} \cup \mathcal{V}_0 } &\sum_{v\in \mathcal{V}\cup \mathcal{V}_0} \sum_{k \in \mathcal{D}_I} x_{kv} \cdot c(k,v)\\
%	&\sum_{v \in \mathcal{V} \cup \mathcal{V}_0} x_{kv} = 1\; \forall k \in \mathcal{D}_I\\
%	&\sum_{k \in \mathcal{D}_I} x_{kv} \leq 1\; \forall v \in \mathcal{V}\\
%	&x_{kv} \in \{0,1\}\; \forall k \in \mathcal{D}_I,\;\forall v \in \mathcal{V}\\
%%	&x_{kv}= 0 \; \forall k \in \mathcal{D}\; \forall v \in \mathcal{V} \land \textit{dist}(k,v) \geq  \bar{S}^k  \\
%	&c(k,v) = 
%	\begin{cases}
%		\bar{\beta}, & \text{if } v = v_0   \lor \gamma(k,v) \geq  \bar{S} \\
%	t_{k,v} , & \text{else. } \\
%	\end{cases}
%	\end{aligned}
%	\end{equation}
%	
	Here, the objective \eqref{eq:off_obj} is to find an assignment of drivers to stations that yields minimal cost. Each driver must be assigned to at least one physical or virtual station \eqref{eq:off_1}, and at most one driver can be assigned to a physical station \eqref{eq:off_2}. A driver assignment cost corresponds to the driving time from its current location to the physical station or to a penalty cost if the station is further away than $\bar{S}$ or if the station is virtual \eqref{eq:off_4}.
	
%	\noindent We provide additional background on the approach to derive A in Appendix~\ref{app:scenario_opt}, based on \cite{JalotaPaccagnanEtAl2021}.
	
%	The goal is to find an algorithm $A$ that minimizes the competitive ratio $\alpha$ for all possible input sequences $I \in \Delta$, such that 
%	\begin{equation}
%	\begin{aligned}
%	A^* =&\argmin_{A} \alpha\\
%	& ON(I_l, A) \leq \alpha OPT(I_l) \forall I \in \Delta, \\
%	&\sum_{v \in \mathcal{V} \cup \mathcal{V}_0} x_{kv} = 1\; \forall k \in \mathcal{D}_I,\; \forall I \in \Delta, \\
%	&\sum_{k \in \mathcal{D}_I} x_{kv} \leq 1\; \forall v \in \mathcal{V},\; \forall I \in \Delta\\
%	&x_{kv} \in [0,1] \forall k \in \mathcal{D}_I,\;\forall v \in \mathcal{V}\;\forall I \in \Delta\; .\\
%	\end{aligned}
%	\end{equation}
%	
	
	To mitigate the computational burden, we chose $A$ to lie on a parametric class and relax the integer constraint on the decision variables $x$. 
	%With $\Delta:=\mathcal{I}$, 
	We then learn the algorithm $A^*$ based on $L$ i.i.d. sampled training input sequences $I_0,...,I_L$ of the uncertainty set $\mathcal{I} \subset \Delta$.
	Accordingly, the learning objective becomes
	\begin{equation}
	\begin{aligned}
	A^* =&\argmin_{A} \alpha\\
	& ON(I_l, A) \leq \alpha OPT(I_l) \forall I_l \in \mathcal{I} ,\; l\in[L] \\
	&\sum_{v \in \mathcal{V} \cup \mathcal{V}_0} x_{kv} = 1\; \forall k \in \mathcal{D}_{I_l},\; \forall I_l \in \mathcal{I} ,\; l\in[L] \\
	&\sum_{k \in \mathcal{D}_{I_l}} x_{kv} \leq 1\; \forall v \in \mathcal{V},\; \forall I_l \in \mathcal{I} ,\; l\in[L] \\
	&x_{kv} \geq 0 \forall k \in \mathcal{D}_{I_l},\;\forall v \in \mathcal{V}\;.\\
	\end{aligned}
	\end{equation}
	
 We introduce the parameterization vector of the algorithm A, as $\vec{p}_{\theta, t} \in \mathbf{R}^{|\mathcal{V}|}$ for each ($\theta, t$)-pair, such that $0 \leq p_{\theta, t,v} \leq 1$ represents the probability for a driver in location $\theta$ with position index $t$ in the input sequence to be assigned to station $v$.
	Thus, we derive the algorithm $A$ with parameterization $\vec{p}^*_{\theta, t}$ for all ($\theta,t$)-pairs, as follows
	\begin{equation}\label{eq:dd_obj}
%	(\vect{p}^*_{\theta, t})_{\theta, t} = \argmin_{p_{\theta, t,v} \forall \theta \in \Theta, t \in [0,...,T], v \in \mathcal{V}} \alpha
	  A = \argmin_{p_{\theta, t,v} \forall \theta \in \Theta, t \in [0,...,T], v \in \mathcal{V}} \alpha
	\end{equation}
	\begin{equation}\label{eq:dd_qual}
	\sum^T_{t=0}  \sum_{v \in \mathcal{V}} p_{\theta^t, t,v} \cdot c(\theta^t, v) \leq \alpha \cdot OPT(I_l) \; \forall l \in [L]\; \theta_t \in I_l\\
	\end{equation}
	\begin{equation}\label{eq:dd_assignment}
	\sum_{v \in \mathcal{V}} p_{\theta, t,v}  =  1\; \forall \theta \in \Theta, \forall t \in [0 ,..., T]
	\end{equation}
%	\begin{equation}\label{eq:dd_assignment_1}
%	p_{\theta, t,v}  =  0\; \forall \theta \in \Theta, \forall t \in [0 ,..., T]
%	\end{equation}
	\begin{equation}\label{eq:dd_capacity}
	\sum^T_{t=0}  \sum_{\theta \in \Theta}  p_{\theta, t,v} \cdot  n(\theta, t, \mathcal{I}) \leq 1\; \forall v \in \mathcal{V}\;,
	\end{equation}
	\begin{equation}\label{eq:dd_domain}
	\alpha \in  \mathbf{R}^+\\	
	\end{equation}
	with $n(\theta, t, \mathcal{I})$ being the likelihood that a driver is located in $\theta$ with position index $t$, estimated based on the occurrence of such combination ($\theta, t$) among all input sequences $I_l,\;l\in[L]$.
%	, and normalized.
%	on the set of input sequences $[L]$. We compute $Pr(\theta, t, \mathcal{I})$ 
%	
	The Objective \eqref{eq:dd_obj} is to find a parameterization that minimizes the competitive ratio $\alpha$, i.e., the ratio between the optimal cost and the cost obtained with the parametric algorithm as introduced in~\eqref{eq:compete_ratio}. Constraint~\eqref{eq:dd_qual} ensures that the ratio between the expected online solution and the optimal offline solution is lower than the minimized competitive ratio $\alpha$. Constraint~\eqref{eq:dd_assignment} enforces that we compute a discrete probability distribution, while Constraint~\eqref{eq:dd_capacity} ensures that we do not assign (in-expectation) more than one driver to a station.
%	 (and differs from Jalota's paper). 
	
	We apply this data-driven policy in state $x_t = (\theta^0, ..., \theta^t, a_0, ..., a_{t-1})$ as follows. First, we exclude already allocated stations, i.e., $v  \in \{a_0, ..., a_{t-1} \}$, from the possible stations assignment by setting the probability to be selected to $0$, such that $p_{\theta,t,v} = 0 $, $\forall\theta \; \text{st. }\theta = \theta^t$. Then, we scale the assignment probabilities of the remaining candidate stations to preserve the discrete probability distribution as 
	$$p'_{\theta,t,v} := \frac{p_{\theta,t,v}}{\sum_{v\in \mathcal{V} p_{\theta,t,v} \cdot \delta_v}} \forall v\in \mathcal{V}\;,$$ 
	with $\delta_v=0$ if $v  \in \{a_0, ..., a_{t-1} \}$. Finally, we select a station $v$ for the driver in location $\theta$ based on the probabilities $p'_{\theta,t,v}$.

	In a setting with weighted players, optimizing weights requires to know each player's payoff resulting from a selfish behavior. However, as requests are unseen, we do not know a player's selfish payoff ahead of time and thus cannot optimize weights a-priori.
	To remedy this issue in practice, one could derive the weights a-posteriori, i.e., after a planning horizon has terminated, and apply them for the next planning horizon. 
	In this case, one needs to additionally ensure that the assignment costs account for the players' weights, as $c(k,v,i) = w_i \cdot c(k,v)$. 
	Consequently, the parameterization needs to be indexed on the players, such that $p_{\theta,t,v,i}$ represents the probability that a driver in location $\theta$ at requesting time $t$ and belonging to player $i$ should be assigned to station $v$.

	\section{Experimental design}\label{sec:experimental_setting}
%	\FloatBarrier
	To analyze the impact of coordinating the charging demand at drivers' and platforms' levels, we derive real-world test instances based on the charging station network for the city of Berlin, Germany (cf. Figure~\ref{fig:map}). 

	We consider three unweighted navigation platforms, i.e., players. The ratio of available stations and requesting drivers is the main factor impacting our results. We accordingly vary the total number of drivers in the system $N\in[2,..,40]$, the search radius ($\bar{S}\in\{1000,2000\}$ meters) and the radius of the circular area in which all drivers depart ($s^r\in\{300,700,1100\}$ meters). We set the penalty $\bar{\beta} = 120$ min, such that failing the search for a driver corresponds to a delay of two hours.
%	,  which would correspond to the time needed to wait until an occupied charger gets freed up..
%	set a large penalty for failing the search $\bar{\beta} = 120$ min.

	\begin{table}[bp]
		\begin{minipage}[t]{0.5\textwidth}				
%			\hspace{-0.5cm}
			\centering
			\scalebox{0.55}{\input{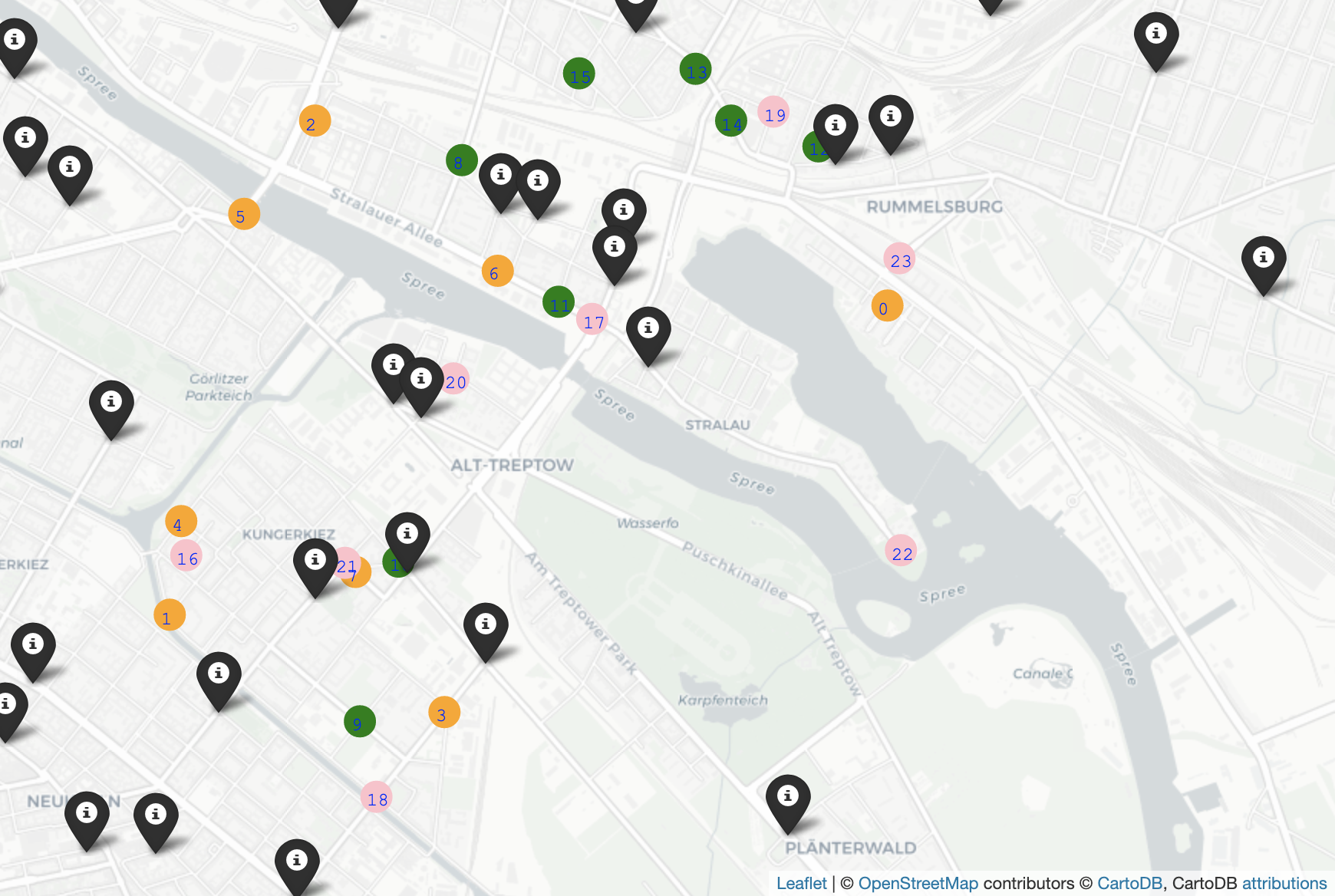}}
			\captionof{figure}{Charging station and driver distribution in a part of Berlin}
			\label{fig:map}			
		\end{minipage}\qquad
		\begin{minipage}[t]{0.4\textwidth}
			\centering
			\vspace{-4cm}
			\hspace{1.5cm}
			\caption{Driver distribution scenarios}
			 {\setlength{\tabcolsep}{0.1cm}\footnotesize
				\begin{tabular}{lrrr}
					\toprule
					\multicolumn{1}{c}{\multirow{2}[4]{*}{scenario}} & \multicolumn{3}{c}{Navigation platform} \\
					\cmidrule{2-4}          & A     & B     & C \\
					\midrule
					\gls{big}   & 25.00\% & 25.00\% & 50.00\% \\
					\gls{all} & 33.33\% & 33.33\% & 33.33\% \\
					\gls{small}   & 40.00\% & 40.00\% & 20.00\% \\
					\bottomrule
			\end{tabular}}
			\label{tab:scenario}%
		\end{minipage} 
		\vspace{-0.3cm}
	\end{table}

%
%
%
%	\begin{figure}[!bp]
%		%		\scalebox{.6}{\includegraphics{chapters/assets/png/3_players/Variant_impact_system_cost}}
%		\centering
%		\scalebox{.6}{\input{chapters/assets_tikz/map}}
%		\caption{Charging station and drivers distribution in a part of Berlin (e.g., for 24 drivers)}
%		\label{fig:map}
%	\end{figure}

%	
	To account for the navigation platforms' heterogeneity, we vary the players' imbalance, with respect to the number of drivers managed by each platform as shown in Table~\ref{tab:scenario}.
	We compare three driver distribution scenarios. In the \gls{small} scenario, one player accounts for 20\% of the total demand, whereas in the \gls{big} scenario, one single player accounts for 50\% of the demand. In both scenarios, the other two players share the remaining demand equally. 
	The \gls{all} scenario represents a homogeneous player scenario, with equal distribution of drivers between players. 
%	Scenarios are summarized in Table~\ref{tab:scenario}.

%	\begin{table}[!bp]
%	  \centering
%	  \caption{Drivers distribution scenarios}
%	    {\setlength{\tabcolsep}{0.1cm}\footnotesize
%	    \begin{tabular}{lrrr}
%	    \toprule
%	    \multicolumn{1}{c}{\multirow{2}[4]{*}{scenario}} & \multicolumn{3}{c}{players} \\
%	\cmidrule{2-4}          & 1     & 2     & 3 \\
%	    \midrule
%	    \gls{big}   & 25.00\% & 25.00\% & 50.00\% \\
%	    \gls{all} & 33.33\% & 33.33\% & 33.33\% \\
%	    \gls{small}   & 40.00\% & 40.00\% & 20.00\% \\
%	    \bottomrule
%	    \end{tabular}}
%	  \label{tab:scenario}%
%	\end{table}%

	For both offline and online charging demand, we benchmark the \gls{vcg}-based assignment against two naive assignment strategies. In \gls{greed}, each EV driver visits the closest station in its vicinity, whereas in \gls{self}, navigation platforms compute the cost-minimum assignment of stations to drivers with respect to their own charging demand. 
%	In both naive strategy cases, visits decisions do not account for possible visit conflicts.
%	
	We compare all three settings, i.e., \gls{vcg}, \gls{greed}, and \gls{self}, with respect to the social cost $S_g$, i.e., the sum of each player's drivers to stations assignment cost, which reflects end-users' satisfaction.
	To analyze the benefits of \gls{vcg} pricing, we further compare each player's payoff, including \gls{vcg} prices when applicable.
	To complement the offline analyses, we analyze the impact of optimally weighting players, such that all players benefit from \gls{vcg} when applying the weighted \gls{vcg} pricing scheme. Here, we solve the optimization problem introduced in Equations~\eqref{eq:obj}-\eqref{eq:better_payoff}, and set the maximal weight value to $W=10$.

	For the online charging demand setting, we assume that for both \gls{greed} and \gls{self}, there is a latency between the time when the drivers stop at the station and the time the station's availability status is up-to-date for the succeeding drivers. We consider a latency of $\tau=3$ minutes, i.e., the time for the driver to stop and start charging and for the system to update the availability information. For \gls{self} the latency only concerns the other player's drivers, as the player knows which stations were assigned to its drivers. We further assume that requests arise every $\Delta_t$ minute with  $\Delta_t \in \{0.5,1.5,2.5\} $ minutes.
%	 and players utilize a greedy online policy in the \gls{self} setting.

	For the data-driven parameterization, we limit the computational burden by identifying a set of at maximum 40 possible locations a driver can depart from ($|\theta|\leq 40$). In most instances, candidate departing locations cover all vertices in the corresponding road network, i.e., roads junctions. 
%	In practice, and with more computational power, we could increase the number of departing locations, such that they can be situated in the middle of a road.
%	 could bet set to a much larger number in a grid world setting, such that each driver belongs to one cell. 
%	TODO : move to limitation or conclusion somehwere elset?
%	Leveraging historical traffic and stations utilization data can help to draw more realistic scenarios used for optimizing the parameters.	
%	
	To evaluate the stochastic data-driven online policy on a test instance with $N$ drivers, we randomly sample $n=500$ start locations for each driver,
%	 among the pool of $40$ candidate start locations. 
%	
%	Besides benchmarking \gls{vcg} with both greedy and data-driven policies against \gls{greed} and \gls{self}, we solve for each value of $N$ each of the $n=500$ instances assuming the charging demand to be known ahead. 
and compute the simulated estimates of the social cost and each player's individual payoff.
	\section{Results}\label{sec:results}

%	TODO: Need to change the stations we consider: assume stations onlu reachable within a certain radius can be assigned to a driver, does that change the theoretical analysis?

	In the following, we first detail our results for an offline charging demand setting in Section~\ref{subsec:off_res}, to obtain an upper bound on the system improvement as well as to understand the interactions between navigation platforms with different driver shares. We further discuss the impact of optimally weighted platforms in this setting in Section~\ref{subsec:off_res_weights}. 
%	which allows additional discussions on players weights.
%	 
	We then discuss our results for an online charging demand setting in Section~\ref{subsec:on_res} and compare it to the offline benchmark.

	\subsection{Offline allocation results}\label{subsec:off_res}

%	 Preliminary results that show that weighting players has little impact,
%	when enough stations, or ok, then no weights, otherwise with the weights conostraints, no feasible solutioon
%   limited range of test instances, where weights can make a diff, but there noo significant link with the proportion of drivers per player 
	
	Figure~\ref{fig:sys_cost} shows the distribution  of the social cost for all three strategies (\gls{self}, \gls{greed}, \gls{vcg}) depending on the number of drivers in the system for both small ($\bar{S}=1000$m) and large ($\bar{S}=2000$m) search radii.
	\begin{figure}[!tp]
%		\scalebox{.6}{\includegraphics{chapters/assets/png/3_players/Variant_impact_system_cost}}
			\centering
			\scalebox{0.65}{\input{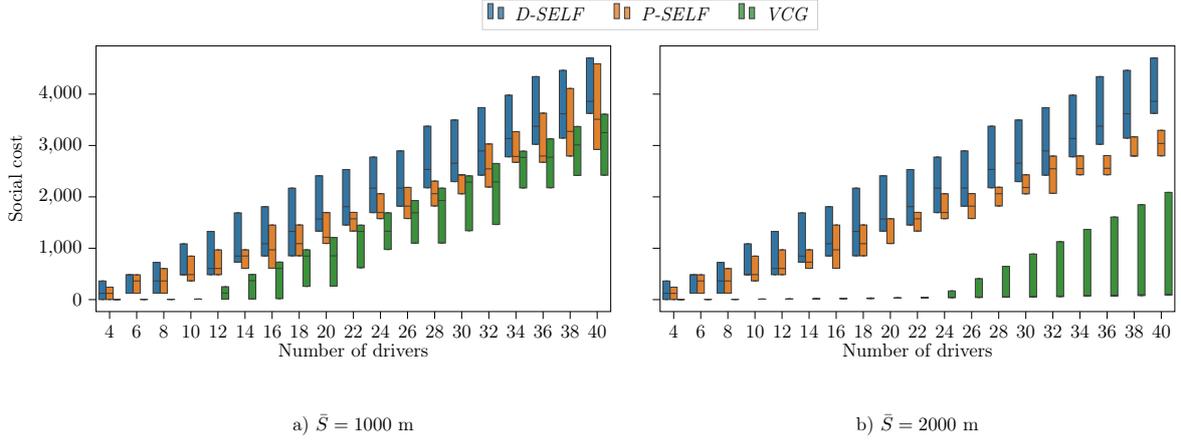}}
		\caption{Benefit of VCG pricing on the social cost}\label{fig:sys_cost}
		\fnote{\footnotesize Results are averaged over all values of $s^r$ with $s^r\{300,700,110\}$ m.}		
%		\caption{Benefit of VCG pricing on the social cost}
%		\label{fig:sys_cost}
%		\fnote{Results are averaged over all values of $s^r$ with $s^r\{300,700,110\}$ m.}
	\end{figure}
	As can bee seen, selfish navigation service platforms can increase the overall user satisfaction 	by optimally balancing their own charging demand (\gls{self}) compared to selfish EV drivers taking greedy visit decisions (\gls{greed}). However, VCG pricing further reduces the social cost significantly. 
%	
%	Our results show that the benefit of platforms coordination increases with the available number of candidate charging stations for the total number of drivers in the system. 
%
	Our results show that the benefit of platform coordination increases when the ratio between the total number of charging stations and the total number of drivers in the system increases.
	In contrast, when the likelihood of visit conflicts is high, i.e., many drivers depart within small distance or with small search radius, the coordination improvement decreases. However, coordination remains necessary in these cases as selfish platform optimization becomes as bad as individual greedy driver decisions.
	
	\begin{res}
		Coordinating platforms with VCG pricing decreases the social cost by up to 52\% compared to an allocation obtained with selfish platforms.
	\end{res}

	Figure~\ref{fig:sys_utility} shows all three platforms' mean payoffs depending on the number of drivers and for each driver distribution scenario. To compare platforms with a heterogeneous driver share, we normalize each platform's payoff by its number of managed drivers.
%	, whereas Figure~\ref{fig:scenario_impact} compares the payoffs obtained for each player depending on the number of drivers, for each distribution scenarios. 
%	To compare players with a heterogeneous share of driver, we normalize each player's payoff by its number of managed drivers.	
	\begin{figure}[!tb]
%		\scalebox{.4}{\includegraphics{chapters/assets/png/3_players/Variant_impact_utility_player}}
		\scalebox{.6}{\input{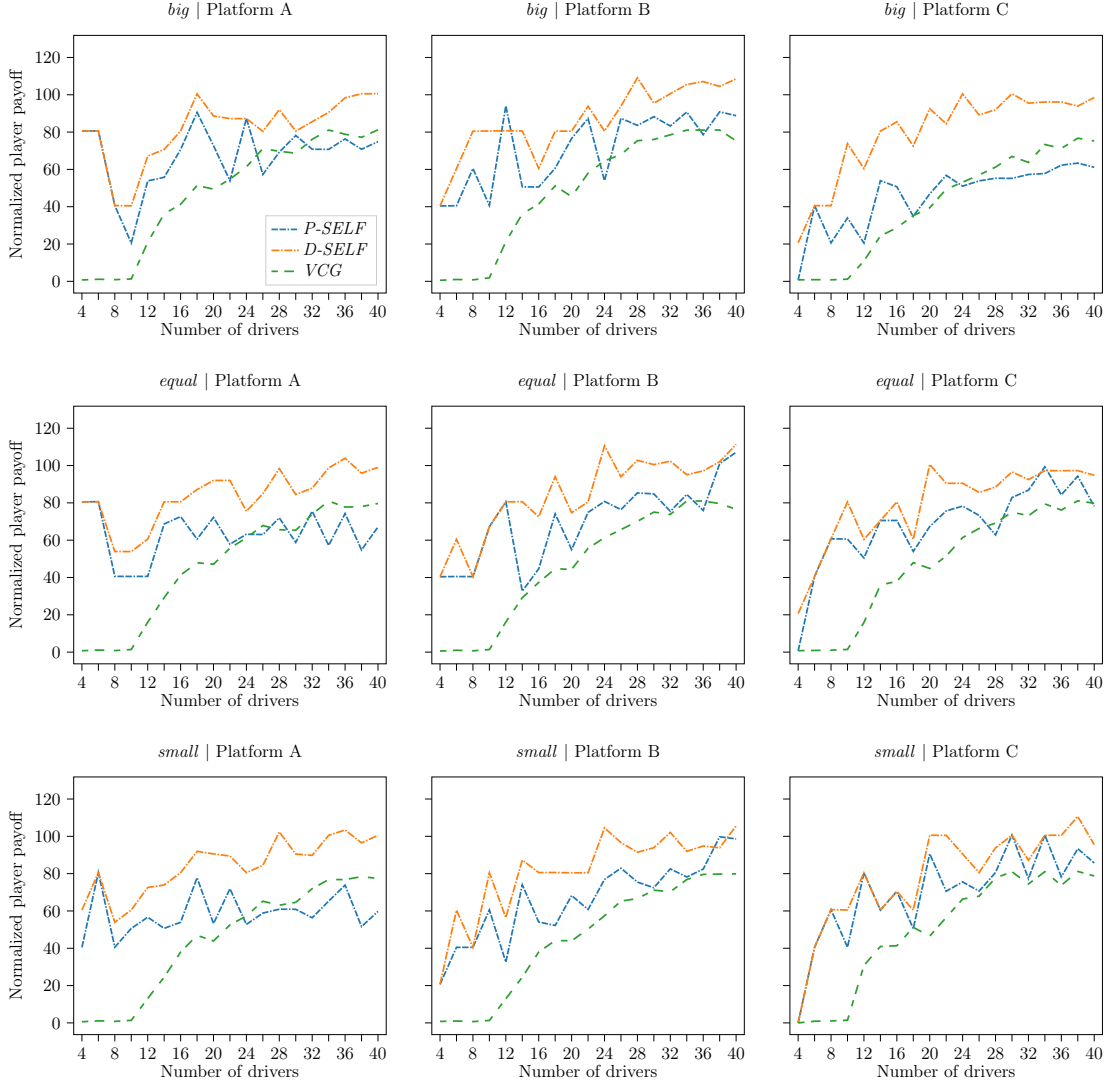}}
		\caption{Impact of the distribution scenario on each platform's payoff for \gls{vcg}, \gls{self}, and \gls{greed}}
		\label{fig:sys_utility}
	\end{figure}
	\begin{figure}[!t]
%		\scalebox{.4}{\includegraphics{chapters/assets/png/3_players/Scenario_impact_utility_player_VCG}}
		\scalebox{.64}{\input{chapters/figures/Scenario_impact_utility_player_VCG}}
		\caption{Impact of the distribution scenario on each platform's payoff}
		\label{fig:scenario_impact}
	\end{figure}
	As can be seen, VCG pricing significantly decreases a platform's payoff on average by 27\% (\gls{all}), respectively 21\% (\gls{big}), and 26\% (\gls{small}) compared to the payoff obtained for platform-optimized allocations (\gls{self}).
	In case of a high charging demand, i.e., for more than $20$ drivers departing in a small vicinity and with lower search times, the induced VCG prices increase a platform's payoff compared to an uncoordinated setting.
	Focusing on the impact of the driver distribution, our results highlight two effects. 
	Utilizing VCG pricing, a platform's payoff decreases if its share of drivers in the system increases (see e.g., Platform C, in Figure~\ref{fig:sys_utility}).
	Contrarily, the improvement obtained with \gls{vcg} pricing compared to a selfish assignment is larger for a platform with proportionally less drivers. In some cases, e.g., for $\bar{S}=1000$ m and a very large number of drivers, \gls{vcg} does not outperform \gls{self} for a platform that manages the highest share of drivers.
	If a platform manages most drivers, then the subsystem composed of other platforms' drivers is smaller such that the platform's impact on it is low, which decreases the large platform's price.
%	 compared to other platforms.
%	 his impact on the subsystem ifs lower in comparison to the other players, and the payment decreases compare to a player with less drivers. 
%	
	In this case, the impact of other platforms on the platform with the largest drivers share decreases as well, such that the benefits of coordination do not always compensate the additional VCG prices induced by coordination.
%	, even if reduced.
	%	
%	 for a platform with an increased share of managed drivers, the impact of other platforms decreases, such that the benefits of coordination do not compensate the additional VCG prices induced by coordination.
%
	\begin{res}
%		A navigation platform managing a low share of drivers benefits from VCG pricing in all instances compared to \gls{self}, compared to a platform with a higher share of drivers.
%		\\
		A navigation platform with a low share of drivers (nearly) always benefits from VCG pricing compared to \gls{self}. In contrast, a platform with a higher share of drivers is sometimes better off not participating in a VCG setting.
	\end{res}

%	 for less drivers, as for a high number of drivers, it is even better for Platform C not to participate in the system coordination.
%	better to have more drivers than fewer than the average, cf Platform C, the payoff is smaller with more drivers. 
%	ii )but we see as well that the improvement compared to a selfish assignment is larger for less drivers, as for a high number of drivers, it is even better for Platform C not to participate in the system coordination.

	Figure~\ref{fig:scenario_impact} compares the mean payoffs obtained for each platform and distribution scenario depending on the number of drivers. 
	In line with Figure~\ref{fig:sys_utility}, results show that a platform with more drivers gets, proportionally to the number of drivers managed, a lower payoff compared to platforms with less drivers. As can be further seen, the impact of a heterogeneous driver share is bigger for a large search radius ($\bar{S} = 2000$ meters) compared to a small search radius ($\bar{S} = 1000$ meters).
	In summary, if a platform participates in a \gls{vcg} setting, its payoff will inversely decrease with the number of drivers it manages in the system, but the benefit of participating in \gls{vcg} compared to \gls{self} with a lot of drivers will decrease, too. If a platform has few drivers, it is better off participating in the system.

	\begin{res}
		If participating in a \gls{vcg} coordinated system, a navigation platform can decrease its cost by increasing its drivers share. However, the relative improvement obtained with \gls{vcg} compared to \gls{self} decreases in this case.
	\end{res}

%Conversely a player with less drivers has a higher utility than payers with more drivers. 
%	
%	however, it does not necessarily imply that a player with more drivers should be
%	
%	While the payoff increases for a player with less drivers, the benefit of using VCG pricing over a selfish drivers assignment increases as well. 
%	Conversely, while the payoff decreases for a player with more drivers, the benefits of using VCG pricing over selfish drivers assignment decreases as well. 
%	
	
	Figure~\ref{fig:metrics}.a shows the average driving time to an available station for any driver in the system, while Figure~\ref{fig:metrics}.b shows the average station assignment success rate for a platform for both $\bar{S}=1000$ m and $\bar{S}=2000$ m.
	\begin{figure}[!tp]
		\begin{tabular}{x{0.45\textwidth}x{.45\textwidth}}
			\scalebox{.8}{% This file was created with tikzplotlib v0.9.15.
\begin{tikzpicture}[font=\small]

\definecolor{color0}{rgb}{0.12156862745098,0.466666666666667,0.705882352941177}
\definecolor{color1}{rgb}{1,0.498039215686275,0.0549019607843137}
\definecolor{color2}{rgb}{0.172549019607843,0.627450980392157,0.172549019607843}

\begin{groupplot}[group style={group size=1 by 2, vertical sep=2cm}]
\nextgroupplot[
legend cell align={left},
legend style={fill opacity=0.8, draw opacity=1, text opacity=1, draw=white!80!black},
scaled x ticks=manual:{}{\pgfmathparse{#1}},
tick align=outside,
tick pos=left,
height=6cm,
width=9cm,
title={$\bar{S}=1000$ m},
x grid style={white!69.0196078431373!black},
xmin=-0.5, xmax=18.5,
xtick style={color=black},
xlabel={Number of drivers},
xtick={0,1,2,3,4,5,6,7,8,9,10,11,12,13,14,15,16,17,18},
xticklabels={4,,8,,12,,16,,20,,24,,28,,32,,36,,40},
y grid style={white!69.0196078431373!black},
ylabel={Driver search time},
ymin=0.137952960392552, ymax=2.63714220318498,
ytick style={color=black}
]
%\addplot [draw=color0, fill=color0, mark=*, only marks]
%table{%
%x  y
%0 0.5083505
%1 0.461277314285714
%2 0.506869540740741
%3 0.526420666666667
%4 0.542973311111111
%5 0.548714148148148
%6 0.566269233333333
%7 0.555562821481482
%8 0.709253651851852
%9 0.566880873333333
%10 0.643499314814815
%11 0.774845027407407
%12 0.833517616190476
%13 0.718842893333333
%14 0.8262734125
%15 0.693792673333333
%16 0.820107311428571
%17 0.855385271693122
%18 1.02199739166667
%};
%\addlegendentry{\gls{self}}
\addplot [line width=1.08pt, color0, dash pattern=on 4pt off 1pt on 4pt off 1pt on 1pt off 1pt]
table {%
0 0.5083505
1 0.461277314285714
2 0.506869540740741
3 0.526420666666667
4 0.542973311111111
5 0.548714148148148
6 0.566269233333333
7 0.555562821481482
8 0.709253651851852
9 0.566880873333333
10 0.643499314814815
11 0.774845027407407
12 0.833517616190476
13 0.718842893333333
14 0.8262734125
15 0.693792673333333
16 0.820107311428571
17 0.855385271693122
18 1.02199739166667
};
\addlegendentry{\gls{self}}
\addplot [line width=1.08pt, color1, dash pattern=on 5pt off 1pt on 1pt off 1pt]
table {%
0 0.4813497
1 0.4732752
2 0.499475033333333
3 0.443555971428571
4 0.488324885714286
5 0.454463455555556
6 0.479707542857143
7 0.568524898888889
8 0.494487453333333
9 0.391235957142857
10 0.45038886
11 0.399836436190476
12 0.50860194
13 0.432797114285714
14 0.425635785714286
15 0.412645871428571
16 0.608018266666667
17 0.433450180952381
18 0.440650351428571
};
\addlegendentry{\gls{greed}}
\addplot [line width=1.08pt, color2, dash pattern=on 4pt off 4pt on 6pt off 6pt]
table {%
0 0.613447066666667
1 0.795072888888889
2 0.701394977777778
3 0.923443166666667
4 0.773287462962963
5 0.877547817037037
6 0.939319059259259
7 0.897138807407407
8 1.02489009671958
9 0.910365586137566
10 0.936881984126984
11 0.880374632380952
12 0.909403037037037
13 0.85360413577381
14 1.00632636148148
15 0.80212037015873
16 0.94954484617284
17 0.873233642328042
18 1.08596465603175
};
\addlegendentry{\gls{vcg}}

\nextgroupplot[
tick align=outside,
tick pos=left,
title={$\bar{S}=2000$ m},
x grid style={white!69.0196078431373!black},
xlabel={Number of drivers},
height=6cm,
width=9cm,
xmin=-0.5, xmax=18.5,
xtick style={color=black},
xtick={0,1,2,3,4,5,6,7,8,9,10,11,12,13,14,15,16,17,18},
xticklabels={4,,8,,12,,16,,20,,24,,28,,32,,36,,40},
y grid style={white!69.0196078431373!black},
ylabel={Driver search time},
ymin=0.137952960392552, ymax=2.63714220318498,
ytick style={color=black}
]
%\addplot [draw=color0, fill=color0, mark=*, only marks]
%table{%
%x  y
%0 0.5083505
%1 0.461277314285714
%2 0.506869540740741
%3 0.526420666666667
%4 0.542973311111111
%5 0.572261651851852
%6 0.5653183
%7 0.555562821481482
%8 0.703225688888889
%9 0.57169854
%10 0.639040885185185
%11 0.743960627407407
%12 0.800871263703704
%13 0.739972183703704
%14 0.86536118
%15 1.01845694253968
%16 1.01008144116402
%17 0.951393305714286
%18 1.01281669296296
%};
\addplot [line width=1.08pt, color0, dash pattern=on 4pt off 1pt on 4pt off 1pt on 1pt off 1pt, forget plot]
table {%
0 0.5083505
1 0.461277314285714
2 0.506869540740741
3 0.526420666666667
4 0.542973311111111
5 0.572261651851852
6 0.5653183
7 0.555562821481482
8 0.703225688888889
9 0.57169854
10 0.639040885185185
11 0.743960627407407
12 0.800871263703704
13 0.739972183703704
14 0.86536118
15 1.01845694253968
16 1.01008144116402
17 0.951393305714286
18 1.01281669296296
};
\addplot [line width=1.08pt, color1, dash pattern=on 5pt off 1pt on 1pt off 1pt]
table {%
0 0.4813497
1 0.4732752
2 0.499475033333333
3 0.443555971428571
4 0.488324885714286
5 0.454463455555556
6 0.479707542857143
7 0.568524898888889
8 0.494487453333333
9 0.391235957142857
10 0.45038886
11 0.399836436190476
12 0.50860194
13 0.432797114285714
14 0.425635785714286
15 0.412645871428571
16 0.608018266666667
17 0.433450180952381
18 0.440650351428571
};
\addplot [line width=1.08pt, color2,  dash pattern=on 4pt off 4pt on 6pt off 6pt]
table {%
0 0.613447066666667
1 0.795072888888889
2 0.701394977777778
3 0.896354614814815
4 0.880042344444444
5 1.12879274
6 1.24552027111111
7 1.39121558518519
8 1.5626079005291
9 1.67910769444444
10 1.83678835
11 1.85784420952381
12 1.91835895555556
13 1.9258249984127
14 2.00516614631874
15 2.12992818492063
16 2.14037928536155
17 2.22580594574685
18 2.29147895479582
};
\end{groupplot}

\end{tikzpicture}}	&
			\scalebox{.8}{% This file was created with tikzplotlib v0.9.15.
\begin{tikzpicture}[font=\small]

\definecolor{color0}{rgb}{0.12156862745098,0.466666666666667,0.705882352941177}
\definecolor{color1}{rgb}{1,0.498039215686275,0.0549019607843137}
\definecolor{color2}{rgb}{0.172549019607843,0.627450980392157,0.172549019607843}

\begin{groupplot}[group style={group size=1 by 2, vertical sep=2cm}]
\nextgroupplot[
legend cell align={left},
legend style={fill opacity=0.8, draw opacity=1, text opacity=1, draw=white!80!black},
scaled x ticks=manual:{}{\pgfmathparse{#1}},
tick align=outside,
tick pos=left,
height=6cm,
width=9cm,
xlabel={Number of drivers},
title={$\bar{S}=1000$ m},
x grid style={white!69.0196078431373!black},
xmin=-0.5, xmax=18.5,
xtick style={color=black},
xtick={0,1,2,3,4,5,6,7,8,9,10,11,12,13,14,15,16,17,18},
xticklabels={4,,8,,12,,16,,20,,24,,28,,32,,36,,40},
y grid style={white!69.0196078431373!black},
ylabel={Driver success rate},
ymin=0.0307692307692308, ymax=1.04615384615385,
ytick style={color=black}
]
%\addplot [draw=color0, fill=color0, mark=*, only marks]
%table{%
%x  y
%0 0.666666666666667
%1 0.555555555555556
%2 0.611111111111111
%3 0.537037037037037
%4 0.527777777777778
%5 0.516666666666667
%6 0.481481481481482
%7 0.481481481481481
%8 0.44973544973545
%9 0.426587301587302
%10 0.388888888888889
%11 0.408950617283951
%12 0.390123456790123
%13 0.366666666666667
%14 0.342424242424242
%15 0.319023569023569
%16 0.342592592592593
%17 0.322649572649573
%18 0.307692307692308
%};
%\addlegendentry{\gls{self}}
\addplot [line width=1.08pt, color0, dash pattern=on 4pt off 1pt on 4pt off 1pt on 1pt off 1pt]
table {%
0 0.666666666666667
1 0.555555555555556
2 0.611111111111111
3 0.537037037037037
4 0.527777777777778
5 0.516666666666667
6 0.481481481481482
7 0.481481481481481
8 0.44973544973545
9 0.426587301587302
10 0.388888888888889
11 0.408950617283951
12 0.390123456790123
13 0.366666666666667
14 0.342424242424242
15 0.319023569023569
16 0.342592592592593
17 0.322649572649573
18 0.307692307692308
};
\addlegendentry{\gls{self}}
\addplot [line width=1.08pt, color1, dash pattern=on 5pt off 1pt on 1pt off 1pt]
table {%
0 0.611111111111111
1 0.5
2 0.574074074074074
3 0.444444444444444
4 0.444444444444444
5 0.361111111111111
6 0.355555555555556
7 0.333333333333333
8 0.261904761904762
9 0.273809523809524
10 0.236111111111111
11 0.270061728395062
12 0.2
13 0.222222222222222
14 0.219191919191919
15 0.196127946127946
16 0.175925925925926
17 0.184472934472934
18 0.156898656898657
};
\addlegendentry{\gls{greed}}
\addplot [line width=1.08pt, color2, dash pattern=on 4pt off 4pt on 6pt off 6pt]
table {%
0 1
1 1
2 1
3 1
4 0.916666666666667
5 0.838888888888889
6 0.77037037037037
7 0.685185185185185
8 0.687830687830688
9 0.573412698412698
10 0.541666666666667
11 0.5
12 0.490123456790123
13 0.444444444444444
14 0.448484848484848
15 0.366161616161616
16 0.37962962962963
17 0.357549857549858
18 0.357142857142857
};
\addlegendentry{\gls{vcg}}

\nextgroupplot[
tick align=outside,
tick pos=left,
height=6cm,
width=9cm,
title={$\bar{S}=2000$ m},
x grid style={white!69.0196078431373!black},
xlabel={Number of drivers},
xmin=-0.5, xmax=18.5,
xtick style={color=black},
xtick={0,1,2,3,4,5,6,7,8,9,10,11,12,13,14,15,16,17,18},
xticklabels={4,,8,,12,,16,,20,,24,,28,,32,,36,,40},
y grid style={white!69.0196078431373!black},
ylabel={Driver success rate},
ymin=0.0307692307692308, ymax=1.04615384615385,
ytick style={color=black}
]
%\addplot [draw=color0, fill=color0, mark=*, only marks]
%table{%
%x  y
%0 0.666666666666667
%1 0.555555555555556
%2 0.611111111111111
%3 0.537037037037037
%4 0.527777777777778
%5 0.538888888888889
%6 0.485185185185185
%7 0.481481481481481
%8 0.481481481481481
%9 0.426587301587302
%10 0.388888888888889
%11 0.422839506172839
%12 0.401234567901235
%13 0.388888888888889
%14 0.361616161616162
%15 0.376262626262626
%16 0.407407407407407
%17 0.366096866096866
%18 0.371794871794872
%};
\addplot [line width=1.08pt, color0, dash pattern=on 4pt off 1pt on 4pt off 1pt on 1pt off 1pt, forget plot]
table {%
0 0.666666666666667
1 0.555555555555556
2 0.611111111111111
3 0.537037037037037
4 0.527777777777778
5 0.538888888888889
6 0.485185185185185
7 0.481481481481481
8 0.481481481481481
9 0.426587301587302
10 0.388888888888889
11 0.422839506172839
12 0.401234567901235
13 0.388888888888889
14 0.361616161616162
15 0.376262626262626
16 0.407407407407407
17 0.366096866096866
18 0.371794871794872
};
\addplot [line width=1.08pt, color1, dash pattern=on 5pt off 1pt on 1pt off 1pt]
table {%
0 0.611111111111111
1 0.5
2 0.574074074074074
3 0.444444444444444
4 0.444444444444444
5 0.361111111111111
6 0.355555555555556
7 0.333333333333333
8 0.261904761904762
9 0.273809523809524
10 0.236111111111111
11 0.270061728395062
12 0.2
13 0.222222222222222
14 0.219191919191919
15 0.196127946127946
16 0.175925925925926
17 0.184472934472934
18 0.156898656898657
};
\addplot [line width=1.08pt, color2,  dash pattern=on 4pt off 4pt on 6pt off 6pt]
table {%
0 1
1 1
2 1
3 1
4 1
5 1
6 1
7 1
8 1
9 1
10 0.986111111111111
11 0.96141975308642
12 0.940740740740741
13 0.922222222222222
14 0.905050505050505
15 0.892255892255892
16 0.87962962962963
17 0.868945868945869
18 0.858363858363858
};
\end{groupplot}

\end{tikzpicture}}	 \\
			\small (a) Search times & \small(b) Success rates \\
		\end{tabular}
		\caption{Impact of VCG pricing on average drivers search times and success rates}
			\label{fig:metrics}
	\end{figure}
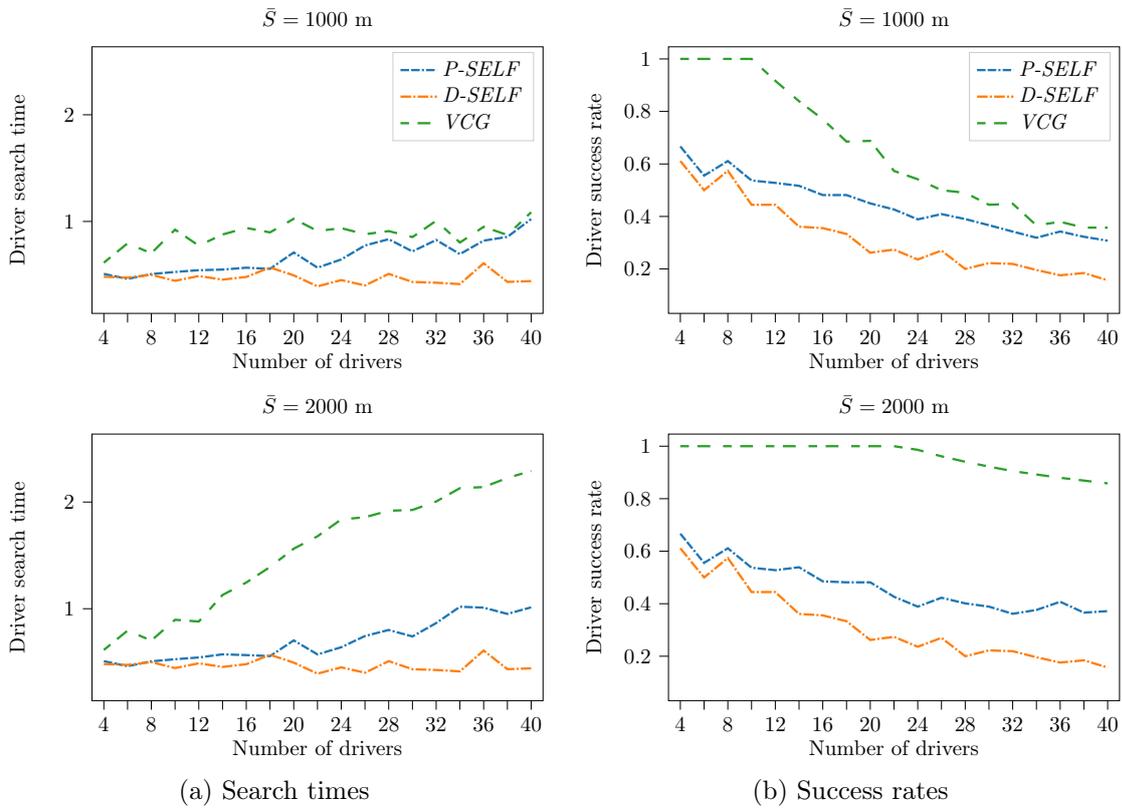
%
%	\begin{figure}[btp]
%		\scalebox{.4}{\includegraphics{chapters/assets/png/3_players/Variant_impact_search_time}}
%		\caption{Impact of the distribution scenario on each player's utility}
%		\label{fig:search_time}
%	\end{figure}
%	\begin{figure}[btp]
%		\scalebox{.2}{\includegraphics{chapters/assets/png/3_players/Variant_impact_success_rate}}
%		\caption{Impact of the distribution scenario on each player's utility}
%		\label{fig:succ_rate}
%	\end{figure}
	Analyzing Figures~\ref{fig:metrics}.a \& \ref{fig:metrics}.b, our results show that coordination through VCG pricing slightly increases the average driven search time needed by its drivers to reach an available station for each platform compared to a selfish or a greedy behavior. However, lower driver search times in selfish settings (\gls{self}\&\gls{greed}) come at the expense of significantly lower success rates, which highlights the need of coordination to decrease station visits conflicts. 
		
	\begin{res}
		Platform coordination mildly increases the average driving time of drivers in the system by 32 seconds on average, compared to local coordination (\gls{self}), and by 47 seconds, compared to uncoordinated drivers (\gls{greed}).
		These longer searches significantly reduce the number of station visit conflicts, leading to a success rate increase of 33\% compared to \gls{self}, and of 49\% compared to \gls{greed}.
%		This search time increase comes at the benefit of a significantly reduced number of station visit conflicts, with a decreased success rate of 72\% on average.
	\end{res}

	\subsection{A-priori weights optimization}\label{subsec:off_res_weights}
		
	In the following, we analyze the impact of optimizing weights prior to applying offline weighted VCG pricing by solving the optimization problem described in Section~\ref{subsec:offvcg_weighted}.
	Table~\ref{tab:diag_weights} summarizes the number of instances in which unweighted VCG pricing ("VCG beneficial") benefits all platforms, i.e., each platform obtains lower cost with VCG than without, as well as the number of instances in which VCG pricing benefits all platforms only if players are optimally weighted \mbox{("Weighted VCG beneficial")}.
	Furthermore, the table shows the number of instances for which optimal weights could not be derived \mbox{("Infeasible weights")}, and the number of instances for which the weights optimization problem could not be solved within the computational time limit \mbox{("Not solved")} of 120 minutes. We differentiate results between low ($\bar{S}=1000$ m) and high ($\bar{S}=2000$ m) search radius.

\begin{table}[!b]
  \centering
%  \SingleSpacedXI
  \caption{Number of test instances that improve with weighted VCG over unweighted VCG, respectively are infeasible, or unsolvable}
     {\setlength{\tabcolsep}{0.08cm}\footnotesize
    \begin{tabular}{ccccc}
    \toprule
    \multicolumn{1}{c}{} & \multicolumn{4}{c}{Number of instances} \\
	\cmidrule{2-5}    $\bar{S}$ & \multicolumn{1}{p{8em}}{VCG beneficial} & \multicolumn{1}{p{14em}}{Weighted VCG beneficial} & \multicolumn{1}{p{6em}}{Infeasible weights} & \multicolumn{1}{p{6em}}{Not solved} \\
    \midrule
    \multicolumn{1}{r}{1000} & 23    & 14    & 83    & 51 \\
    \multicolumn{1}{r}{2000} & 125   & 5     & 24    & 17 \\
    \midrule
    total & 148 (43\%)   & 19 (6\%)   & 107 (31\%)  & 68 (20\%) \\
    \bottomrule
    \end{tabular}}%
  \label{tab:diag_weights}%
\end{table}%

%	\begin{figure}[bp]
%		%		\scalebox{.4}{\includegraphics{chapters/assets/png/3_players/Scenario_impact_utility_player_VCG}}
%		\centering
%		\scalebox{1}{\input{chapters/assets_tikz/pie2_2hours}}
%%		\scalebox{1}{\input{chapters/assets_tikz/pie2}}
%		\caption{Number of test instances that improve with weighted VCG over unweighted VCG, respectively are infeasible, or insolvable}
%		\label{fig:diag_weights}
%	\end{figure}
%	
	As can bee seen, weighted VCG pricing increases the total number of instances in which all platforms benefit from \gls{vcg} participation from 43\% to 49\%. 
%	\\
%As can be seen, the 
%	43 \% of the instances are better for all drivers with the unweighted VCG, and this number raises to 49\% when weights are optimized prior to applying weighted VCG pricing. 
	%		
	There remain 31\% of test instances that yield no feasible weights and 20\% of instances for which the weights optimization problem could not be solved within the limited computational time.
	Focusing on test instances that yield infeasible weights, either unweighted \gls{vcg} already benefits all platforms, or the ratio of available stations per driver is too low to distribute all drivers across stations. In the latter case, prices will be high for any weights values for at least one platform, which accordingly results in higher payoff with \gls{vcg} than with \gls{self}.
	In a limited number of cases (6\% of the total number of test instances), the ratio of stations per driver is low enough such that unweighted \gls{vcg} does not benefit all platforms, but large enough such that a minor reallocation of stations to platforms ensures that all platforms benefit from \gls{vcg} over \gls{self}. Here, platforms who initially did not benefit from \gls{vcg} have larger weights than other platforms.
	Moreover, we observed that the range of improvement through weights optimization depends on the search radius, with 60\% of the instances benefiting from weighted VCG in a low search radius ($\bar{S}=1000$ m) setting against only 4\% benefiting in a large search radius ($\bar{S}=2000$ m) setting. In the latter case, initial benefits of VCG pricing are higher for a larger search radius ($\bar{S}=2000$ m) due to the higher number of available stations per driver, which in turn decreases the improvement potential.
	
	\begin{res}
		Optimizing platforms' weights increases the number of instances in which all platforms benefit from participating in \gls{vcg} from 43\% to 49\%.
%		Excluding infeasible instances, optimizing platforms' weights increases by 13\% the number of instances in which all platforms benefit from participating in \gls{vcg}.
	\end{res}

	To further detail the impact of weighting players, Figures~\ref{fig:weights}a-\ref{fig:weights}c in Appendix~\ref{app:add_num_results} show the averaged payoff per platform and scenario in both unweighted and weighted cases, depending on the total number of drivers. 
%	For instances that could not solve the weights optimization, unweighted payoffs are shown. 
%	Results highlight little differences between weighted players' payoffs and unweighted players' payoffs, and no clear trend on the weights values depending on the scenario. 
%	
%	We optimize weight ssuch that vcg better for all, but NOT such that the difference between players' payoffs is minimized. In the second case, we would likely inrcease the weights of the player that has the highest payoff, such that the mechanism becomes more fair wrt to this aspect.
%	
	Platform C appears to benefit the most from the weight optimization in the \gls{small} scenario, when managing only 20\% of the total number of drivers. Weighting platforms slightly benefits the system by decreasing a platform's payoff.
	However, results show little differences between weighted platforms' payoffs and unweighted platforms' payoffs, and no clear trend on the weights values depending on the scenario. Accordingly, we focus on the unweighted platforms setting in the following analyses.

	\subsection{Online allocation results}\label{subsec:on_res}
	
	The following section analyzes the benefit of \gls{vcg} coordination with online charging requests.
%	 compared to the perfect-information benchmark offline.
%	
 	Preliminary studies (cf. Appendix~\ref{app:add_num_results}) show a similar effect of driver imbalance on a platform's payoff in an online compared to an offline setting.
	Accordingly, we focus the results discussion on the balanced platforms scenario \gls{all} to evaluate the performances of \gls{vcg} and compare the benefits with respect to a perfect-information benchmark (\textit{OFF}).

Figure~\ref{fig:dd_rel} shows the social cost deviation between the data-driven online policy (\gls{vcg}-dd) and the greedy online policy (\gls{vcg}-greedy), computed as follows $\Delta(\hat{\alpha}) = \nicefrac{(\hat{\alpha}_{\text{dd}}  - \hat{\alpha}_{\text{greedy}})}{\hat{\alpha}_{\text{greedy}}}$, with $\hat{\alpha}_{\text{dd}}$, respectively $\hat{\alpha}_{\text{greedy}}$, being the realized social cost for the data-driven, respectively greedy, online policy.
Figure~\ref{fig:dd_abs} compares the search cost distribution obtained with \gls{vcg}-dd and \gls{vcg}-greedy online policies against the perfect-information benchmark (\textit{OFF}), per number of drivers, for a small ($\bar{S}=1000$ meters) and large search radius ($\bar{S}=2000$ meters).
%	

%	\begin{figure}[!bp]
%		%		\scalebox{.6}{\includegraphics{chapters/assets/png/3_players/Variant_impact_system_cost}}
%		\FIGURE{
%			\centering
%			\scalebox{.6}{\input{chapters/assets_tikz/Variant_impact_system_cost}}
%		}
%		{Benefit of VCG pricing on the social cost\label{fig:sys_cost}}
%		{\footnotesize Results are averaged over all values of $s^r$ with $s^r\{300,700,110\}$ m.}		
%		%		\caption{Benefit of VCG pricing on the social cost}
%		%		\label{fig:sys_cost}
%		%		\fnote{Results are averaged over all values of $s^r$ with $s^r\{300,700,110\}$ m.}
%	\end{figure}

\begin{figure}[!b]
	\scalebox{1}{\input{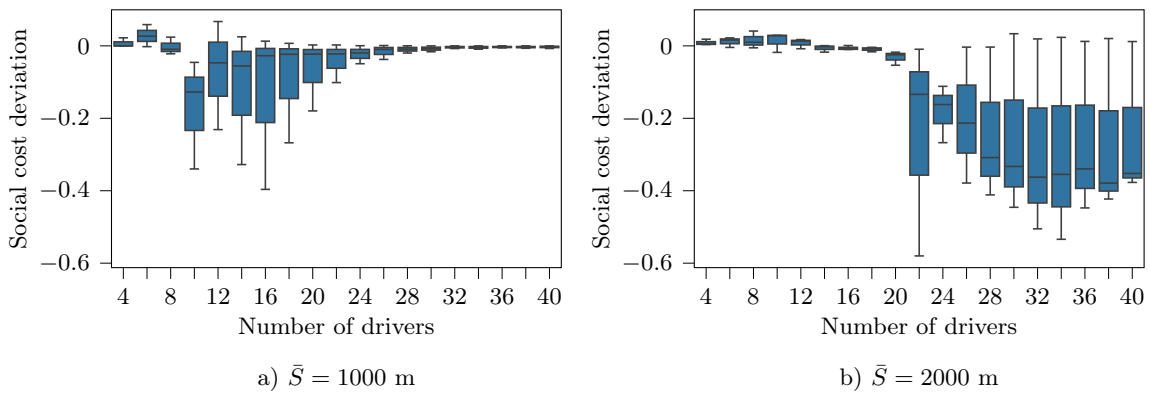}}
	%		\scalebox{.4}{\includegraphics{chapters/assets/png_alternatives/DD_cost_relative}}
	\caption{Social cost deviation  $\Delta(\hat{\alpha})$ between the \gls{vcg}-dd and \gls{vcg}-greed online policies} \label{fig:dd_rel}
\end{figure}
%\\

\begin{figure}[!t]
	\scalebox{1}{\input{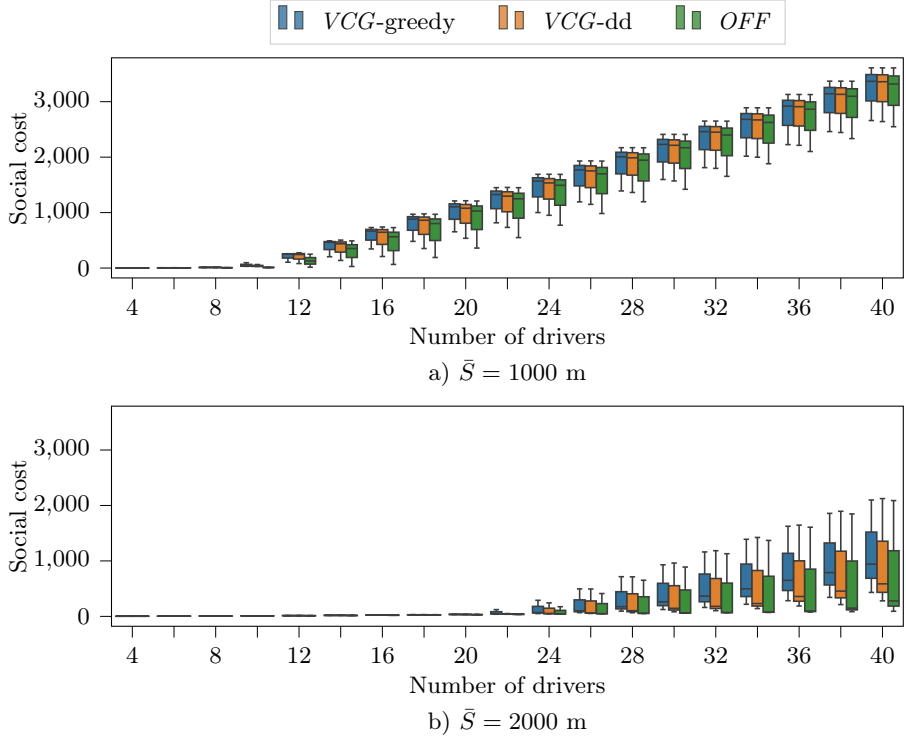}}
	%		\scalebox{.4}{\includegraphics{chapters/assets/png_alternatives/DD_cost_relative}}
	\caption{Social cost obtained with the \gls{vcg}-dd and \gls{vcg}-greedy online policies, and the perfect-information benchmark \textit{OFF}}\label{fig:dd_abs}
	\fnote{\footnotesize Results are aggregated over all values of $s^r$ with $s^r\{300,700,110\}$ m.}		
\end{figure}

%
%	Figure~\ref{fig:dd} compares the relative social cost obtained with the greedy and with the data-driven online policy, compared to the offline \gls{vcg} result. We compute the relative social cost as follows $\Delta(\hat{\alpha}) = \nicefrac{(\hat{\alpha}_{\text{cmp}}  - \hat{\alpha}_{\text{off}})}{\hat{\alpha}_{\text{off}}}$, with $\hat{\alpha}_{\text{off}}$ being the social optimum computed a-posteriori, and $\hat{\alpha}_{\text{cmp}}$ being the realized social cost for the greedy or data-driven online policies respectively.
%	
%	\FloatBarrier
%	
	Our results (cf. Figure~\ref{fig:dd_rel}) show an improvement of the data-driven policy over the greedy allocation policy in an online setting, due to a better anticipation of possible visit conflicts of the data-driven policy. 
%	s
	For $\bar{S}=1000$, the averaged improvement over all values of $s^r$ is 23\% in the best case ($N=10$ drivers), whereas for $\bar{S}=2000$ it increases to 34\% ($N=22$ drivers).
%	of up to 23\% when $\bar{S}=1000$ meters and 34\% when $\bar{S}=2000$,
%	
	As can bee seen, the benefit of a data-driven policy is highest for a small number of drivers in case of a small search radius ($\bar{S}=1000$ meters) but highest for a large number of drivers in case of a large search radius ($\bar{S}=2000$ meters).
	In the former case, the social outcome for a higher number of drivers is generally bad due to a high number of unavoidable station visit conflicts.
	Here, the bottleneck significantly limits the improvement potential, such that the two online policies and the perfect-information setting (\textit{OFF}) yield equally bad outcomes (cf. Figure~\ref{fig:dd_abs}).
	In contrast, the likelihood of station visit conflicts is low for a small number of drivers in the latter case, such that a greedy policy performs already nearly as good as a perfect-information setting allocation (\textit{OFF}). In this case, there will always be enough available stations to allocate to incoming drivers, which will not increase the social cost with penalty costs.

	\begin{res}
		A data-driven online assignment policy improves the social outcome when the improvement potential is highest, i.e., for a small number of drivers with small search radius or for a large number of drivers with large search radius. 
%		The averaged cost reduction is limited to 1.3 \% for a small search radius but increases to 14\% for a large search radius, with a maximal reduction of 55\% for $s^r=700$ meters and $32$ drivers.
	\end{res}

	\begin{res}
%		A data-driven online assignment policy improves the social outcome when the improvement potential is highest, i.e., for a small number of drivers with small search radius or for a large number of drivers with large search radius. 
		A data-driven online assignment policy decreases the social cost on average by 14\% given a large search radius, and yields a maximum reduction of 55\% for $s^r=700$ meters and $32$ drivers, compared to a greedy online assignment policy.
%		The averaged cost reduction is limited to 1.3 \% for a small search radius but increases to 14\% for a large search radius, with a maximal reduction of 55\% for $s^r=700$ meters and $32$ drivers.
	\end{res}

%	Figure~\ref{fig:all} compares the realized social cost obtained with all four online allocation strategies, \gls{self}, \gls{greed}, \gls{vcg}-greedy and \gls{vcg}-dd as well as with the perfect-information allocation \textit{OFF}. Results are split between a small, respectively average and large, delta time between two drivers requests, $\Delta=0.5$ minutes (see~\ref{fig:all}a), respectively  $\Delta=1.5$ (see~\ref{fig:all}b) and $\Delta=2.5$ minutes (see~\ref{fig:all}c).

Table~\ref{tab:delta_impact} shows the deviation of the realized social cost obtained with \gls{self}, \gls{greed}, and \gls{vcg}-greedy compared to the \gls{vcg}-dd online policy.
For \gls{self} and \gls{greed}, we detail results for a small time span between two driver requests ($\Delta_t=0.5$ minutes), an average time span ($\Delta_t=1.5$ minutes) and a large time span ($\Delta_t=2.5$ minutes). For both \gls{vcg}-greedy and \gls{vcg}-dd online policies, results are independent of $\Delta_t$ as the central decision-maker is aware of charging station' availability. 
%(see~\ref{fig:all}a), respectively  $\Delta=1.5$ (see~\ref{fig:all}b) and $\Delta=2.5$ minutes (see~\ref{fig:all}c).
%
As can be seen, both coordinated online assignment policies significantly outperform the two naive benchmarks.
%, while performing close to the perfect-information assignment. 
%In line with Figure~\ref{fig:dd_rel}, the online data-driven policy further decreases the social cost obtained with an online greedy policy.
In line with Figure~\ref{fig:dd_rel}, results further show the superiority of the online \gls{vcg}-dd policy over the online \gls{vcg}-greedy policy, particularly in cases with widespread drivers, i.e., for a large search radius ($\bar{S} \geq $ 2000 m) or for a large departing area  ($s^r \geq 700$ m). 
Our results further show that the benefit of coordinating platforms increases with the requests frequency and are highest for $\Delta_t=0.5$ minutes.
We note that for a lower request frequency, i.e., $\Delta_t=2.5$ minutes, a platform does not benefit from locally coordinating its charging demand. Here, the realized social outcome is better when EV drivers greedily decide on their station visits (\gls{greed}) compared to when navigation platforms locally coordinate their charging demand (\gls{self}). 
%In this case, outdated information due to a higher elapsed time between two requests may negatively bias the local allocation of stations to drivers compared to drivers greedily searching for stations without additional information.
%As the time elapsed between two requests increases, the information outdates faster
%, which differs from the results obtained for the offline charging demand experiments (see Section~\ref{subsec:off_res}).
%

%COMPARED TO DD
% Table generated by Excel2LaTeX from sheet 'Sheet2'
\begin{table}[!tp]
\centering
%  \SingleSpacedXI
  \caption{Social cost deviation for \gls{self}, \gls{greed}, and \gls{vcg}-greedy compared to \gls{vcg}-dd} \label{tab:delta_impact}
    {\setlength{\tabcolsep}{0.08cm}
    	\small
    \begin{tabular}{crrrrrrrrrrr}
   	\toprule
    \multicolumn{1}{c}{$S$} & \multicolumn{1}{c}{$s^r$} &       & \multicolumn{1}{c}{\gls{vcg}-greedy} &       & \multicolumn{3}{c}{\gls{self}} &       & \multicolumn{3}{c}{\gls{greed}} \\
\cmidrule{6-8}\cmidrule{10-12}          &       &       &       &       & \multicolumn{1}{c}{$\Delta_t=0.5$} & \multicolumn{1}{c}{$\Delta_t=1.5$} & \multicolumn{1}{c}{$\Delta_t=2.5$} &       & \multicolumn{1}{c}{$\Delta_t=0.5$} & \multicolumn{1}{c}{$\Delta_t=1.5$} & \multicolumn{1}{c}{$\Delta_t=2.5$} \\
    \midrule
    \multirow{3}[1]{*}{1000} & 300   &       & -0.17\% &       & 15.6\% & 12.2\% & 10.7\% &       & 34.4\% & 11.9\% & 6.20\% \\
          & 700   &       & 0.95\% &       & 13.3\% & 9.23\% & 7.63\% &       & 19.0\% & 8.81\% & 5.31\% \\
          & 1100  &       & 4.23\% &       & 37.6\% & 25.0\% & 20.5\% &       & 51.9\% & 24.3\% & 15.1\% \\
          &       &       &       &       &       &       &       &       &       &       &  \\
    \multirow{3}[0]{*}{2000} & 300   &       & -0.19\% &       & 193\% & 167\% & 155\% &       & 259\% & 171\% & 114\% \\
          & 700   &       & 77.5\% &       & 994\% & 778\% & 670\% &       & 1199\% & 780\% & 536\% \\
          & 1100  &       & 45.6\% &       & 1571\% & 1185\% & 1004\% &       & 1895\% & 1195\% & 811\% \\
     \bottomrule
    \end{tabular}}%
	 \fnote{The table shows the search cost deviation computed as $\Delta(\hat{\alpha}) = \nicefrac{(\hat{\alpha}_{\text{op}}  - \hat{\alpha}_{\text{dd}})}{\hat{\alpha}_{\text{dd}}}$, with $\hat{\alpha}_{\text{dd}}$, being the realized social cost for the data-driven online policy and $\hat{\alpha}_{\text{op}}$ being the realized social cost for \gls{self}, \gls{greed}, and \gls{vcg}-greedy respectively.}
\end{table}%

	\begin{res}
		Online coordination with \gls{vcg} pricing yields a significant system benefit compared to an uncoordinated setting by decreasing the social cost by 42\% when platforms act selfishly (\gls{self}), and 44\% when drivers act selfishly (\gls{greed}).
	\end{res}
	
	\begin{res}
		 Local charging demand coordination (\gls{self}) can yield a higher social cost than without any coordination at all (\gls{greed}).
	\end{res}

	Finally, Figure~\ref{fig:payoff} compares the realized platform payoff, averaged over all demand realizations and platforms, normalized by the number of navigated drivers per platform for each online setting, against the offline benchmark (\textit{OFF}).
	These results highlight the trade-off that exists between the additional cost induced by coordination via VCG pricing and the cost reduction due to a system-optimized allocation.
%	
%	 For the \gls{vcg} (both data-driven and greedy policies), the payoff includes the induced payment. 
	 As can be seen, VCG pricing mostly benefits a platform when the ratio of stations compared to the number drivers in the system is large enough, as otherwise the price paid by a platform increases due to the platform's negative impact on the system. For a large search radius ($\bar{S}=2000$ meters), the online \gls{vcg}-dd policy outperforms any naive allocation strategy, independent of the total number of drivers in the system.
	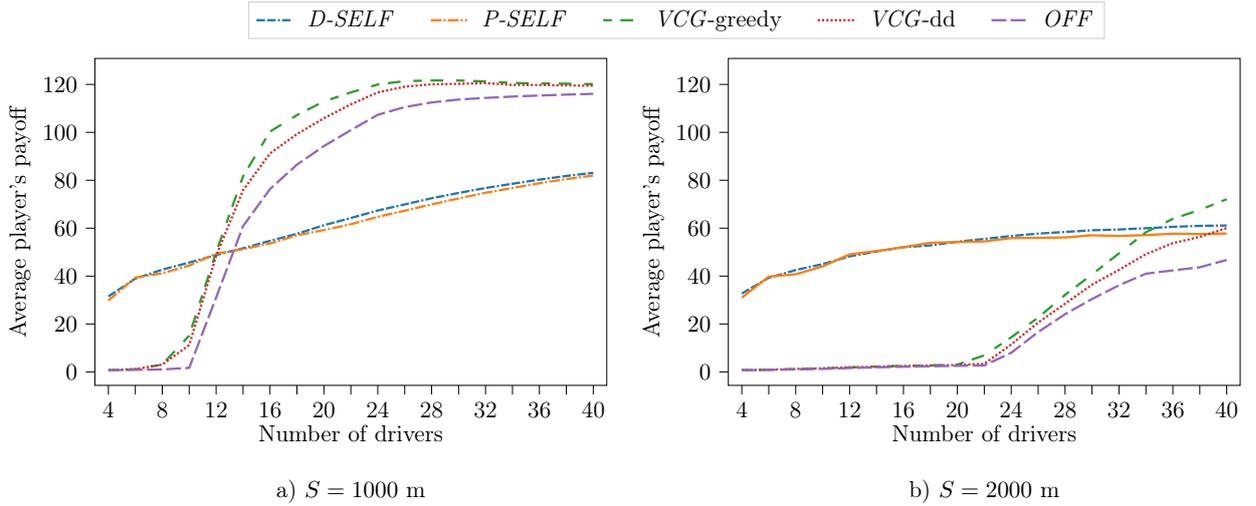
\begin{figure}[tp]
%		\scalebox{.4}{\includegraphics{chapters/assets/png/3_players/online/utility_all}}
		\scalebox{.8}{% This file was created with tikzplotlib v0.9.15.
\begin{tikzpicture}

\definecolor{color0}{rgb}{0.12156862745098,0.466666666666667,0.705882352941177}
\definecolor{color1}{rgb}{1,0.498039215686275,0.0549019607843137}
\definecolor{color2}{rgb}{0.172549019607843,0.627450980392157,0.172549019607843}
\definecolor{color3}{rgb}{0.83921568627451,0.152941176470588,0.156862745098039}
\definecolor{color4}{rgb}{0.580392156862745,0.403921568627451,0.741176470588235}

\begin{groupplot}[group style={group size=2 by 1, horizontal sep=2cm}]
\nextgroupplot[
legend cell align={left},
legend style={fill opacity=0.8, draw opacity=1, text opacity=1, at={(0.3,1.06)}, anchor=south west, draw=white!80!black, legend columns=5, /tikz/every even column/.append style={column sep=0.5cm}, /tikz/every odd column/.append style={column sep=0.15cm}},
%legend style={
%	fill opacity=0.8,
%	draw opacity=1,
%	text opacity=1,
%	at={(0.03,0.97)},
%	anchor=north west,
%	draw=white!80!black
%},
scaled x ticks=manual:{}{\pgfmathparse{#1}},
tick align=outside,
tick pos=left,
title={a) $S = 1000$ m},
title style={at={(0.5,-0.3)},anchor=north,yshift=-0.2},
width=10cm,
height=7cm,
x grid style={white!69.0196078431373!black},
xlabel={Number of drivers},
xmin=-0.5, xmax=18.5,
xtick style={color=black},
xtick={0,1,2,3,4,5,6,7,8,9,10,11,12,13,14,15,16,17,18},
xticklabels={4,,8,,12,,16,,20,,24,,28,,32,,36,,40},
y grid style={white!69.0196078431373!black},
ylabel={Average player's payoff},
ymin=-5.62382852916267, ymax=130.840946856149,
ytick style={color=black}
]
%\addplot [draw=color0, fill=color0, mark=*, only marks]
%table{%
%x  y
%0 31.5036111111111
%1 38.927037037037
%2 42.6844444444444
%3 45.6481111111111
%4 48.7844444444444
%5 51.609126984127
%6 54.7064583333333
%7 57.6513580246914
%8 61.2286111111111
%9 64.2269191919192
%10 67.3774074074074
%11 69.9125213675214
%12 72.4902380952381
%13 74.7008888888889
%14 76.7590277777778
%15 78.5285947712418
%16 80.2610185185185
%17 81.7983040935672
%18 83.0898611111111
%};
%\addlegendentry{\gls{greed}}
\addplot [line width=1.08pt, color0, dash pattern=on 4pt off 1pt on 4pt off 1pt on 1pt off 1pt]
table {%
0 31.5036111111111
1 38.927037037037
2 42.6844444444444
3 45.6481111111111
4 48.7844444444444
5 51.609126984127
6 54.7064583333333
7 57.6513580246914
8 61.2286111111111
9 64.2269191919192
10 67.3774074074074
11 69.9125213675214
12 72.4902380952381
13 74.7008888888889
14 76.7590277777778
15 78.5285947712418
16 80.2610185185185
17 81.7983040935672
18 83.0898611111111
};
\addlegendentry{\gls{greed}}
\addplot [line width=1.08pt, color1, dash pattern=on 5pt off 1pt on 1pt off 1pt]
table {%
0 29.7844444444444
1 39.3731481481481
2 41.16375
3 44.3851111111111
4 49.3837037037037
5 51.2372222222222
6 53.578125
7 57.0339506172839
8 59.2069444444444
9 61.624696969697
10 64.7506481481482
11 67.2859401709402
12 69.9278571428571
13 72.4268148148148
14 74.7734722222222
15 76.7722222222222
16 78.7700617283951
17 80.4752339181287
18 81.9806111111111
};
\addlegendentry{\gls{self}}
\addplot [line width=1.08pt, color2, dash pattern=on 4pt off 4pt on 6pt off 6pt]
table {%
0 0.733459915866667
1 1.08700502862222
2 3.08532961047407
3 15.1536110512074
4 49.9917234508222
5 81.6578632720311
6 100.318192056889
7 107.197621802341
8 112.894047805526
9 116.657512737667
10 119.977193798133
11 121.383288590189
12 121.659152558628
13 121.640632937982
14 121.259313025259
15 120.618304119745
16 120.409284468156
17 120.31287617387
18 120.130151868608
};
\addlegendentry{\gls{vcg}-greedy}
\addplot [line width=1.08pt, color3, dash pattern=on 1pt off 1pt]
table {%
0 0.755566382
1 1.15755987124444
2 3.06361742616296
3 11.1948610642519
4 48.4167400687333
5 75.8658056935333
6 91.1884176935585
7 99.275457338163
8 105.878615186959
9 111.699620462662
10 116.6557591303
11 119.048345092396
12 120.052750348306
13 120.234959617644
14 120.543306269413
15 119.784211644844
16 119.769725321193
17 119.593094504726
18 119.44351832727
};
\addlegendentry{\gls{vcg}-dd}
\addplot [line width=1.08pt, color4, dash pattern=on 7pt off 2pt]
table {%
0 0.695399568088889
1 0.856700558444444
2 1.03916631782222
3 1.68105919174074
4 31.0555401268889
5 60.8485500566356
6 76.3444555713452
7 86.5092621464593
8 94.2554853328931
9 100.913792479995
10 107.333913702933
11 110.540489610706
12 112.479116688028
13 113.690176399316
14 114.377128988278
15 114.925981146382
16 115.376088480156
17 115.740345552186
18 116.066535208813
};
\addlegendentry{\textit{OFF}}

\nextgroupplot[
tick align=outside,
tick pos=left,
title={b) $S = 2000$ m},
title style={at={(0.5,-0.3)},anchor=north,yshift=-0.2},
width=10cm,
height=7cm,
x grid style={white!69.0196078431373!black},
xlabel={Number of drivers},
xmin=-0.5, xmax=18.5,
xtick style={color=black},
xtick={0,1,2,3,4,5,6,7,8,9,10,11,12,13,14,15,16,17,18},
xticklabels={4,,8,,12,,16,,20,,24,,28,,32,,36,,40},
y grid style={white!69.0196078431373!black},
ylabel={Average player's payoff},
ymin=-5.62382852916267, ymax=130.840946856149,
ytick style={color=black}
]
%\addplot [draw=color0, fill=color0, mark=*, only marks]
%table{%
%x  y
%0 32.6858333333333
%1 39.2942592592593
%2 42.5066666666667
%3 45.007
%4 48.2215740740741
%5 50.2796031746032
%6 52.0690972222222
%7 52.7954320987654
%8 54.3609444444444
%9 55.4917676767677
%10 56.7486574074074
%11 57.7354273504273
%12 58.4174206349206
%13 59.1246666666667
%14 59.4454861111111
%15 59.9721895424837
%16 60.5553395061728
%17 60.967514619883
%18 61.1063055555556
%};
\addplot [line width=1.08pt, color0, dash pattern=on 4pt off 1pt on 4pt off 1pt on 1pt off 1pt, forget plot]
table {%
0 32.6858333333333
1 39.2942592592593
2 42.5066666666667
3 45.007
4 48.2215740740741
5 50.2796031746032
6 52.0690972222222
7 52.7954320987654
8 54.3609444444444
9 55.4917676767677
10 56.7486574074074
11 57.7354273504273
12 58.4174206349206
13 59.1246666666667
14 59.4454861111111
15 59.9721895424837
16 60.5553395061728
17 60.967514619883
18 61.1063055555556
};
\addplot [line width=1.08pt, color1,forget plot]
table {%
0 30.9597222222222
1 39.7940740740741
2 40.76375
3 44.0258888888889
4 49.0556481481482
5 50.4673809523809
6 52.0507638888889
7 53.9112345679012
8 54.2139444444444
9 54.4531818181818
10 55.8646759259259
11 55.9824358974359
12 56.1264682539682
13 57.0392962962963
14 56.7348958333333
15 57.1033333333333
16 57.7024382716049
17 57.5470467836257
18 57.7625
};
\addplot [line width=1.08pt, color2, dash pattern=on 4pt off 4pt on 6pt off 6pt, forget plot]
table {%
0 0.732553887866667
1 0.907965768222222
2 1.19979168176296
3 1.47710692461481
4 1.87406179728889
5 2.22488251872
6 2.50679444489481
7 2.70396214968889
8 3.05522345414815
9 6.91126938490476
10 14.4296299601778
11 22.6775463256222
12 32.1699236530193
13 40.6598391283111
14 49.4772937782877
15 58.2936898954795
16 63.912222593963
17 67.6185610198593
18 72.0842876933797
};
\addplot [line width=1.08pt, color3, dash pattern=on 1pt off 1pt, forget plot]
table {%
0 0.752585920844445
1 0.932564912
2 1.23703185724444
3 1.51213747775556
4 1.8967553204
5 2.19715466840889
6 2.47864966781037
7 2.66006141608889
8 2.87598423240635
9 3.38066998810476
10 11.4492863679333
11 20.5945745694975
12 28.4735490444079
13 36.4286026272711
14 42.5657255374618
15 49.1599655505347
16 53.7297335266148
17 56.3904883481413
18 60.017229121724
};
\addplot [line width=1.08pt, color4, dash pattern=on 7pt off 2pt, forget plot]
table {%
0 0.690626344444444
1 0.814961919244444
2 1.04254403285926
3 1.25403796766667
4 1.59477230177778
5 1.89797966390667
6 2.15763361888
7 2.34999662940741
8 2.52075484113228
9 2.70119063301746
10 7.97464538326667
11 16.6339121648852
12 24.0933760122617
13 30.4311737825067
14 36.1633377186667
15 40.9940176108869
16 42.3643402583926
17 43.6780694049727
18 46.7298942619067
};
\end{groupplot}

\end{tikzpicture}}		
		\caption{Averaged platform's payoff for \gls{self}, \gls{greed}, \gls{vcg}-dd, \gls{vcg}-greedy, and \textit{OFF}.}
		\label{fig:payoff}
	\end{figure}

	\begin{res}
		\gls{vcg} pricing significantly outperforms any naive strategy (\gls{self}, \gls{greed}) for larger search areas ($\bar{S}=2000$m) by decreasing a platform's payoff on average by 58\% compared to \gls{self} and by 61\% compared to \gls{greed}.
%		
%		In very congested scenarios, VCG pricing can however increase a player's payoff , when coordination can nearly not improve the social outcome, e.g, for a large number of drivers for with a small search area radius.			
	\end{res}
	
	\begin{res}
%		VCG significantly outperforms any naive strategy (\gls{self}, \gls{greed}) with larger search areas.
		%		
		In some cases, \gls{vcg} pricing may slightly increase a platform's payoff, when the room for improving the social outcome with coordination is small, e.g, for a large number of drivers in a small search area.			
	\end{res}

	\section{Conclusion} \label{sec:conclusion}

	In this paper, we analyzed the dynamics between several self-interested navigation service platforms that seek to best allocate charging stations to \gls{ev} drivers.
	 In this context, we studied the problem of conflicting charging station assignments realized by independent platforms from a game-theoretical perspective and introduced the \gls{fcsag} game. %, in which the platforms constitute the players.
%	
%	We defined a player's payoff to account for penalty costs every time a driver is assigned to a conflicting station but is not the closest driver assigned to it. 
	We showed that the game can neither be represented as a congestion game, nor admits a guaranteed \gls{pne}.
%
%	By defining payoffs to account for a penalty cost for a player, i.e., a platform, whose assigned driver to a conflicting station is not the closest assigned one, we showed that the game cannot be represented as a congestion game, nor admits a guaranteed \gls{pne}.
%	
	To steer the system towards a stable and optimal social outcome, we studied the VCG mechanism in both offline and online settings, such that coordinated platforms' assignment decisions benefit the overall system. We further discussed how to extend the mechanism to account for heterogeneously weighted players in both settings. In the online setting, we introduced a data-driven online allocation policy.	
%	combined VCG with a data-driven online allocation policy.	
%	We discussed in both cases how to extend the mechanism to account for players heterogeneously weighted. We further implemented a data-driven online allocation policy within the online VCG mechanism.
%	
	We analyzed the benefits of implementing our mechanisms by conducting a case study for the city of Berlin, and showed that coordination decreases the social cost by up to 52\% in the offline setting and by up to 42\% in the online setting.
	We further showed in an offline setting that optimized players' weights may increase the likelihood that participants benefit from participating in the VCG mechanism.
	Finally, our results showed that the online data-driven allocation policy performs on average slightly better than a myopic policy, but may significantly improve the social outcome by up to 55\% in congested scenarios.
%	

%\onehalfspacing

%
%% finally the bib file
\singlespacing{
%\footnotesize
\bibliographystyle{model5-names}%\biboptions{authoryear}
\bibliography{main}} % if more than one, comma separated
%
%\newpage
%% and the appendices
\onehalfspacing
\begin{appendices}
	\normalsize
	\section{Proofs}\label{app:proofs}
\FloatBarrier

\begin{proof}[Proof of Proposition \ref{prop:noPNE}]
	We show the non-guaranteed PNE existence by finding a game instance that does not possess any. 
	Figure~\ref{fig:game} visualizes such a game instance. The instance comprises three stations $a$, $b$, and $c$ for two players $p_1$ and $p_2$: the first player has two drivers $d_1$ and $d_2$, the second player has only one driver $e_1$.
%	\begin{figure}
%		\caption{G1' instance with no PNE}
%		\label{fig:game}
%		\input{chapters/assets_tikz/G1_nonash}
%	\end{figure}
	None of the possible game outcomes corresponds to a Nash equilibrium.
	If each player individually optimizes the assignment of its drivers to the available station, we obtain the following profile $((b,c),(c))$, with $c$ conflicting. From there, $p_2$ should re-assign its unique driver to $b$, which is then conflicting. In this case, $p_1$ is better off reassigning $d_1$ to $c$ and $d_2$ to $a$, as the total cost (=17) for strategy $(c,a)$ is lower than the total cost (=18) for strategy $(a,c)$. However, $p_1$ should now re-assign its unique driver to $c$ such that $p_1$'s best response is again its individually optimized assignment solution. The last strategy profile now equals the initial strategy profile, such that the path $P = ((b,c), c), ((b,c),b), ((c,a),b), ((c,a),c), ((b,c),c)$ of length $5$ is an improvement path. Additionally, starting from any other strategy profile not included in $P$, one may reach $P$ within at most two best responses moves, such that there exists no sink in the best response dynamics graph and, accordingly, no \gls{pne}.
	\begin{figure}[btp]
		\caption{G1' instance with no PNE}
		\label{fig:game}
		\tikzset{every picture/.style={line width=0.75pt}} %set default line width to 0.75pt        

\begin{tikzpicture}[x=0.75pt,y=0.75pt,yscale=-1,xscale=1]
%uncomment if require: \path (0,1285); %set diagram left start at 0, and has height of 1285

%Shape: Ellipse [id:dp2746702079311003] 
\draw   (347.1,154.62) .. controls (347.1,151.88) and (349.21,149.66) .. (351.81,149.66) .. controls (354.41,149.66) and (356.53,151.88) .. (356.53,154.62) .. controls (356.53,157.36) and (354.41,159.58) .. (351.81,159.58) .. controls (349.21,159.58) and (347.1,157.36) .. (347.1,154.62) -- cycle ;
%Shape: Ellipse [id:dp2895283566484097] 
\draw   (348.7,270.36) .. controls (348.7,267.62) and (350.81,265.4) .. (353.41,265.4) .. controls (356.01,265.4) and (358.13,267.62) .. (358.13,270.36) .. controls (358.13,273.1) and (356.01,275.32) .. (353.41,275.32) .. controls (350.81,275.32) and (348.7,273.1) .. (348.7,270.36) -- cycle ;
%Straight Lines [id:da5481984456400095] 
\draw  [dash pattern={on 0.84pt off 2.51pt}]  (351.81,154.62) -- (353.41,265.4) ;
%Straight Lines [id:da46889197949102734] 
\draw  [dash pattern={on 0.84pt off 2.51pt}]  (351.81,159.58) -- (518.21,351.8) ;
%Straight Lines [id:da6474230362079321] 
\draw  [dash pattern={on 0.84pt off 2.51pt}]  (518.21,351.8) -- (353.41,275.32) ;
%Shape: Ellipse [id:dp8880328102318942] 
\draw   (518.21,351.8) .. controls (518.21,349.06) and (520.32,346.84) .. (522.93,346.84) .. controls (525.53,346.84) and (527.64,349.06) .. (527.64,351.8) .. controls (527.64,354.54) and (525.53,356.76) .. (522.93,356.76) .. controls (520.32,356.76) and (518.21,354.54) .. (518.21,351.8) -- cycle ;
%Curve Lines [id:da23242944569149038] 
\draw    (355.81,438.2) .. controls (287.47,329.31) and (311.81,189.58) .. (351.81,159.58) ;
%Curve Lines [id:da3395757167764393] 
\draw    (355.81,438.2) .. controls (351.47,346.11) and (334.67,358.47) .. (353.41,275.32) ;
%Curve Lines [id:da6447775240017086] 
\draw    (355.81,438.2) .. controls (433.87,350.91) and (478.21,381.8) .. (518.21,351.8) ;
%Curve Lines [id:da49491369328391466] 
\draw    (504.61,303.8) .. controls (489.07,289.67) and (449.87,251.27) .. (351.81,159.58) ;
%Curve Lines [id:da3649757830045852] 
\draw    (504.61,303.8) .. controls (473.87,300.07) and (406.67,283.27) .. (353.41,275.32) ;
%Shape: Rectangle [id:dp5784515578570744] 
\draw  [fill={rgb, 255:red, 74; green, 144; blue, 226 }  ,fill opacity=1 ] (504.61,303.8) -- (512.09,303.8) -- (512.09,311.27) -- (504.61,311.27) -- cycle ;
%Shape: Rectangle [id:dp437077576789233] 
\draw  [fill={rgb, 255:red, 74; green, 144; blue, 226 }  ,fill opacity=1 ] (355.81,438.2) -- (363.29,438.2) -- (363.29,445.67) -- (355.81,445.67) -- cycle ;
%Curve Lines [id:da16513256292074519] 
\draw    (518.21,351.8) .. controls (519.47,341.31) and (518.67,327.27) .. (512.09,311.27) ;
%Shape: Rectangle [id:dp9666365229768072] 
\draw  [fill={rgb, 255:red, 208; green, 2; blue, 27 }  ,fill opacity=1 ] (357.41,348.6) -- (364.89,348.6) -- (364.89,356.07) -- (357.41,356.07) -- cycle ;
%Curve Lines [id:da1313847305779361] 
\draw    (364.89,348.6) .. controls (364.01,325.13) and (357.87,310.47) .. (353.41,275.32) ;
%Curve Lines [id:da17180780243212546] 
\draw    (364.89,348.6) .. controls (364.01,325.13) and (324.27,300.07) .. (351.81,159.58) ;
%Curve Lines [id:da15433475378207917] 
\draw    (364.89,348.6) .. controls (409.07,345.67) and (451.47,344.07) .. (518.21,351.8) ;

% Text Node
\draw (349.2,132.2) node [anchor=north west][inner sep=0.75pt]   [align=left] {$\displaystyle a$};
% Text Node
\draw (354,249) node [anchor=north west][inner sep=0.75pt]   [align=left] {$\displaystyle b$};
% Text Node
\draw (529.2,341.8) node [anchor=north west][inner sep=0.75pt]   [align=left] {$\displaystyle c$};
% Text Node
\draw (354.61,213.01) node [anchor=north west][inner sep=0.75pt]  [font=\scriptsize] [align=left] {$\displaystyle 5$};
% Text Node
\draw (426.61,252.21) node [anchor=north west][inner sep=0.75pt]  [font=\scriptsize] [align=left] {$\displaystyle 7$};
% Text Node
\draw (417.01,316.21) node [anchor=north west][inner sep=0.75pt]  [font=\scriptsize] [align=left] {$\displaystyle 6$};
% Text Node
\draw (296.21,303.41) node [anchor=north west][inner sep=0.75pt]  [font=\scriptsize,color={rgb, 255:red, 74; green, 144; blue, 226 }  ,opacity=1 ] [align=left] {$\displaystyle 15$};
% Text Node
\draw (353.2,447.4) node [anchor=north west][inner sep=0.75pt]   [align=left] {$\displaystyle d_{1}$};
% Text Node
\draw (517.2,297.8) node [anchor=north west][inner sep=0.75pt]   [align=left] {$\displaystyle d_{2}$};
% Text Node
\draw (363.15,355.33) node [anchor=north west][inner sep=0.75pt]   [align=left] {$\displaystyle e_{1}$};
% Text Node
\draw (329.81,341.81) node [anchor=north west][inner sep=0.75pt]  [font=\scriptsize,color={rgb, 255:red, 74; green, 144; blue, 226 }  ,opacity=1 ] [align=left] {$\displaystyle 10$};
% Text Node
\draw (362.61,310.61) node [anchor=north west][inner sep=0.75pt]  [font=\scriptsize,color={rgb, 255:red, 208; green, 2; blue, 27 }  ,opacity=1 ] [align=left] {$\displaystyle 5$};
% Text Node
\draw (413.81,389.81) node [anchor=north west][inner sep=0.75pt]  [font=\scriptsize,color={rgb, 255:red, 74; green, 144; blue, 226 }  ,opacity=1 ] [align=left] {$\displaystyle 10$};
% Text Node
\draw (400.21,351.41) node [anchor=north west][inner sep=0.75pt]  [font=\scriptsize,color={rgb, 255:red, 208; green, 2; blue, 27 }  ,opacity=1 ] [align=left] {$\displaystyle 4$};
% Text Node
\draw (325.81,229.81) node [anchor=north west][inner sep=0.75pt]  [font=\scriptsize,color={rgb, 255:red, 208; green, 2; blue, 27 }  ,opacity=1 ] [align=left] {$\displaystyle 10$};
% Text Node
\draw (435.41,222.61) node [anchor=north west][inner sep=0.75pt]  [font=\scriptsize,color={rgb, 255:red, 74; green, 144; blue, 226 }  ,opacity=1 ] [align=left] {$\displaystyle 7$};
% Text Node
\draw (521.01,325.81) node [anchor=north west][inner sep=0.75pt]  [font=\scriptsize,color={rgb, 255:red, 74; green, 144; blue, 226 }  ,opacity=1 ] [align=left] {$\displaystyle 3$};
% Text Node
\draw (437.01,280.21) node [anchor=north west][inner sep=0.75pt]  [font=\scriptsize,color={rgb, 255:red, 74; green, 144; blue, 226 }  ,opacity=1 ] [align=left] {$\displaystyle 8$};

\end{tikzpicture}
	\end{figure}
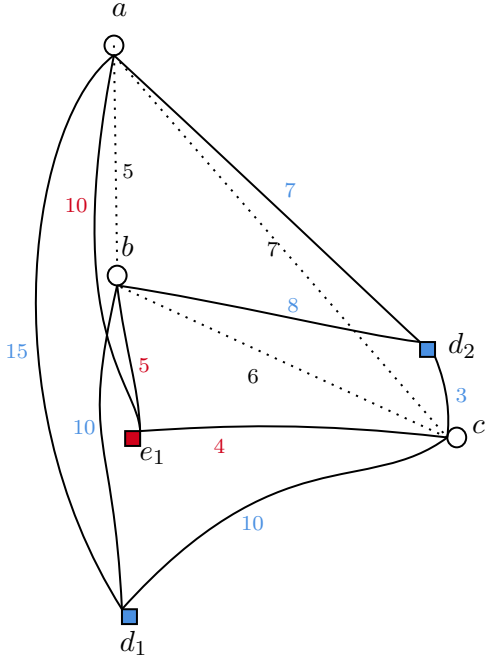
\end{proof}

\begin{proof}[Proof of Proposition \ref{prop:approxNE}]
%The existence of a $\rho$-\gls{pne} with $\rho \leq  \bar{\beta} - \Delta_t$, with $\Delta_t= t_{\text{max}} -  t_{\text{min} }$ being the difference between the largest and the lowest station driver driving times, is not guaranteed for the \gls{fcsag} game
We show the non-existence guarantee of a $\rho$-\gls{pne} with $\rho \leq  \bar{\beta} - \Delta_t$, by using an intermediate game G, defined with the same set of strategies than the \gls{fcsag} game but with new payoff functions. 
We ignore travel times, such that assigning a station induces a penalty cost only for the non-closest drivers, i.e.,
\begin{equation}\label{eq:cost_II}
c_i(k,s) = 
\begin{cases}
\bar{\beta}, & \text{if } s_i(k) = v_0   \\
0 , & \text{if } \forall k' \in \mathcal{D}_j \; \; \forall j \in \mathcal{N}, \;j\neq i \;  (s_j(k') = s_i(k)) \Rightarrow (t_{k',v} \geq t_{k,v})    \\
\bar{\beta}& \text{otherwise } \;.\\
\end{cases}
\end{equation}
Accordingly, player $i$'s payoff can be expressed as $u_i(s) = k\cdot \bar{\beta}$, with $k\in \mathbf{N}$.
%	
%	: station-driver assignment cost ($c_i(k,s)$) equals to $0$ if player $i$ has the closest driver assigned to $k$, else equals to $\bar{\beta}$. 
%	(check code : sr=700 | N=8 | rad=400) 
Similar to Proposition~\ref{prop:noPNE}, G has no guaranteed \gls{pne}, as we can find an instance of G that admits none. For such instance,
\begin{equation}
\forall s \in \mathcal{S}, \forall i \in \mathcal{N}, \exists s'_i \in \mathcal{S}_i, \text{st.  }u_i(s'_i, s_{-i}) < u_i(s)\; .
\end{equation}
%	 $\forall s \in \mathcal{S}, \forall i \in \mathcal{N}, \exists s'_i \in \mathcal{S}_i$, such that 
%	$$ u_i(s'_i, s_{-i}) < u_i(s).$$ 
When a player can strictly decrease its payoff by switching from strategy $s_i$ to $s'_i$, then its payoff decreases by at least $\bar{\beta}$, such that $u_i(s) - u_i(s'_i, s_{-i}) \geq \bar{\beta}$, which corresponds to one additional player being successfully assigned to a station.

By assumption, each feasible profile $s$ in G is feasible in the \gls{fcsag} game and the previous condition implies that for some \gls{fcsag} game instances, a player can decrease its cost by obtaining at least one more driver successfully assigned to a station. 
In this case, assume for player $i$, that $l$ drivers are successfully assigned to a station in $s_i$. For any alternative strategy improving $s'_i$, this implies at least $l+1$ successfully assigned drivers.
Let us denote with $t_{\text{max}}$ and $t_{\text{min}}$ the highest and lowest driving times between stations and drivers. We let $n_i = |\mathcal{D}_i|$ be the number of drivers belonging to $i$. 
Then, we have 
\begin{equation}
u_i(s_i, s_{-i}) \geq (n_i - l)\cdot t_{\text{min}} + l \cdot \bar{\beta}
\end{equation}
and 
\begin{equation}
u_i(s'_i, s_{-i}) \leq (n_i - l - 1)\cdot t_{\text{max}} +( l - 1)\cdot \bar{\beta}\;,
\end{equation}
such that $i$ decreases its cost by switching from $s_i$ to $s'_i$ by at least $\Delta = \bar{\beta} + (t_{\text{min} } - t_{\text{max}}) $.
Finally, we showed that $\rho \geq \Delta$.
%, which is too large to have $\rho$-\gls{pne} as an interesting refinement equilibrium. 

\end{proof}

\begin{proof}[Proof of Proposition~\ref{prop:no_deviation}]

Assuming that all players tell the truth, the principal computes the cost-minimal allocation 
\begin{equation}
a = \argmin_{b \in A}  \sum_{k \in \mathcal{D}}  c(k,b) \; ,
\end{equation}
with $c(k,a)$ as defined in Equation~\eqref{eq:cost_VCG}.
The payoff of player $i$ corresponds to  $\sum_{k \in \mathcal{D}}  c(k,a)  -	h_{-i}$. 
We let $a'$ be the modified allocation, when player~$i$ reassigns its drivers to its stations allocated by the principal. All drivers that do not belong to player~$i$ are still assigned to the same station in $a'$. The payoff of player~$i$ in this case corresponds to 
\begin{equation}
\sum_{k \in \mathcal{D}}  c(k,a')  -	h_{-i}\;.
\end{equation}
However, as allocation~$a$ minimizes the total cost, i.e., the assignment cost of drivers to stations, then 
\begin{equation}
\sum_{k \in \mathcal{D}}  c(k,a') \geq \sum_{k \in \mathcal{D}}  c(k,a) \;,
\end{equation}
such that deviating from the assigned allocation of stations to drivers increases the payoff of a player. 
Concluding, a player has no benefit from assigning its drivers differently to the stations allocated by the principal.
\end{proof}

\begin{proof}[Proof of Proposition~\ref{prop:no_lie}]
First, let us assume that a driver cannot deviate from the received assignment of drivers to stations. Let $\hat{\theta}$ be the reported information, and with a slight abuse of notation let $\hat{k}$ be a reported driver information (i.e., its location), while $k$ denotes the actual driver information.
Let $\hat{a}=f(\hat{\theta})$. 
%
%\sum_{k \in \mathcal{D}}  c(\hat{k},\hat{a})) - h_{-i} =
Then a player's payoff corresponds to $ \sum_{k \in \mathcal{D}}  c(k, \hat{a}) - h_{-i} $, as the assignment cost $ c(\hat{k},\hat{a})$ corresponds to the cost with actual driver information, but with allocation $\hat{a}$ computed based on the driver's reported information by the corresponding player, such that $c(\hat{k},\hat{a}) = c(k, \hat{a})$.
We differentiate two cases:
\begin{enumerate}
	
\item[i)] \textit{Correct number of reported drivers:}
If a single player correctly reports its number of drivers but misreports its drivers' location, such that $\hat{\theta} \neq \theta$, then 
\begin{equation}
\sum_{k \in \mathcal{D}}  c(k, \hat{a}) \geq \sum_{k \in \mathcal{D}}  c(k, f(\theta))\;,
\end{equation} 
as $f(\theta)=\argmin_{b \in A}  \sum_{k \in \mathcal{D}}  c(k,b) $, and the player's payoff increases when not telling the truth.

\item[ii)] \textit{Incorrect number of reported drivers:}
If a single player lies on the number of reported drivers, then each non-reported driver $\hat{k}$ does not get any assigned station, which yields a penalty cost for each of them, such that  $c(k, \hat{a}) = \bar{\beta} \geq c(k, a) \; \forall a \neq \hat{a}$. Then, for each extra reported driver $\hat{k}$ the actual driver assignment cost becomes $c(k, \hat{a}) = 0 $. Accordingly, the principal optimizes the allocation of actual drivers on a reduced subset of stations $\mathcal{V}' \subset \mathcal{V}$, i.e., stations that are not assigned to virtual reported drivers, such that
\begin{equation}
\min_{b \in A, \mathcal{V}'} \sum_{k \in \mathcal{D}}  c(k, b) \geq \min_{b \in A, \mathcal{V}} \sum_{k \in \mathcal{D}, \mathcal{V}}  c(k, a)\;.
\end{equation}
In theses two cases of misreported drivers' numbers, a player's payoff is at least smaller when all drivers are correctly reported. Furthermore, following i), a player's payoff is at least smaller when all players' information are correctly reported.
\end{enumerate}
Second, let us assume that a driver can deviate from the received assignment of drivers to stations. 
%If a player does not deviate, then lying will increase its payoff. 
If a player lies, then its payoff is larger than when telling the truth as shown above, such that 
\begin{equation}
\sum_{k \in \mathcal{D}}  c(k, \hat{a}) + h_{-i} \geq \sum_{k \in \mathcal{D}}  c(k, f(\theta)) + h_{-i} \;,
\end{equation}
as $f(\theta)=\argmin_{b \in A}  \sum_{k \in \mathcal{D}}  c(k,b)$. As its payoff is not minimized when lying, it can deviate from the allocation $\hat{a}$ by reassigning other drivers to its received stations, such that the new allocation $a'$ decreases its payoff obtained with allocation $\hat{a}$. However, the newly obtained allocation $\hat{a}$ cost is at least as large as the minimum cost allocation obtained when telling the truth, such that a player is better off reporting its true information and not reassigning its drivers to different charging stations.
%will not be better than the cost-minimal allocation, i.e., obtained when telling the truth, such that
%
%it will always be larger than the payoff that is minimized when telling the truth, i.e., 
%$$ \sum_{k \in \mathcal{D}}  c(k, \hat{a}) + h_{-i} \geq \sum_{k \in \mathcal{D}}  c(k, a') + h_{-i} \geq \sum_{k \in \mathcal{D}}  c(k, f(\theta)) + h_{-i} \; .$$
		
\end{proof}

\begin{proof}[Proof of Proposition \ref{prop:BNIC}]
	We recall the proof exposed in Theorem 1 from \cite{ParkesSinghEtAl2004}, that transposes the regular VCG incentive-compatibility to an in-expectation incentive-compatibility.
	
	%	Transposing the regular VCG incentives-compatibility justification to an in-expectation incentives compatibility and reusing theorem 1 from \href{https://web.eecs.umich.edu/~baveja/Papers/mdponlinefinal.pdf}{Parkes' paper}, we now justify the Bayesian-Nash Incentives compatibility of $\mathcal{M}$ with the introduced payment rule. 
	%	Justification (cf. VCG incentives-compatibility  + Parkes' paper):
	
	Let $t$ be the last decision epoch before the first epoch, where $i$ misreports information $\theta^{t+1}$ (and may further misreport her information), such that  $\hat{\theta}^{t}= \theta^{t}$. By expliciting payment $p_i$ in Equation~\ref{eq:BNIC_1}, omitting $OPT_{\theta_{-i}}$ that does not depend on player $i$, and omitting $\sum_{i \in \mathcal{N}}\sum^{\tau-1}_{\tau=0} d^t_\tau (x_\tau, \pi_\tau(x_\tau))$ that does not depend on realizations later than $t$, we need to show that
	\begin{equation}\label{eq:BINC}
	\begin{aligned}
	&E_{\tau>t }[\sum_{i\in \mathcal{N}}\sum^T_{\tau=t} d^i_\tau (x_\tau, \pi_\tau(x_\tau))]      \leq  E_{\tau>t} [\sum_{i\in \mathcal{N}}\sum^T_{\tau=t} d^i_\tau (\hat{x}_\tau, \pi_\tau(\hat{x}_\tau))]  \; \; \forall \hat{\theta}^{t} \; \forall t \;\\
	\Leftrightarrow\;\; &E_{\tau>t }[\sum^T_{\tau=t} d_\tau (x_\tau, \pi_\tau(x_\tau))]      \leq  E_{\tau>t} [\sum^T_{\tau=t} d_\tau (\hat{x}_\tau, \pi_\tau(\hat{x}_\tau))]  \; \; \forall \hat{\theta}^{t} \; \forall t \;.
	\end{aligned}
	\end{equation}
	The immediate cost for a misreported information depends on the driver's true information, i.e., actual location, and on the decision taken upon the misreported information, such that 
	$$\sum_{i\in \mathcal{N}}\sum^T_{\tau=t} d^i_\tau (\hat{x}_\tau, \pi_\tau(\hat{x}_\tau)) = \sum_{i\in \mathcal{N}}\sum^T_{\tau=t} d^i_\tau ( (\theta_{<\tau}, \hat{a}_{\leq \tau -1}), \pi_\tau(\hat{x}_\tau))\;.$$
	The right part of~\eqref{eq:BINC} corresponds to the expected cost from $x_t$ (as $\hat{x}_t = x_t$ by assumption) following policy $\hat{\pi}$, with $\hat{\pi}(x_\tau) = \pi(\hat{x}_\tau)\; \forall \tau \in [t, ..., T]$, i.e., $V^{\hat{\pi}}(x_t)$.
	The left part of~\eqref{eq:BINC} corresponds to the expected cost from $x_t$ following policy $\pi$,  $V^{\pi}(x_t)$, such that Equation~\eqref{eq:BINC} can be rewritten as 
	$$ V^{\pi}(x_t) \leq V^{\hat{\pi}}(x_t).$$
	Since $\pi$ is the optimal MDP-policy by assumption, Equation~\ref{eq:BINC} holds; otherwise this contradicts $\pi$'s optimality. Finally, we justified that a player has no incentive to reveal false information.
\end{proof}

\begin{proof}[Proof of Proposition \ref{prop:BNIC2}]
	We need to show that the following inequality 
	\begin{equation}\label{eq:BNIC_2}
	\begin{aligned}
	&E_{\tau>t }[v^i(\theta^i, a_{\leq T})  -   \tilde{p}^i(\theta, \pi)] \leq  E_{\tau>t}[ v^i(\theta^i, \hat{a}_{\leq T})  -   \tilde{p}^i(\hat{\theta}, \pi)] \; \; \forall \hat{\theta}^{t} \; \forall t \\ 
	\Leftrightarrow &w_i \cdot E_{\tau>t }[v^i(\theta^i,a_{\leq T})  -   \tilde{p}^i(\theta, \pi)] \leq  w_i \cdot E_{\tau>t}[ v^i(\theta^i,\hat{a}_{\leq T})  -   \tilde{p}^i(\hat{\theta}, \pi)] \; \; \forall \hat{\theta}^{t} \; \forall t \\
	\end{aligned}
	\end{equation}
	holds. By omitting the same terms in \eqref{eq:BNIC_2} than in \eqref{eq:BNIC_1} (see Proposition~\ref{prop:BNIC}), we can show that the following inequality holds
	\begin{equation}\label{eq:BNIC_3}
	\begin{aligned}
	\Leftrightarrow & E_{\tau>t }[\sum_{i\in \mathcal{N}}\sum^T_{\tau=t} w_i \cdot d^i_\tau (x_\tau, \pi_\tau(x_\tau))]      \leq   E_{\tau>t} [\sum_{i\in \mathcal{N}}\sum^T_{\tau=t} w_i \cdot d^i_\tau (\hat{x}_\tau, \pi_\tau(\hat{x}_\tau))]  \; \; \forall \hat{\theta}^{t} \; \forall t \;\\
	\Leftrightarrow & E_{\tau>t }[\sum^T_{\tau=t} \tilde{d}_\tau (x_\tau, \pi_\tau(x_\tau))]      \leq E_{\tau>t} [\sum^T_{\tau=t} \tilde{d}_\tau (\hat{x}_\tau, \pi_\tau(\hat{x}_\tau))]  \; \; \forall \hat{\theta}^{t} \; \forall t \;.
	\end{aligned}
	\end{equation}
	Using the same argument as previously, with $\pi$ being in this case the optimal policy solving the MDP with updated immediate cost $\tilde{d}$, we justified inequality  \eqref{eq:BNIC_3}. Accordingly, a player has no incentive to reveal false information when players are weighted, if the payment rule is weighted too.
\end{proof}

\section{Additional Numerical results}\label{app:add_num_results}

\paragraph{Weights optimization:}

	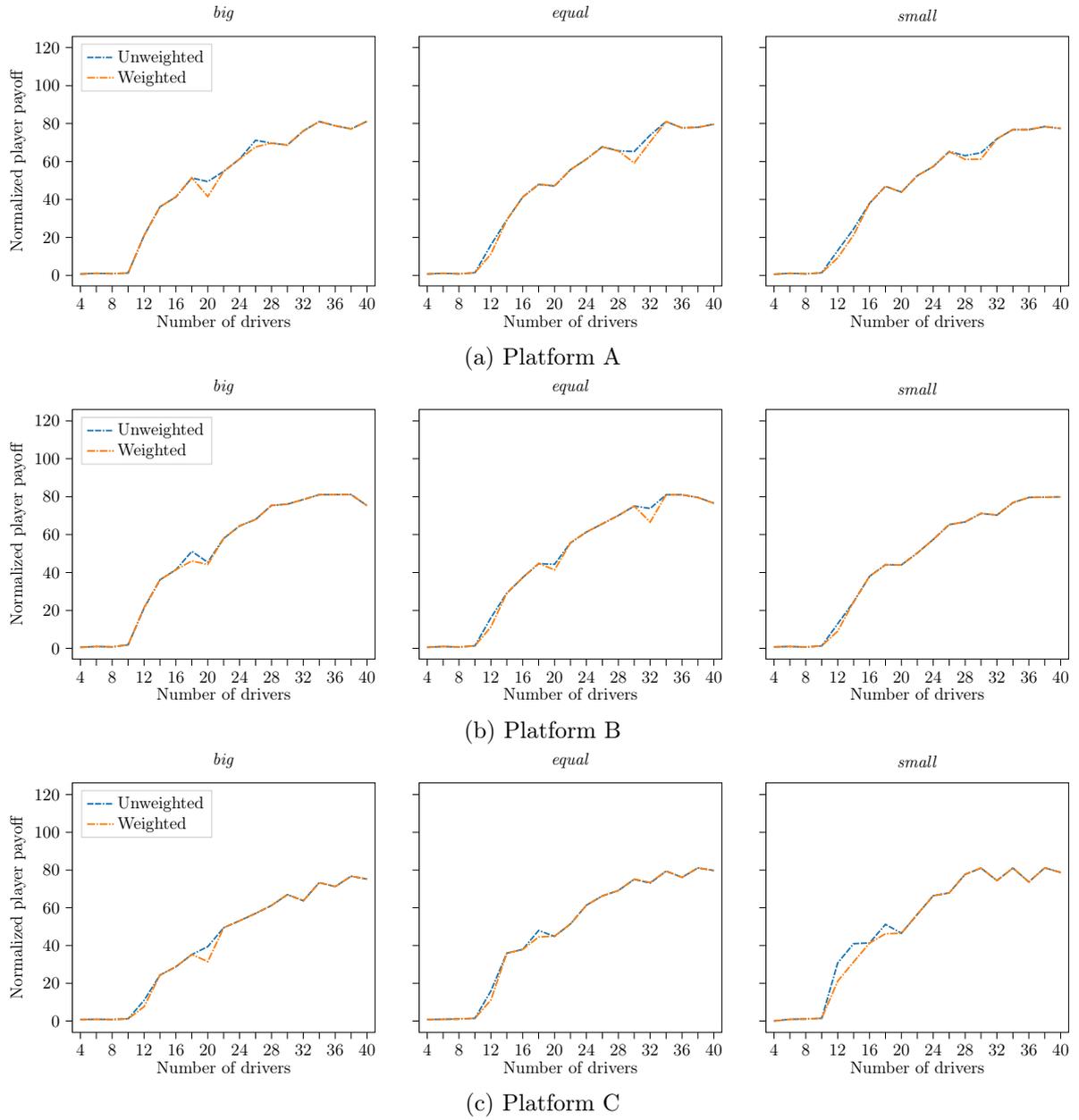
\begin{figure}[!bp]
	\begin{tabular}{x{1\textwidth}}
		\scalebox{.65}{% This file was created with tikzplotlib v0.9.15.
\begin{tikzpicture} %[font=\small]

\definecolor{color0}{rgb}{0.12156862745098,0.466666666666667,0.705882352941177}
\definecolor{color1}{rgb}{1,0.498039215686275,0.0549019607843137}

\begin{groupplot}[group style={group size=3 by 1}]
\nextgroupplot[
legend cell align={left},
legend style={
  fill opacity=0.8,
  draw opacity=1,
  text opacity=1,
  at={(0.03,0.97)},
  anchor=north west,
  draw=white!80!black
},
tick align=outside,
tick pos=left,
title={\gls{big}},
x grid style={white!69.0196078431373!black},
xlabel={Number of drivers},
xmin=-0.5, xmax=18.5,
xtick style={color=black},
xtick={0,1,2,3,4,5,6,7,8,9,10,11,12,13,14,15,16,17,18},
xticklabels={4,,8,,12,,16,,20,,24,,28,,32,,36,,40},
y grid style={white!69.0196078431373!black},
ylabel={Normalized player payoff},
ymin=-5.53197659367708, ymax=125.917193133885,
ytick style={color=black}
]
%\addplot [draw=color0, fill=color0, mark=*, only marks]
%table{%
%x  y
%0 0.7321472
%1 1.1099724
%2 0.896994266666666
%3 1.27613036666667
%4 20.9647478
%5 36.1308948166667
%6 41.3185177666667
%7 51.3419554
%8 49.42824464
%9 54.7891688444444
%10 61.3625101888889
%11 71.1867454
%12 69.6829251619047
%13 68.6815992416667
%14 76.1306230833333
%15 81.0703764416667
%16 78.8460219925926
%17 77.1943228066666
%18 81.2171535466667
%};
%\addlegendentry{Unweighted}
\addplot [line width=1.08pt, color0, dash pattern=on 4pt off 1pt on 4pt off 1pt on 1pt off 1pt]
table {%
0 0.7321472
1 1.1099724
2 0.896994266666666
3 1.27613036666667
4 20.9647478
5 36.1308948166667
6 41.3185177666667
7 51.3419554
8 49.42824464
9 54.7891688444444
10 61.3625101888889
11 71.1867454
12 69.6829251619047
13 68.6815992416667
14 76.1306230833333
15 81.0703764416667
16 78.8460219925926
17 77.1943228066666
18 81.2171535466667
};
\addlegendentry{Unweighted}
\addplot [line width=1.08pt, color1, dash pattern=on 5pt off 1pt on 1pt off 1pt]
table {%
0 0.7321472
1 1.1099724
2 0.896994266666666
3 1.27613036666667
4 20.9963733333333
5 36.1308948166667
6 41.3185177666667
7 51.4251291006862
8 41.5596416671961
9 54.7783904
10 61.3513278222222
11 67.6533975036561
12 69.6692018666667
13 68.67084815
14 76.108516725
15 81.045033375
16 78.8147609481482
17 77.1393907866667
18 81.1396698333333
};
\addlegendentry{Weighted}

\nextgroupplot[
scaled y ticks=manual:{}{\pgfmathparse{#1}},
tick align=outside,
tick pos=left,
title={\gls{all}},
x grid style={white!69.0196078431373!black},
xlabel={Number of drivers},
xmin=-0.5, xmax=18.5,
xtick style={color=black},
xtick={0,1,2,3,4,5,6,7,8,9,10,11,12,13,14,15,16,17,18},
xticklabels={4,,8,,12,,16,,20,,24,,28,,32,,36,,40},
y grid style={white!69.0196078431373!black},
ymin=-5.53197659367708, ymax=125.917193133885,
ytick style={color=black},
yticklabels={}
]
%\addplot [draw=color0, fill=color0, mark=*, only marks]
%table{%
%x  y
%0 0.7321472
%1 1.1099724
%2 0.818352622222222
%3 1.40610055555556
%4 16.0151758166667
%5 29.1873700133333
%6 41.28381756
%7 47.8856537111111
%8 47.0987898
%9 55.7176286190476
%10 61.261450775
%11 67.8338122592593
%12 65.6091891111111
%13 65.2841193333333
%14 73.9018553515152
%15 81.0227477030303
%16 77.7260097111111
%17 78.0346347025641
%18 79.6571154820513
%};
\addplot [line width=1.08pt, color0, dash pattern=on 4pt off 1pt on 4pt off 1pt on 1pt off 1pt, forget plot]
table {%
0 0.7321472
1 1.1099724
2 0.818352622222222
3 1.40610055555556
4 16.0151758166667
5 29.1873700133333
6 41.28381756
7 47.8856537111111
8 47.0987898
9 55.7176286190476
10 61.261450775
11 67.8338122592593
12 65.6091891111111
13 65.2841193333333
14 73.9018553515152
15 81.0227477030303
16 77.7260097111111
17 78.0346347025641
18 79.6571154820513
};
\addplot [line width=1.08pt, color1, dash pattern=on 5pt off 1pt on 1pt off 1pt]
table {%
0 0.7321472
1 1.1099724
2 0.818352622222222
3 1.40610055555556
4 11.1732574333333
5 29.1873700133333
6 41.24520072
7 48.0293684287787
8 47.2431203384249
9 55.708389952381
10 61.253064
11 67.6035513414044
12 65.598515437037
13 59.2737077513906
14 70.3359040102003
15 81.0043163818182
16 77.7025639277778
17 77.9923793025641
18 79.597512625641
};

\nextgroupplot[
scaled y ticks=manual:{}{\pgfmathparse{#1}},
tick align=outside,
tick pos=left,
title={\gls{small}},
x grid style={white!69.0196078431373!black},
xlabel={Number of drivers},
xmin=-0.5, xmax=18.5,
xtick style={color=black},
xtick={0,1,2,3,4,5,6,7,8,9,10,11,12,13,14,15,16,17,18},
xticklabels={4,,8,,12,,16,,20,,24,,28,,32,,36,,40},
y grid style={white!69.0196078431373!black},
ymin=-5.53197659367708, ymax=125.917193133885,
ytick style={color=black},
yticklabels={}
]
%\addplot [draw=color0, fill=color0, mark=*, only marks]
%table{%
%x  y
%0 0.612590533333333
%1 1.1099724
%2 0.818352622222222
%3 1.37372231666667
%4 13.08731748
%5 24.5147362555555
%6 37.9563628333333
%7 46.9203734857143
%8 43.8987625666667
%9 52.5364476518518
%10 57.3344622733333
%11 65.2275541266667
%12 63.0675412242424
%13 64.6235315944444
%14 71.9627909025641
%15 76.8224715619048
%16 76.8188269142857
%17 78.4105219155555
%18 77.4021942541667
%};
\addplot [line width=1.08pt, color0, dash pattern=on 4pt off 1pt on 4pt off 1pt on 1pt off 1pt, forget plot]
table {%
0 0.612590533333333
1 1.1099724
2 0.818352622222222
3 1.37372231666667
4 13.08731748
5 24.5147362555555
6 37.9563628333333
7 46.9203734857143
8 43.8987625666667
9 52.5364476518518
10 57.3344622733333
11 65.2275541266667
12 63.0675412242424
13 64.6235315944444
14 71.9627909025641
15 76.8224715619048
16 76.8188269142857
17 78.4105219155555
18 77.4021942541667
};
\addplot [line width=1.08pt, color1, dash pattern=on 5pt off 1pt on 1pt off 1pt]
table {%
0 0.612590533333333
1 1.1099724
2 0.818352622222222
3 1.37372231666667
4 9.21899804197242
5 21.2521931604711
6 37.9241821333333
7 47.0425514157101
8 43.808203925
9 52.5292620222222
10 57.3277528533333
11 65.2141115266667
12 61.1220952352994
13 61.2583933272106
14 71.9491869897436
15 76.8079898095238
16 76.7987305285714
17 78.3739005688889
18 77.3537669333333
};
\end{groupplot}

\end{tikzpicture}}	\\
		\small (a) Platform A\\
		\scalebox{.65}{% This file was created with tikzplotlib v0.9.15.
\begin{tikzpicture} %[font=\small]

\definecolor{color0}{rgb}{0.12156862745098,0.466666666666667,0.705882352941177}
\definecolor{color1}{rgb}{1,0.498039215686275,0.0549019607843137}

\begin{groupplot}[group style={group size=3 by 1}]
\nextgroupplot[
legend cell align={left},
legend style={
  fill opacity=0.8,
  draw opacity=1,
  text opacity=1,
  at={(0.03,0.97)},
  anchor=north west,
  draw=white!80!black
},
tick align=outside,
tick pos=left,
title={\gls{big}},
x grid style={white!69.0196078431373!black},
xlabel={Number of drivers},
xmin=-0.5, xmax=18.5,
xtick style={color=black},
xtick={0,1,2,3,4,5,6,7,8,9,10,11,12,13,14,15,16,17,18},
xticklabels={4,,8,,12,,16,,20,,24,,28,,32,,36,,40},
y grid style={white!69.0196078431373!black},
ylabel={Normalized player payoff},
ymin=-5.66085383866667, ymax=125.922007945333,
ytick style={color=black}
]
%\addplot [draw=color0, fill=color0, mark=*, only marks]
%table{%
%x  y
%0 0.551516266666667
%1 0.996431
%2 0.7979494
%3 1.89429103333333
%4 21.2128951777778
%5 36.14066765
%6 41.4688884
%7 51.1456825166667
%8 45.4021469333333
%9 57.9844299777778
%10 64.5714174666667
%11 67.9657574444444
%12 75.3759129619047
%13 76.0522283833333
%14 78.5166229583333
%15 81.0612018333333
%16 81.0662501629629
%17 81.1072943333333
%18 75.28503576
%};
%\addlegendentry{Unweighted}
\addplot [line width=1.08pt, color0, dash pattern=on 4pt off 1pt on 4pt off 1pt on 1pt off 1pt]
table {%
0 0.551516266666667
1 0.996431
2 0.7979494
3 1.89429103333333
4 21.2128951777778
5 36.14066765
6 41.4688884
7 51.1456825166667
8 45.4021469333333
9 57.9844299777778
10 64.5714174666667
11 67.9657574444444
12 75.3759129619047
13 76.0522283833333
14 78.5166229583333
15 81.0612018333333
16 81.0662501629629
17 81.1072943333333
18 75.28503576
};
\addlegendentry{Unweighted}
\addplot [line width=1.08pt, color1, dash pattern=on 5pt off 1pt on 1pt off 1pt]
table {%
0 0.551516266666667
1 0.996431
2 0.7979494
3 1.89429103333333
4 21.4623053321986
5 36.14066765
6 41.4688884
7 46.0997573166667
8 44.2324004510555
9 57.9736515333333
10 64.5602351
11 68.0037108946295
12 75.3621896666667
13 76.0414772916667
14 78.4945166
15 81.0358587666667
16 81.0349891185185
17 81.0523623133333
18 75.2075520466667
};
\addlegendentry{Weighted}

\nextgroupplot[
scaled y ticks=manual:{}{\pgfmathparse{#1}},
tick align=outside,
tick pos=left,
title={\gls{all}},
x grid style={white!69.0196078431373!black},
xlabel={Number of drivers},
xmin=-0.5, xmax=18.5,
xtick style={color=black},
xtick={0,1,2,3,4,5,6,7,8,9,10,11,12,13,14,15,16,17,18},
xticklabels={4,,8,,12,,16,,20,,24,,28,,32,,36,,40},
y grid style={white!69.0196078431373!black},
ymin=-5.66085383866667, ymax=125.922007945333,
ytick style={color=black},
yticklabels={}
]
%\addplot [draw=color0, fill=color0, mark=*, only marks]
%table{%
%x  y
%0 0.551516266666667
%1 0.996431
%2 0.700842044444444
%3 1.36413966666667
%4 16.1259928166667
%5 29.13532228
%6 37.3647965466667
%7 44.6579359666667
%8 44.3521089333333
%9 55.6489202761905
%10 61.3507153583333
%11 65.686526837037
%12 70.0716341111111
%13 75.04432186
%14 73.7913898121212
%15 81.0164631212121
%16 81.0019620722222
%17 79.5935132307692
%18 76.5636897641025
%};
\addplot [line width=1.08pt, color0, dash pattern=on 4pt off 1pt on 4pt off 1pt on 1pt off 1pt]
table {%
0 0.551516266666667
1 0.996431
2 0.700842044444444
3 1.36413966666667
4 16.1259928166667
5 29.13532228
6 37.3647965466667
7 44.6579359666667
8 44.3521089333333
9 55.6489202761905
10 61.3507153583333
11 65.686526837037
12 70.0716341111111
13 75.04432186
14 73.7913898121212
15 81.0164631212121
16 81.0019620722222
17 79.5935132307692
18 76.5636897641025
};
\addplot [line width=1.08pt, color1, dash pattern=on 5pt off 1pt on 1pt off 1pt]
table {%
0 0.551516266666667
1 0.996431
2 0.700842044444444
3 1.36413966666667
4 11.2909395079073
5 29.13532228
6 37.3261797066667
7 44.8016506843343
8 41.3909294476191
9 55.6396816095238
10 61.3423285833333
11 65.6777882682368
12 70.060960437037
13 75.0386730966759
14 66.5829584148879
15 80.9980318
16 80.9785162888889
17 79.5512578307692
18 76.5040869076923
};

\nextgroupplot[
scaled y ticks=manual:{}{\pgfmathparse{#1}},
tick align=outside,
tick pos=left,
title={\gls{small}},
x grid style={white!69.0196078431373!black},
xlabel={Number of drivers},
xmin=-0.5, xmax=18.5,
xtick style={color=black},
xtick={0,1,2,3,4,5,6,7,8,9,10,11,12,13,14,15,16,17,18},
xticklabels={4,,8,,12,,16,,20,,24,,28,,32,,36,,40},
y grid style={white!69.0196078431373!black},
ymin=-5.66085383866667, ymax=125.922007945333,
ytick style={color=black},
yticklabels={}
]
%\addplot [draw=color0, fill=color0, mark=*, only marks]
%table{%
%x  y
%0 0.784050733333333
%1 0.996431
%2 0.700842044444444
%3 1.31859135
%4 12.94571956
%5 24.5078795555556
%6 37.9918066666667
%7 44.1432013428571
%8 43.9997719666667
%9 50.2924709925926
%10 57.4182408666667
%11 65.26401628
%12 66.6560312060606
%13 71.1333508666667
%14 70.2949041589744
%15 76.8426525714286
%16 79.5989454142857
%17 79.7737444888889
%18 79.9096711125
%};
\addplot [line width=1.08pt, color0, dash pattern=on 4pt off 1pt on 4pt off 1pt on 1pt off 1pt, forget plot]
table {%
0 0.784050733333333
1 0.996431
2 0.700842044444444
3 1.31859135
4 12.94571956
5 24.5078795555556
6 37.9918066666667
7 44.1432013428571
8 43.9997719666667
9 50.2924709925926
10 57.4182408666667
11 65.26401628
12 66.6560312060606
13 71.1333508666667
14 70.2949041589744
15 76.8426525714286
16 79.5989454142857
17 79.7737444888889
18 79.9096711125
};
\addplot [line width=1.08pt, color1, dash pattern=on 5pt off 1pt on 1pt off 1pt]
table {%
0 0.784050733333333
1 0.996431
2 0.700842044444444
3 1.31859135
4 9.06157937333334
5 24.5765098518653
6 37.9596259666667
7 44.0691586380952
8 44.0154425959814
9 50.285285362963
10 57.4115314466667
11 65.25057368
12 66.6980631365157
13 71.2166430744694
14 70.2813002461538
15 76.8281708190476
16 79.5788490285714
17 79.7371231422222
18 79.8612437916667
};
\end{groupplot}

\end{tikzpicture}}	 \\
		\small(b) Platform B \\
		\scalebox{.65}{% This file was created with tikzplotlib v0.9.15.
\begin{tikzpicture} %[font=\small]

\definecolor{color0}{rgb}{0.12156862745098,0.466666666666667,0.705882352941177}
\definecolor{color1}{rgb}{1,0.498039215686275,0.0549019607843137}

\begin{groupplot}[group style={group size=3 by 1}]
\nextgroupplot[
legend cell align={left},
legend style={
  fill opacity=0.8,
  draw opacity=1,
  text opacity=1,
  at={(0.03,0.97)},
  anchor=north west,
  draw=white!80!black
},
tick align=outside,
tick pos=left,
title={\gls{big}},
x grid style={white!69.0196078431373!black},
xlabel={Number of drivers},
xmin=-0.5, xmax=18.5,
xtick style={color=black},
xtick={0,1,2,3,4,5,6,7,8,9,10,11,12,13,14,15,16,17,18},
xticklabels={4,,8,,12,,16,,20,,24,,28,,32,,36,,40},
y grid style={white!69.0196078431373!black},
ylabel={Normalized player payoff},
ymin=-6.0124602964956, ymax=126.261666226408,
ytick style={color=black}
]
%\addplot [draw=color0, fill=color0, mark=*, only marks]
%table{%
%x  y
%0 0.784050733333333
%1 0.897663066666667
%2 0.8037162
%3 1.14339627777778
%4 10.9041634555556
%5 24.3328258444444
%6 28.7427520166667
%7 35.23599616
%8 39.40061816
%9 49.4284243333333
%10 53.1260711888889
%11 57.0046407333333
%12 61.2230837952381
%13 66.9733598095238
%14 63.7408455416667
%15 73.3540911037037
%16 71.2243765074074
%17 76.7307000333333
%18 75.2090795133333
%};
%\addlegendentry{Unweighted}
\addplot [line width=1.08pt, color0, dash pattern=on 4pt off 1pt on 4pt off 1pt on 1pt off 1pt]
table {%
0 0.784050733333333
1 0.897663066666667
2 0.8037162
3 1.14339627777778
4 10.9041634555556
5 24.3328258444444
6 28.7427520166667
7 35.23599616
8 39.40061816
9 49.4284243333333
10 53.1260711888889
11 57.0046407333333
12 61.2230837952381
13 66.9733598095238
14 63.7408455416667
15 73.3540911037037
16 71.2243765074074
17 76.7307000333333
18 75.2090795133333
};
\addlegendentry{Unweighted}
\addplot [line width=1.08pt, color1, dash pattern=on 5pt off 1pt on 1pt off 1pt]
table {%
0 0.784050733333333
1 0.897663066666667
2 0.8037162
3 1.14339627777778
4 7.53835266666667
5 24.3328258444444
6 28.7427520166667
7 35.2692656402745
8 31.4697310968757
9 49.4219572666667
10 53.1204800055556
11 57.0209064976984
12 61.216222147619
13 66.9672163285714
14 63.7297923625
15 73.3428275185185
16 71.2087459851852
17 76.7001822444444
18 75.1703376566667
};
\addlegendentry{Weighted}

\nextgroupplot[
scaled y ticks=manual:{}{\pgfmathparse{#1}},
tick align=outside,
tick pos=left,
title={\gls{all}},
x grid style={white!69.0196078431373!black},
xlabel={Number of drivers},
xmin=-0.5, xmax=18.5,
xtick style={color=black},
xtick={0,1,2,3,4,5,6,7,8,9,10,11,12,13,14,15,16,17,18},
xticklabels={4,,8,,12,,16,,20,,24,,28,,32,,36,,40},
y grid style={white!69.0196078431373!black},
ymin=-6.0124602964956, ymax=126.261666226408,
ytick style={color=black},
yticklabels={}
]
%\addplot [draw=color0, fill=color0, mark=*, only marks]
%table{%
%x  y
%0 0.784050733333333
%1 0.897663066666667
%2 1.04828973333333
%3 1.42522831666667
%4 15.89511605
%5 35.9429059166667
%6 37.9501455
%7 47.9683531333333
%8 44.8244633777778
%9 51.4577123666667
%10 61.376646175
%11 66.2786445166667
%12 69.1281518666667
%13 75.0700242733333
%14 73.0997126666667
%15 79.4376871666667
%16 76.1540118666667
%17 81.1324603611111
%18 79.7629649047619
%};
\addplot [line width=1.08pt, color0, dash pattern=on 4pt off 1pt on 4pt off 1pt on 1pt off 1pt, forget plot]
table {%
0 0.784050733333333
1 0.897663066666667
2 1.04828973333333
3 1.42522831666667
4 15.89511605
5 35.9429059166667
6 37.9501455
7 47.9683531333333
8 44.8244633777778
9 51.4577123666667
10 61.376646175
11 66.2786445166667
12 69.1281518666667
13 75.0700242733333
14 73.0997126666667
15 79.4376871666667
16 76.1540118666667
17 81.1324603611111
18 79.7629649047619
};
\addplot [line width=1.08pt, color1, dash pattern=on 5pt off 1pt on 1pt off 1pt]
table {%
0 0.784050733333333
1 0.897663066666667
2 1.04828973333333
3 1.42522831666667
4 11.0600627412406
5 35.9429059166667
6 37.9179648
7 44.5198655555556
8 44.9928490059401
9 51.4496285333333
10 61.3682594
11 66.2688136267664
12 69.11854556
13 75.2021024299759
14 73.415933657887
15 79.4207917888889
16 76.1305660833333
17 81.0866836777778
18 79.7076193952381
};

\nextgroupplot[
scaled y ticks=manual:{}{\pgfmathparse{#1}},
tick align=outside,
tick pos=left,
title={\gls{small}},
x grid style={white!69.0196078431373!black},
xlabel={Number of drivers},
xmin=-0.5, xmax=18.5,
xtick style={color=black},
xtick={0,1,2,3,4,5,6,7,8,9,10,11,12,13,14,15,16,17,18},
xticklabels={4,,8,,12,,16,,20,,24,,28,,32,,36,,40},
y grid style={white!69.0196078431373!black},
ymin=-6.0124602964956, ymax=126.261666226408,
ytick style={color=black},
yticklabels={}
]
%\addplot [draw=color0, fill=color0, mark=*, only marks]
%table{%
%x  y
%0 0
%1 0.897663066666667
%2 1.04828973333333
%3 1.44488476666667
%4 30.7578286
%5 40.9721491
%6 41.36949665
%7 51.2074634666667
%8 46.4995711333333
%9 56.5104200833333
%10 66.3930335
%11 67.9097781555555
%12 77.7233472
%13 80.9920521888889
%14 74.4539660666667
%15 81.0774710444444
%16 73.7326306416666
%17 81.19604725
%18 78.75836595
%};
\addplot [line width=1.08pt, color0, dash pattern=on 4pt off 1pt on 4pt off 1pt on 1pt off 1pt, forget plot]
table {%
0 0
1 0.897663066666667
2 1.04828973333333
3 1.44488476666667
4 30.7578286
5 40.9721491
6 41.36949665
7 51.2074634666667
8 46.4995711333333
9 56.5104200833333
10 66.3930335
11 67.9097781555555
12 77.7233472
13 80.9920521888889
14 74.4539660666667
15 81.0774710444444
16 73.7326306416666
17 81.19604725
18 78.75836595
};
\addplot [line width=1.08pt, color1, dash pattern=on 5pt off 1pt on 1pt off 1pt]
table {%
0 0
1 0.897663066666667
2 1.04828973333333
3 1.44488476666667
4 21.087030004931
5 31.3175803126398
6 41.3212256
7 46.2142522166667
8 46.5309123919629
9 56.4942524166667
10 66.37625995
11 67.8873738222222
12 77.7279588922337
13 81.1204503316226
14 74.4244909222222
15 81.0436802888889
16 73.6974619666667
17 81.127382225
18 78.6615113083333
};
\end{groupplot}

\end{tikzpicture}}	 \\
		\small(c) Platform C \\
	\end{tabular}
	\caption{Impact of weighted \gls{vcg} on each platform's payoff in an online setting}
	\label{fig:weights}
\end{figure}

Figures~\ref{fig:weights}a-\ref{fig:weights}c show the averaged payoff per platform per scenario, depending on the total number of drivers. Payoffs are shown for both weighted and unweighted \gls{vcg}. 	For instances that could not solve the weights optimization, unweighted payoffs are shown. 
%For instances that could not solve the weights optimization, unweighted payoffs are shown.

\paragraph{Online averaged platform payoffs:}

\begin{figure}[btp]
	%		\scalebox{.4}{\includegraphics{chapters/assets/png/3_players/Scenario_impact_utility_player_VCG}}
	\scalebox{.64}{\input{chapters/figures/Scenario_impact_utility_player_VCG_online}}
	\caption{Impact of the distribution scenario on each platform's payoff in an online setting}
	\label{fig:scenario_impact_on}
\end{figure}

Similar to Figure~\ref{fig:scenario_impact}, Figure~\ref{fig:scenario_impact_on} compares the normalized payoffs obtained for each platform depending on the number of drivers, for each distribution scenario in an online setting. 
Analogous to offline results, online results show that a platform with a lower share of drivers (platform B in \gls{small}) obtains a higher payoff than a platform with a higher share of managed drivers.
Compared to the offline setting, platforms' payoffs are slightly higher due to the approximation induced by the online allocation policy compared to a perfect-information allocation (cf. Figure~\ref{fig:scenario_impact}).

\end{appendices}

\end{document}